\documentclass[12pt]{article}

\usepackage{amsmath,amsthm,amssymb,mathrsfs,bbm}
\usepackage[toc, page]{appendix}
\usepackage{newpxtext,newpxmath}
\usepackage{geometry,pdflscape,setspace}
\usepackage{graphicx}
\usepackage{fancyhdr}
\usepackage{authblk}
\usepackage[normalem]{ulem}

\usepackage{subcaption}

\usepackage{tikz}
\usetikzlibrary{plotmarks,arrows,calc,patterns}

\usepackage[format=hang, textfont={md}, labelfont={bf}, justification=raggedright]{caption}
\usepackage{natbib,threeparttable,dcolumn,multicol,lscape,bigstrut,booktabs,multirow,colortbl,array}
\usepackage[colorlinks=true,citecolor=blue,linkcolor=blue,%
    anchorcolor=black,urlcolor=blue,pdfborder={0 0 0},pdfstartview=FitH]{hyperref}

\graphicspath{{./Figure/}}

\theoremstyle{plain}
\newtheorem{theorem}{Theorem}
\newtheorem{lemma}{Lemma}
\newtheorem{proposition}{Proposition}
\newtheorem{corollary}{Corollary}

\theoremstyle{definition}
\newtheorem{assumption}{Assumption}
\newtheorem{definition}{Definition}

\theoremstyle{remark}
\newtheorem{remark}{Remark}
\newtheorem{example}{Example}

\newcommand{\indep}{\perp \!\!\! \perp}

\begin{document}
\title{Robust Identification in Randomized Experiments with Noncompliance}
\author{
D\'{e}sir\'{e} K\'{e}dagni \thanks{Corresponding author, email address: \href{mailto: dkedagni@unc.edu}{dkedagni@unc.edu}}, \ \
Huan Wu \thanks{Email address: \href{mailto: huan.wu@unc.edu}{huan.wu@unc.edu}}, \ \ and \ 
Yi Cui \thanks{Email address: \href{mailto: yicui@unc.edu}{yicui@unc.edu}}
,\ \
\\
\textit{University of North Carolina at Chapel Hill}
}
\date{This version:\thanks{We are grateful to Felipe Goncalves, Ismael Mourifi\'e, Claudia Noack, Vira Semenova, and Kaspar Wuthrich for their helpful comments. We are also grateful to participants at the UNC-Chapel Hill Econometrics Workshop, and the 2024 Midwest Econometrics Group conference. First version: August 6, 2024.} \ \today.
}
\maketitle
\onehalfspacing

\begin{abstract}
Instrument variable (IV) methods are widely used in empirical research to identify causal effects of a policy. In the local average treatment effect (LATE) framework, the IV estimand identifies the LATE under three main assumptions: random assignment, exclusion restriction, and monotonicity. However, these assumptions are often questionable in many applications, leading some researchers to doubt the causal interpretation of the IV estimand. This paper considers a robust identification of causal parameters in a randomized experiment setting with noncompliance where the standard LATE assumptions could be violated. 
We discuss identification under two sets of weaker assumptions: random assignment and exclusion restriction (without monotonicity), and random assignment and monotonicity (without exclusion restriction). 
We derive sharp bounds on some causal parameters under these two sets of relaxed LATE assumptions. Finally, we apply our method to revisit the random information experiment conducted in \cite{bursztyn2020misperceived} and find that the standard LATE assumptions are jointly incompatible in this application. We then estimate the robust identified sets under the two sets of relaxed assumptions.


    
\end{abstract}

 \maketitle
{\footnotesize \textbf{Keywords}: Instrumental variable, monotonicity, exclusion restriction, local average controlled direct effect.  

}

\clearpage

\section{Introduction}
\label{sec:intro}

Researchers are often concerned about the credibility of their assumptions and the conclusions they reach \citep[see \textit{The Law of Decreasing Credibility} in][]{manski2003partial, manski2011policy}. When the stated assumptions are refuted by the data, analysts usually resort to weaker versions of the assumptions. There may exist multiple ways of relaxing a set of assumptions when they are rejected by the data. Some researchers may choose to relax the assumptions in a continuous way \citep[as in][]{Masten2021SalvagingModels}, while others may consider discrete relaxations by dropping some of the assumptions \citep[as in][]{li2024discordant}. In this paper, we consider a robust identification of causal parameters in a randomized experiment setting with noncompliance where the standard local average treatment effect (LATE) assumptions could be violated. We follow \citeauthor{li2024discordant}'s (\citeyear{li2024discordant}) approach to propose a misspecification robust bound for a vector-valued parameter of various causal parameters. 

Three main assumptions (random assignment, exclusion restriction, and monotonicity) are commonly used in the LATE framework. When these assumptions are jointly rejected, researchers could choose to relax some of the assumptions if they have some additional information on which assumptions may be causing the rejection. For example, \cite{Kedagni2023IdentifyingTypes} proposes an identification strategy that allows for violations of the random assignment assumption, \cite{deChaisemartin2017ToleratingMonotonicity}, \cite{Huber2017SharpNoncompliance}, \cite{small2017instrumental}, \cite{Noack2021SensitivityAssumption}, \cite{dahl2023never}, among others have proposed different ways of relaxing the monotonicity assumption, while \cite{cai2008bounds} and \cite{flores2013partial} allow for violations of the exclusion restriction assumption. But these relaxations may appear arbitrary in some circumstances unless the analyst has an unambiguous reason to convince the scientific community that she knows exactly which assumptions are the source of the rejection in the data. 

The current paper aims at developing a unifying approach that would provide the maximal combination of the three assumptions that are compatible with the data. We consider a randomized experiment setting where the random assignment assumption is maintained. We discuss identification under two sets of weaker assumptions: random assignment and exclusion restriction (without monotonicity), and random assignment and monotonicity (without exclusion restriction). 

We first introduce two causal parameters: the local average treatment-controlled direct effect (LATCDE), and the local average instrument-controlled direct effect (LAICDE). We define a twenty-dimensional real-valued vector parameter of interest that contains the eight LATCDEs, the eight LAICDEs, and the four types' probabilities. Second, we derive sharp bounds for the twenty parameters under respectively (i) random assignment, exclusion restriction, and monotonicity; (ii) random assignment, and exclusion restriction; and (iii) random assignment, and monotonicity. Specifically, under the random assignment and monotonicity assumptions, we derive sharp bounds on the local average treatment-controlled direct effects for the always-takers and never-takers, respectively, and the total average controlled direct effect for the compliers. Additionally, we show that the intent-to-treat effect can be expressed as a convex weighted average of these three effects. Third, we propose a misspecification robust bound for the vector parameter of interest under the three possible sets of assumptions. Finally, we apply our method to analyze the randomized information experiment conducted by \cite{bursztyn2020misperceived}. We find that the LATE assumptions are jointly rejected in this setting. However, when either the monotonicity or exclusion restriction is dropped, the remaining two assumptions cannot be rejected. We provide estimates of identified sets for our parameters of interest, which are robust to violations of monotonicity or exclusion assumptions.




The remainder of the paper is organized as follows. Section \ref{analy} presents the model, the assumptions, and the parameters of interest. In Sections \ref{sec:late:ass}, \ref{sec:ra_er:ass}, and \ref{sec:ra_mon:ass}, we derive bounds on the parameters under respectively (i) random assignment, exclusion restriction, and monotonicity; (ii) random assignment, and exclusion restriction; and (iii) random assignment, and monotonicity. Section \ref{sec:mrb} presents the misspecification robust bound, Section \ref{sec:emp} shows the empirical illustration, and finally Section \ref{sec:summary} concludes. The proofs of the main results are relagated to the appendix from the supplementary material.    

\section{Analytical framework}\label{analy}
We consider a potential outcome model with a binary treatment and a binary instrumental variable. Let $D \in \{0, 1\}$ be the observed treatment, where $D = 1$ indicates receiving the treatment and $D = 0$ not receiving it. Let $Y$ denote the observed outcome, and $Z$ denote a binary instrumental variable, which takes values in $\mathcal{Z} = \{0, 1\}$. Following the convention, we use $D_0$ and $D_1$ to represent the potential treatment status, where $D_z$ denotes the potential treatment if the instrument $Z$ is externally set to $z$. Also, $Y_{dz}$ represents the potential outcome if the treatment and instrument are externally set to $d$ and $z$, respectively. Then, we have the following potential outcome model,
\begin{equation} \label{eq:basic_model}
    \left\{\begin{array}{l}
        Y = \left(Y_{11}Z + Y_{10}(1 - Z)\right) D + \left(Y_{01}Z + Y_{00}(1 - Z)\right) (1 - D), \\
        D = D_1 Z + D_0 (1 - Z),
        \end{array}\right.
\end{equation}
where $(Y, D, Z)$ is the observed data, $\left(Y_{00}, Y_{01}, Y_{10}, Y_{11}, D_0, D_1\right)$ is a vector of latent variables. 

The vector $\left(D_0, D_1\right)$ represents the four unobserved groups commonly known in the treatment effect literature as the \textit{types} in the language of \cite{Angrist1996IdentificationVariables}. We have
$$
\begin{array}{ll}
\left(D_0=1, D_1=1\right): \text { always-takers }(a), & \left(D_0=0, D_1=0\right): \text { never-takers }(n), \\
\left(D_0=0, D_1=1\right): \text { compliers }(c), & \left(D_0=1, D_1=0\right): \text { defiers }(df) .
\end{array}
$$
Let $T$ denote the random type of an individual with support $\{a, c, n, d f\}$. Let $p_t$ denote the proportion of type $t$ in the population. Thus, we have the corresponding probabilities for these four types (strata/groups) of the population: $p_n$, $p_{df}$, $p_c$, and $p_a$.
We define our parameters of interest: 
\begin{eqnarray*}
    \theta_{zt} &\equiv& \mathbb E[Y_{1z}-Y_{0z} \mid T=t]\ \ \text{(local average instrument-controlled direct effect, LAICDE)},\\
 \delta_{dt} &\equiv& \mathbb E[Y_{d1}-Y_{d0} \mid T=t]\ \ \text{(local average treatment-controlled direct effect, LATCDE)},
\end{eqnarray*}
for $z \in\{0,1\}$, $d \in\{0,1\},$ and $t\in\{a,c,df,n\}$. We have eight $\theta$-parameters, and eight $\delta$-parameters. Combined with the four types' probabilities $p_t$, we define a $20 \times 1$ column vector $\Gamma$ of our parameters of interest,   
$$\Gamma\equiv [\theta_{0a}, \theta_{1a}, \theta_{0c}, \theta_{1c}, \theta_{0df}, \theta_{1df}, \theta_{0n}, \theta_{1n}, \delta_{0a}, \delta_{1a}, \delta_{0c}, \delta_{1c}, \delta_{0df}, \delta_{1df}, \delta_{0n}, \delta_{1n}, p_a, p_c, p_{df}, p_n]'.$$
Elements of the vector $\Gamma$ can be viewed as building block parameters for many commonly used causal parameters. For instance, we can define some commonly used parameters as functionals of $\Gamma$:
\begin{eqnarray*}
    \theta_t &\equiv& \sum_{z\in \{0,1\}} \mathbb P(Z=z) \theta_{zt} \ \ \text{(local average treatment effect for type $t$, $LATE_t$)},\\
   \theta &\equiv& \sum_{t \in \{a,c,df,n\}} p_t \theta_t \ \ \text{(average treatment effect, $ATE$)}, \\
   \delta_d &\equiv& \sum_{t \in \{a,c,df,n\}} p_t \delta_{dt}\ \ \text{(average controlled direct effect, $ACDE$)}.  
\end{eqnarray*}

In this paper, our goal is to provide a robust identified set for these parameters under the following set of assumptions considered in the original work by \cite{imbens1994identification}.
\begin{assumption}[Random assignment, RA] \label{ass:RA}
    The instrument variable $Z$ is independent of potential outcomes and potential treatments:
    \begin{equation*}
        Z \indep (Y_{11}, Y_{10}, Y_{01}, Y_{00}, D_1, D_0).
    \end{equation*}
\end{assumption}
Assumption \ref{ass:RA} requires that the instrument $Z$ be independent of all potential outcomes and potential treatments. This assumption will likely hold in randomized experiments. We may assume that it holds conditional on covariates in stratified randomized experiments. This assumption is maintained throughout the paper. 


\begin{assumption}[Exclusion restriction, ER] \label{ass:er}
    $Y_{dz}=Y_{dz'}=Y_d$ for all $d$, $z$, and $z'$.
\end{assumption}
Assumption \ref{ass:er} imposes that the instrument $Z$ does not have a direct effect on the potential outcome. It is allowed to affect the outcome $Y$ only through the treatment $D$. We allow for potential violations of this assumption in this paper. 

\begin{assumption}[Monotonicity, MON] \label{ass:mon}
    Either $D_1 \geq D_0$ or $D_0 \geq D_1$.
\end{assumption}
Assumption \ref{ass:mon}, also often referred to as \textit{uniformity} (see \citeauthor{heckman2018unordered}, \citeyear{heckman2018unordered}), rules out the existence of both compliers and defiers in the population. This assumption is questionable in many empirical settings, especially the judge IV design literature. 
To illustrate, suppose there are two judges who are randomly assigned to cases/defendants. A jury will make a decision on whether or not to convict a defendant. A researcher is interested in studying the effect of the jury's decision $D$ on an outcome $Y$ (recidivism/employment). The judges assignment variable $Z$ can be considered as an IV. There is no \textit{a priori} reason to assume monotonicity in the jury's decision with respect to the judges. Hence, all four types are likely present in this setting.  
Therefore, we allow for complete violations of this monotonicity assumption in this paper. 

\cite{Fiorini2021ScrutinizingDesigns} extensively discussed situations in which the monotonicity assumption might be questionable. 
For illustration purposes, we consider the measurement of the returns to schooling as our motivating example. 
\begin{example}[Returns to schooling]
     Suppose that the college decision is based on whether the opportunity cost $V_1$ is less than a threshold $Q_1(Z)$ and the cost of attending college $V_2$ (e.g., psychological cost of effort, tuition) is less than a different threshold $Q_2(Z)$, where $Z$ denotes an IV, say college proximity. This two-dimensional decision model is called the double hurdle model in \cite{Lee2018IdentifyingTreatments}. Let $Y$ be the outcome of interest (wage), $D$ be the treatment variable (college attendance), and $Z$ be the instrument (college proximity) so that the college choice model is given as below: 
\begin{equation}
D= \mathbbm{1}\left\{V_1 < Q_1(Z), V_2 < Q_2(Z) \right \},
\end{equation}
In this framework, the monotonicity assumption will likely fail if for example the instrument lowers $Q_1$ while at the same time it increases $Q_2$, that is, $Q_1(0) < Q_1(1)$, and $Q_2(1) < Q_2(0)$. As a consequence, we will have all four types of individuals in the population: compliers $(c)$, defiers $(df)$, never-takers $(n)$, and always-takers $(a)$ as illustrated in Figure \ref{types}.\footnote{Definition of the four types in the population:

$c=\{V_1 < Q_1(1),V_2 < Q_2(1)\}\cap\{V_1 \geq Q_1(0)\cup V_2 \geq Q_2(0)\},$

$df=\{V_1 < Q_1(0),V_2 < Q_2(0)\}\cap\{V_1 \geq Q_1(1)\cup V_2 \geq Q_2(1)\},$

$a=\{V_1 < Q_1(0),V_2 < Q_2(0)\}\cap\{V_1 < Q_1(1), V_2 < Q_2(1)\},$

$n=\{V_1 \geq Q_1(0)\cup V_2 \geq Q_2(0)\}\cap\{V_1 \geq Q_1(1)\cup V_2 \geq Q_2(1)\}$.}

\begin{figure}   
    \centering
\begin{tikzpicture}[x=0.75pt,y=0.75pt,yscale=-1,xscale=1]

\draw  (172,267.25) -- (474.5,267.25)(202.25,22.45) -- (202.25,294.45) (467.5,262.25) -- (474.5,267.25) -- (467.5,272.25) (197.25,29.45) -- (202.25,22.45) -- (207.25,29.45)  ;
\draw    (202,61) -- (427.5,60.45) ;
\draw    (427.5,60.45) -- (426.5,267.45) ;
\draw [line width=2.25]  [dash pattern={on 2.53pt off 3.02pt}]  (203,104) -- (253.5,103.45) ;
\draw [line width=2.25]  [dash pattern={on 2.53pt off 3.02pt}]  (253.5,103.45) -- (252.5,266.45) ;
\draw [line width=2.25]  [dash pattern={on 2.53pt off 3.02pt}]  (202.5,217.45) -- (364.5,217.45) ;
\draw [line width=2.25]  [dash pattern={on 2.53pt off 3.02pt}]  (364.5,217.45) -- (364.5,267.45) ;

\draw (164,28) node [anchor=north west][inner sep=0.75pt]   [align=left] {$\displaystyle V_{2}$};
\draw (486,272) node [anchor=north west][inner sep=0.75pt]   [align=left] {$\displaystyle V_{1}$};
\draw (324,141) node [anchor=north west][inner sep=0.75pt]   [align=left] {$\displaystyle n$};
\draw (217,150) node [anchor=north west][inner sep=0.75pt]   [align=left] {$\displaystyle \textcolor[rgb]{0.82,0.01,0.11}{df}$};
\draw (303,232) node [anchor=north west][inner sep=0.75pt]   [align=left] {\textcolor[rgb]{0.29,0.56,0.89}{$\displaystyle c$}};
\draw (224,232) node [anchor=north west][inner sep=0.75pt]   [align=left] {\textcolor[rgb]{0.49,0.83,0.13}{$\displaystyle a$}};
\draw (187,272) node [anchor=north west][inner sep=0.75pt]   [align=left] {$\displaystyle 0$};
\draw (181,51) node [anchor=north west][inner sep=0.75pt]   [align=left] {$\displaystyle 1$};
\draw (422,278) node [anchor=north west][inner sep=0.75pt]   [align=left] {$\displaystyle 1$};
\draw (156,96) node [anchor=north west][inner sep=0.75pt]   [align=left] {$\displaystyle Q_{2}( 0)$};
\draw (157,210) node [anchor=north west][inner sep=0.75pt]   [align=left] {$\displaystyle Q_{2}( 1)$};
\draw (228,276) node [anchor=north west][inner sep=0.75pt]   [align=left] {$\displaystyle Q_{1}( 0)$};
\draw (343,276) node [anchor=north west][inner sep=0.75pt]   [align=left] {$\displaystyle Q_{1}( 1)$};

\end{tikzpicture}
    \caption{Visualization of types (double hurdle model)}
    \label{types}
\end{figure}
\end{example}

In addition, growing up near a college could have a direct effect on someone's ability or skills necessary to complete some tasks that they would not have done otherwise. These skills would positively affect people's wages, which would make the college proximity instrument violate the exclusion restriction assumption.

We consider three different combinations (or menus) of assumptions.
\begin{eqnarray*}
    A_1&\equiv&\{RA, ER, MON\},\\
    A_2&\equiv& \{RA,ER\},\\
    A_3&\equiv& \{RA,MON\}.
\end{eqnarray*}
We denote $A$ the set of all three combinations of assumptions: $A\equiv \{A_1,A_2,A_3\}.$ We assume that we are in a randomized experiment setting where Assumption \ref{ass:RA} holds. For this reason, we do not consider the menu $\{ER,MON\}$ as an option. In general, this combination alone would not yield informative bounds without additional assumptions \citep[see][for further details]{Kedagni2023IdentifyingTypes}.

\section{Identification under random assignment, monotonicity, and exclusion restriction}\label{sec:late:ass}
Let $\mathcal{Y}$ denote the support of the outcome $Y$. We first note that under Assumption \ref{ass:er}, $Y_{dz}=Y_{d}$ for all $d$ and $z$. Under Assumptions \ref{ass:RA} and \ref{ass:er}, we have for any Borel set $A\subseteq \mathcal{Y}$, 
\begin{eqnarray}
\label{important3} 
\mathbb{P}(Y \in A, D=1  \vert  Z=1) &=& p_c \mathbb{P}(Y_1 \in A  \vert  T=c)+p_a \mathbb{P}(Y_1 \in A  \vert  T=a), \\
\label{important4}
\mathbb{P}(Y \in A, D=1  \vert  Z=0) &=& p_{df} \mathbb{P}(Y_1 \in A  \vert  T=df)+p_a \mathbb{P}(Y_1 \in A  \vert  T=a),\\
\label{important}
\mathbb{P}(Y \in A, D=0  \vert  Z=1) &=& p_{d f} \mathbb{P}(Y_0 \in A  \vert  T=d f)+p_n \mathbb{P}(Y_0 \in A  \vert  T=n), \\
\label{important2}
\mathbb{P}(Y \in A, D=0  \vert  Z=0) &=& p_c \mathbb{P}(Y_0 \in A  \vert  T=c)+p_n \mathbb{P}(Y_0 \in A  \vert  T=n).
\end{eqnarray}
\equationautorefname~\eqref{important3} shows that the distribution of $Y$ in the treatment group for individuals assigned to this group is a mixture of the distribution of $Y_1$ for the compliers and the always-takers. And \equationautorefname~\eqref{important4} shows that the distribution of $Y$ in the treatment group for individuals assigned to the control group is a mixture of the distribution of $Y_1$ for the defiers and the always-takers. A similar decomposition holds for the control group, as shown in Equations~\eqref{important}-\eqref{important2}.

Let us take the difference between Equations \eqref{important4} and \eqref{important3}. We have
\begin{eqnarray}
&&\mathbb{P}(Y \in A, D=1  \vert  Z=0)-\mathbb{P}(Y \in A, D=1  \vert  Z=1) \label{eq:maintest1}\\
&&\qquad \qquad \qquad \qquad \qquad =p_{df}\mathbb{P}(Y_1 \in A  \vert  T=df) - p_c \mathbb{P}(Y_1 \in A  \vert  T=c).\nonumber
\end{eqnarray}
Similarly, the following holds from Equations \eqref{important} and \eqref{important2}. 
\begin{eqnarray}
&& \mathbb{P}(Y \in A, D=0  \vert  Z=1)-\mathbb{P}(Y \in A, D=0  \vert  Z=0) \label{eq:maintest2}\\
&&\qquad \qquad \qquad \qquad \qquad = p_{df}\mathbb{P}(Y_0 \in A  \vert  T=df) - p_c \mathbb{P}(Y_0 \in A  \vert  T=c).\nonumber 
\end{eqnarray}

For $A=\mathcal{Y}$, Equation \eqref{eq:maintest1} implies
\begin{equation}
\mathbb{P}(D=1  \vert  Z=0)-\mathbb{P}(D=1  \vert  Z=1)=p_{df}-p_c. \label{eq:maintest3}
\end{equation}

Under Assumption \ref{ass:mon}, either $p_{df}=0$ or $p_c=0$. Without loss of generality, assume that $p_{df}=0$. Then, Equation \eqref{eq:maintest3} implies that $p_c=\mathbb{P}(D=1  \vert  Z=1)-\mathbb{P}(D=1  \vert  Z=0)$. As a result, $p_a$ and $p_n$ are also identified: $p_a=\mathbb P(D=1\vert Z=0)$, and $p_n=\mathbb P(D=0\vert Z=1)$. Under Assumption \ref{ass:er}, $\delta_{dt}=0$ for all $d$ and $t$, and $\theta_{0t}=\theta_{1t}$ for all $t$. Under Assumptions~\ref{ass:RA}-\ref{ass:mon}, $\theta_{zc}$ is identified if $\mathbb{P}(D=1  \vert  Z=1)-\mathbb{P}(D=1  \vert  Z=0)>0$ \citep{imbens1994identification}: $$\theta_{1c}=\theta_{0c}=\frac{\mathbb E[Y\vert Z=1]-\mathbb E[Y\vert Z=0]}{\mathbb E[D\vert Z=1]-\mathbb E[D\vert Z=0]},$$
and from Equations \eqref{eq:maintest1}-\eqref{eq:maintest2}, the following testable implications must hold:
\begin{eqnarray}
\mathbb{P}(Y \in A, D=1  \vert  Z=0)-\mathbb{P}(Y \in A, D=1  \vert  Z=1) &\leq& 0, \label{eq:testimpl1}\\
\mathbb{P}(Y \in A, D=0  \vert  Z=1)-\mathbb{P}(Y \in A, D=0  \vert  Z=0) &\leq & 0,\label{eq:testimpl2}
\end{eqnarray}
for all Borel set $A$ \citep{Balke1997BoundsCompliance, Heckman2005StructuralEvaluation, Kitagawa2015AValidity, Mourifie2017TestingAssumptions}. The remaining $\theta$-parameters are not identified and lie in the real line $\mathbb R$. 
The identified set for $\Gamma$ is 
\begin{eqnarray*} 
\Theta_I(A_1)=\left\{\begin{array}{l}
        \Theta_I^1(A_1),\ \text{ if }\ \mathbb{E}[D  \vert  Z=1] - \mathbb{E}[D  \vert  Z=0]> 0,\\
        \Theta_I^2(A_1),\ \text{ if }\ \mathbb{E}[D  \vert  Z=1] - \mathbb{E}[D  \vert  Z=0]< 0,\\
       \Theta_I^3(A_1),\ \text{ if }\ \mathbb{E}[D  \vert  Z=1] - \mathbb{E}[D  \vert  Z=0]= 0,
        \end{array}\right.
\end{eqnarray*}

where  
\begin{eqnarray*}
    \Theta_I^1(A_1) &=& \Bigg \{\Gamma \in \mathbb R^{20}: \theta_{0a}=\theta_{1a} \in \left[\mathbb E[Y\mid D=1, Z=0]-\sup \mathcal Y, \mathbb E[Y\mid D=1, Z=0]-\inf \mathcal Y\right],\\
   && \qquad \theta_{0c}=\theta_{1c}=\frac{\mathbb E[Y\vert Z=1]-\mathbb E[Y\vert Z=0]}{\mathbb E[D\vert Z=1]-\mathbb E[D\vert Z=0]},\\
   && \qquad \theta_{0df}=\theta_{1df} \in \left[\inf \mathcal Y-\sup \mathcal Y, \sup \mathcal Y-\inf \mathcal Y\right],\\
   && \qquad \theta_{0n}=\theta_{1n} \in \left[\inf \mathcal Y-\mathbb E[Y\mid D=0, Z=1], \sup \mathcal Y-\mathbb E[Y\mid D=0, Z=1]\right],\\ 
   && \qquad \delta_{0a}=\delta_{1a}=\delta_{0c}=\delta_{1c}=\delta_{0df}=\delta_{1df}=\delta_{0n}=\delta_{1n}=0,\\
   && \qquad p_a=\mathbb{E}[D  \vert  Z=0], p_c=\mathbb{E}[D  \vert  Z=1] - \mathbb{E}[D  \vert  Z=0], p_{df}=0, p_n=\mathbb{E}[1-D  \vert  Z=1],\\
   && \qquad \text{and inequalities \eqref{eq:testimpl1}-\eqref{eq:testimpl2} hold.} \Bigg \}.
\end{eqnarray*}

Let $\tilde{Z}\equiv 1-Z$, and let $\tilde{T}$ and $\tilde{\Gamma}$ be respectively the corresponding $T$ and $\Gamma$ defined based on the instrument $\tilde{Z}$. We have $\tilde{a}=n$, $\tilde{c}=df$, $\tilde{df}=c$, $\tilde{n}=a$, $\tilde{\theta}_{\tilde{z}\tilde{a}}=\theta_{(1-z)n}$, $\tilde{\theta}_{\tilde{z}\tilde{n}}=\theta_{(1-z)a}$, $\tilde{\theta}_{\tilde{z}\tilde{c}}=\theta_{(1-z)df}$, $\tilde{\theta}_{\tilde{z}\tilde{df}}=\theta_{(1-z)c}$, $\tilde{\delta}_{d\tilde{a}}=\delta_{dn}$, $\tilde{\delta}_{d\tilde{n}}=\delta_{da}$, $\tilde{\delta}_{d\tilde{c}}=\delta_{d df}$, and $\tilde{\delta}_{d\tilde{df}}=\delta_{dc}$. Then $\Theta_I^2(A_1)=\tilde{\Theta}_I^1(A_1),$ where $\tilde{\Theta}_I^1(A_1)$ is the corresponding $\Theta_I^1(A_1)$ defined for $\tilde{\Gamma}$ based on $(Y,D,\tilde{Z})$ and $\tilde{T}$.

If $\mathbb E[D\mid Z=1]-\mathbb E[D\mid Z=0]=0$, then $p_c-p_{df}=0,$ and under Assumption \ref{ass:mon}, $p_c=p_{df}=0$. This implies $p_a=\mathbb P(D=1\mid Z=1)=\mathbb P(D=1\mid Z=0)=\mathbb P(D=1),$ and $p_n=\mathbb P(D=0)$. Furthermore, inequalities \eqref{eq:testimpl1}-\eqref{eq:testimpl2} must hold with equality, i.e., $Z \indep (Y,D)$. Hence, 
\begin{eqnarray*}
    \Theta_I^3(A_1) &=& \Bigg \{\Gamma \in \mathbb R^{20}: \theta_{0a}=\theta_{1a} \in \left[\mathbb E[Y\mid D=1, Z=0]-\sup \mathcal Y, \mathbb E[Y\mid D=1, Z=0]-\inf \mathcal Y\right],\\
    && \qquad \theta_{0n}=\theta_{1n} \in \bigg[\inf \mathcal Y -\mathbb{E}\left[Y \mid D=0, Z=0\right], \sup \mathcal Y - \mathbb{E}\left[Y \mid D=0, Z=0\right]\bigg],\\ 
    && \qquad \theta_{0c}=\theta_{1c}, \theta_{0df}=\theta_{1df} \in \left[\inf \mathcal Y-\sup \mathcal Y, \sup \mathcal Y-\inf \mathcal Y\right],\\
   &&\qquad \delta_{0a}=\delta_{1a}=\delta_{0c}=\delta_{1c}=\delta_{0df}=\delta_{1df}=\delta_{0n}=\delta_{1n}=0,\\
   && \qquad p_a=\mathbb{E}[D],\ p_c=p_{df}=0,\  p_n=\mathbb{E}[1-D],\\
    && \qquad \text{and inequalities \eqref{eq:testimpl1}-\eqref{eq:testimpl2} hold with equality.}   \Bigg \}.
\end{eqnarray*}

\section{Identification under random assignment and exclusion restriction} \label{sec:ra_er:ass}
Building on the above equations \eqref{eq:maintest1}-\eqref{eq:maintest3}, we can derive bounds on the probability of being a defier as follows: 
\begin{align*}
& \max \left\{\max _{s \in\{0,1\}}\left\{\sup _A\{\mathbb{P}(Y \in A, D=s  \vert  Z=1-s)-\mathbb{P}(Y \in A, D=s  \vert  Z=s)\}\right\}, 0\right\} \\
& \leq p_{d f} \leq \min \{\mathbb{E}[D  \vert  Z=0], \mathbb{E}[1-D  \vert  Z=1]\} .
\end{align*}

\cite{Kedagni2020GeneralizedAssumption} and \cite{Kitagawa2021TheIndependence} show that the following restriction is a sharp testable implication of Assumptions \ref{ass:RA} and \ref{ass:er}:
\begin{eqnarray}\label{test:imp}
    \max_{d\in\{0,1\}} \int_{\mathcal Y} \sup_{z}f_{Y,D \vert Z}(y,d\vert z) d \mu(y) \leq 1.
\end{eqnarray}
In the case when the instrument $Z$ is binary, we can prove that the violation of testable implication in Equation \eqref{test:imp} is equivalent to the identified bounds for the proportion of defiers being empty. Therefore, the following proposition holds.

\begin{proposition}[Sharp bounds for $p_{df}$]\label{prop_pdf}
Under Assumptions \ref{ass:RA} and \ref{ass:er}, the identified set $\Theta_{I}(p_{df})$ for the probability of defiers is given by
\begin{equation*}
    \begin{aligned}
        \Theta_{I}(p_{df}) = & \bigg[\max _{s \in\{0,1\}}\left\{\sup _A\{\mathbb{P}(Y \in A, D=s  \vert  Z=1-s)-\mathbb{P}(Y \in A, D=s  \vert  Z=s)\}\right\}, \\
        & \qquad \qquad \qquad \min \{\mathbb{E}[D  \vert  Z=0], \mathbb{E}[1-D  \vert  Z=1]\}\bigg].
    \end{aligned}
\end{equation*}
Moreover, the identified set $\Theta_{I}(p_{df})$ is empty if and only if inequality \eqref{test:imp} is violated.

\end{proposition}

\begin{proof}
    The detailed proof is provided in Appendix \ref{app:proof_prop1}.
\end{proof}

The bounds in $\Theta_I(p_{df})$ were first derived by \cite{richardson2010analysis} for binary outcomes and recently generalized to any outcomes by \cite{Noack2021SensitivityAssumption}. This is an intermediate result for our main results in this paper.
If the lower bound of $\Theta_I(p_{df})$ is greater than zero, the standard LATE assumptions are jointly rejected (or incompatible), aligning with \cite{Balke1997BoundsCompliance}, \cite{imbens1997estimating}, \cite{Heckman2005StructuralEvaluation}, \cite{Kitagawa2015AValidity}, and \cite{Mourifie2017TestingAssumptions}. 
If instead inequalities \eqref{eq:testimpl1}-\eqref{eq:testimpl2} hold, then the testable implications developed for the LATE assumptions hold, and inequality~\eqref{test:imp} becomes redundant, since the LATE testable implications are sharp \citep{Kitagawa2015AValidity, Mourifie2017TestingAssumptions}. In such a context, the lower bound of $\Theta_I(p_{df})$ is 0, as the supremum in the lower bound is achieved when $A=\emptyset$. \cite{Huber2017SharpNoncompliance} derive bounds on $p_{df}$ under a weaker version of Assumption \ref{ass:RA}, $\mathbb E[Y_{dz}\mid T, Z]=E[Y_{dz}\mid T]$ and $Z \indep T$, which implies the \cite{Manski1990NonparametricEffects} mean independence assumption between the potential outcomes and the instrument, $\mathbb E[Y_{dz} \mid Z]=\mathbb E[Y_{dz}]$. This latter mean independence assumption together with the exclusion restriction assumption~\ref{ass:er} imply similar conditions to \eqref{test:imp}: $\sup_z \mathbb E[YD+(1-D)\inf \mathcal Y \mid Z=z] \leq \inf_z \mathbb E[YD+(1-D)\sup \mathcal Y \mid Z=z]$ and $\sup_z \mathbb E[Y(1-D)+D\inf \mathcal Y\mid Z=z] \leq \inf_z \mathbb E[Y(1-D)+D\sup \mathcal Y\mid Z=z]$. These conditions could be rejected if the support $\mathcal Y$ is bounded. For example, when the outcome is binary (i.e., $\mathcal Y=\{0,1\}$), the testable implication for mean independence and Assumption \ref{ass:er} is identical to inequality~\eqref{test:imp}.  

\cite{Kitagawa2021TheIndependence} focuses on the identified region of marginal distributions of potential outcomes, $Y_1$ and $Y_0$, under Assumptions \ref{ass:RA} and \ref{ass:er}, in the same context of a binary treatment and a binary instrument. This paper has a result that the identified region of marginal distributions of potential outcomes is non-empty if and only if inequality \eqref{test:imp} holds. We focus on the identified set of type proportions and local average treatment effects here and also prove in Proposition \ref{prop_pdf} that the identified set of the defier proportion is non-empty if and only if inequality \eqref{test:imp} is satisfied. Both of these results indicate that the non-emptyness of two different parameters (marginal potential outcome distributions and defiers probability) under Assumptions \ref{ass:RA} and \ref{ass:er} leads to the same testable  inequality~\eqref{test:imp}.

\cite{kwon2024testing} consider a more general framework with a binary instrument and a multivalued discrete treatment. They derive bounds on the proportion of always-takers. Their bounds on the fraction of always-takers can be used to derive bounds on the defiers probability in our framework, since we can express the defiers probability as a function of the always-takers probability.

While our ultimate goal is to derive the identified set for $\Gamma$ under the bundle of assumptions $A_2$, we are first to going derive bounds on the LATE for compliers and defiers, respectively, as these are causal parameters that receive a particular attention in the causal inference literature. Before we move on, we introduce some additional notation. Denote $\mu_{d t} \equiv \mathbb{E}\left[Y_d  \vert  T=t\right], d \in\{0,1\}$ and $t \in$ $\{a, c, d f, n\}$. Under Assumptions~\ref{ass:RA} and~\ref{ass:er}, We have $p_a=\mathbb{E}[D  \vert  Z=0]-p_{d f}$, and $p_n= \mathbb{E}[1-D  \vert  Z=1]-p_{d f}.$ Lastly, we have $p_c=1-p_{df}-p_n-p_a=\mathbb{E}[D  \vert  Z=1] - \mathbb{E}[D  \vert  Z=0] +p_{df}$, where $p_{df} \in \Theta_I(p_{df}).$

To proceed, we are going to derive sharp bounds for $\mathbb P(Y_1 \in A \vert T=a)$ using Equations~\eqref{important3}-\eqref{important4}. 
For simplicity, suppose $p_{df}$ is an interior point of $\Theta_I(p_{df})$ such that $0 < p_{df} < \min\{\mathbb E[D\vert Z=0], \mathbb E[1-D \vert Z=1]\}$.
Equation \eqref{important3} implies 
\begin{eqnarray*}
    \mathbb P(Y_1 \in A \vert T=a) &=& \frac{\mathbb P(Y\in A, D=1 \vert Z=1)- p_c \mathbb P(Y_1 \in A \vert T=c)}{p_a}.
\end{eqnarray*}
Since $\mathbb P(Y_1 \in A \vert T=c) \in [0,1],$ and $\mathbb P(Y_1 \in A \vert T=a) \in [0,1]$, we have
\begin{eqnarray*}
   \max\left\{\frac{\mathbb P(Y\in A, D=1 \vert Z=1)- p_c}{p_a},0\right\} \leq \mathbb P(Y_1 \in A \vert T=a) \leq \min\left\{\frac{\mathbb P(Y\in A, D=1 \vert Z=1)}{p_a},1\right\}.
   \end{eqnarray*}

Similarly, Equation \eqref{important4} implies
\begin{eqnarray*}
   \max\left\{\frac{\mathbb P(Y\in A, D=1 \vert Z=0)- p_{df}}{p_a},0\right\} \leq \mathbb P(Y_1 \in A \vert T=a) \leq \min\left\{\frac{\mathbb P(Y\in A, D=1 \vert Z=0)}{p_a},1\right\},   
   \end{eqnarray*}
   
   Hence, Equations \eqref{important3}-\eqref{important4} together imply 
   \begin{eqnarray}
   && \max\left\{\frac{\mathbb P(Y\in A, D=1 \vert Z=1)- p_{c}}{p_a}, \frac{\mathbb P(Y\in A, D=1 \vert Z=0)- p_{df}}{p_a},0\right\}\nonumber\\
   && \qquad \qquad \leq \mathbb P(Y_1 \in A \vert T=a) \leq \label{eq:func}\\
   && \min\left\{\frac{\mathbb P(Y\in A, D=1 \vert Z=1)}{p_a}, \frac{\mathbb P(Y\in A, D=1 \vert Z=0)}{p_a},1\right\}. \nonumber  
   \end{eqnarray}

A similar reasoning holds for \eqref{important}-\eqref{important2} and helps partially identify $\mathbb P(Y_0 \in A \vert T=n)$. The identified sets for $F_{dt}\equiv F_{Y_d\vert T=t}$, $d\in\{0,1\}$, $t\in \{c,df\}$ are given in Proposition \ref{prop:distbound}.

\begin{proposition}\label{prop:distbound}
For a given $p_{df}$ interior point of $\Theta_I(p_{df})$, pointwise sharp bounds for the distributions $F_{dt}$ are given below:
\begin{eqnarray*}
    F_{1a}^{LB}(y)&\equiv& \max\left\{F_{1a}^{LB_1}(y),F_{1a}^{LB_0}(y)\right\} \leq F_{1a}(y) \leq \min\left\{F_{1a}^{UB_1}(y),F_{1a}^{UB_0}(y)\right\}\equiv F_{1a}^{UB}(y),\\
    F_{0n}^{LB}(y) &\equiv& \max\left\{F_{0n}^{LB_1}(y),F_{0n}^{LB_0}(y)\right\} \leq F_{0n}(y) \leq \min\left\{F_{0n}^{UB_1}(y),F_{0n}^{UB_0}(y)\right\}\equiv F_{0n}^{UB}(y),\\
    F_{1c}(y) &=& \frac{\mathbb P(Y \leq y, D=1 \vert Z=1)-p_a F_{1a}(y)}{p_c },\ F_{0c}(y) = \frac{\mathbb P(Y \leq y, D=0 \vert Z=0)-p_n F_{0n}(y)}{p_c },\\
    F_{1df}(y) &=& \frac{\mathbb P(Y \leq y, D=1 \vert Z=0)-p_a F_{1a}(y)}{p_{df} },\ F_{0df}(y) = \frac{\mathbb P(Y \leq y, D=0 \vert Z=1)-p_n F_{0n}(y)}{p_{df}},
\end{eqnarray*}
where $F_{1a}^{\ell}(y)$, $F_{0n}^{\ell}(y)$, $\ell \in \{LB_0, LB_1, UB_0, UB_1\}$ are defined in Appendix \ref{app:proof_prop2_thm1}, 
and $p_a=\mathbb{E}[D  \vert  Z=0]-p_{d f}$, $p_n= \mathbb{E}[1-D  \vert  Z=1]-p_{d f},$ $p_c=1-p_{df}-p_n-p_a=\mathbb{E}[D  \vert  Z=1] - \mathbb{E}[D  \vert  Z=0] +p_{df}$.

\end{proposition}
While Proposition \ref{prop:distbound} shows \textit{pointwise} sharp bounds for the distributions, we characterize their \textit{functional} sharp bounds (in the terminology of \cite{mourifie2020sharp}) in Appendix \ref{app:proof_prop2_thm1}. 
In the following, we explain the intuition behind the reason why the bounds for $F_{1a}(y)$ are not functionally sharp. Suppose that we are interested in bounding $F_{1a}(y')-F_{1a}(y)$ where $y < y'$. One can take the difference of the bounds in Proposition \ref{prop:distbound} and obtain
$$F_{1a}^{LB}(y')-F_{1a}^{UB}(y) \leq F_{1a}(y')-F_{1a}(y) \leq F_{1a}^{UB}(y')-F_{1a}^{LB}(y).$$
But this approach will not lead to sharp bounds for $F_{1a}(y')-F_{1a}(y)$. Instead, we are going to replace the Borel set $A$ by $(y,y']$ in Equation \eqref{eq:func}. Doing this will yield tighter bounds than the pointwise difference of the bounds in Proposition \ref{prop:distbound}. Pointwise sharp bounds for $F_{1a}(y)$ yield sharp bounds on location parameters such as the mean and the quantiles, but they will not deliver sharp bounds on spread parameters like the interquantile range. Functional sharp bounds for $F_{1a}(y)$ yield sharp bounds for interquantile range parameters.

From the results in Proposition \ref{prop:distbound}, we derive bounds on the mean potential outcomes for types. Let $\mu_F$ denote the expected value of a given cumulative distribution function (cdf) $F$, and $\mu_{dt}\equiv \mathbb E[Y_d \vert T=t]$. Corollary \ref{cor:mubounds} provides sharp bounds on $\mu_{dt}$.    
\begin{corollary}\label{cor:mubounds}
 For a given $p_{df}$ interior point of $\Theta_I(p_{df})$, sharp bounds on $\mu_{1a}$ and $\mu_{0n}$ are given by:
    \begin{eqnarray*}
      && \mu_{F_{1a}^{UB}}(p_a) \leq \mu_{1a}(p_a) \leq \mu_{F_{1a}^{LB}}(p_a),\\
      && \mu_{F_{0n}^{UB}}(p_n) \leq \mu_{0n}(p_n) \leq \mu_{F_{0n}^{LB}}(p_n),
    \end{eqnarray*}
where $p_a=\mathbb{E}[D  \vert  Z=0]-p_{d f}$, $p_n= \mathbb{E}[1-D  \vert  Z=1]-p_{d f}.$ 
\end{corollary}
The sharp bounds in Corollary \ref{cor:mubounds} can be difficult to compute in practice. Tractable valid outer sets are then proposed:
\begin{eqnarray*}
    \max\left\{\mu_{F_{1a}^{UB_1}}, \mu_{F_{1a}^{UB_0}}\right\} &\leq& \mu_{1a} \leq \min\left\{\mu_{F_{1a}^{LB_1}}, \mu_{F_{1a}^{LB_0}}\right\}\\
    \max\left\{\mu_{F_{0n}^{UB_1}}, \mu_{F_{0n}^{UB_0}}\right\} &\leq& \mu_{0n} \leq \min\left\{\mu_{F_{0n}^{LB_1}}, \mu_{F_{0n}^{LB_0}}\right\}.
\end{eqnarray*}

\begin{remark}
    $\max\left\{\mu_{F_{1a}^{UB_1}}, \mu_{F_{1a}^{UB_0}}\right\}=\mu_{F_{1a}^{UB}}$ if $F_{1a}^{UB_1}$ first-order stochastically dominates $F_{1a}^{UB_0}$ or vice versa. Similarly, $\min\left\{\mu_{F_{1a}^{LB_1}}, \mu_{F_{1a}^{LB_0}}\right\}=\mu_{F_{1a}^{LB}}$ if $F_{1a}^{LB_1}$ first-order stochastically dominates $F_{1a}^{LB_0}$ or vice versa. However, in general, 
    $\max\left\{\mu_{F_{1a}^{UB_1}}, \mu_{F_{1a}^{UB_0}}\right\} \leq \mu_{F_{1a}^{UB}}$, and $\min\left\{\mu_{F_{1a}^{LB_1}}, \mu_{F_{1a}^{LB_0}}\right\}\geq \mu_{F_{1a}^{LB}}$. 
\end{remark}

The bounds in Corollary \ref{cor:mubounds} provides sharp bounds on $\mu_{1a}(p_a)$ and $\mu_{0n}(p_n)$ for discrete, continuous, or mixed outcomes. Closed-form expressions for the bounds can be obtained for continuous outcomes with strictly increasing cdfs. These are called \citeauthor{Lee2009TrainingEffects}'s (\citeyear{Lee2009TrainingEffects}) bounds. Lemma \ref{lem:lee2009} displays the bounds. 
\begin{lemma}\label{lem:lee2009}
    The following holds for continuous outcomes with strictly increasing cdfs.
    \begin{eqnarray*}
    \mu_{F_{1a}^{LB_z}} &=& \mathbb E\left[Y\vert D=1, Z=z, Y> F^{-1}_{Y\vert D=1, Z=z}\left(1-\frac{p_a}{\mathbb E[D\vert Z=z]}\right)\right],\\
\mu_{F_{1a}^{UB_z}} &=& \mathbb E\left[Y\vert D=1, Z=z, Y< F^{-1}_{Y\vert D=1, Z=z}\left(\frac{p_a}{\mathbb E[D\vert Z=z]}\right)\right],\\
\mu_{F_{0n}^{LB_z}} &=& \mathbb E\left[Y\vert D=0, Z=z, Y> F^{-1}_{Y\vert D=0, Z=z}\left(1-\frac{p_{n}}{\mathbb E[1-D\vert Z=z]}\right)\right],\\
\mu_{F_{0n}^{UB_z}} &=& \mathbb E\left[Y\vert D=0, Z=z, Y< F^{-1}_{Y\vert D=0, Z=z}\left(\frac{p_{n}}{\mathbb E[1-D\vert Z=z]}\right)\right].
\end{eqnarray*}
\end{lemma}

Let $\mathring{\Theta}_I(p_{df})$ denote the interior of $\Theta_I(p_{df})$. The following theorem proposes sharp bounds on the LATEs for compliers $\theta_{0c}=\theta_{1c}$ and defiers $\theta_{0df}=\theta_{1df}$.  
\begin{theorem}\label{thm1}
Suppose inequality \eqref{test:imp} holds, and $p_{df}$ is an interior point of $\Theta_I(p_{df})$. Under Assumptions \ref{ass:RA} and \ref{ass:er}, sharp bounds for $\mu_{1 c}$ and $\mu_{1 d f}$ are as follows:
\begin{eqnarray*}
\begin{aligned}
& \mu_{1 c}(p_a) = \frac{\mathbb{E}[Y D  \vert  Z=1]-p_a \mu_{1 a}(p_a)}{\mathbb{E}[D  \vert  Z=1]-p_a},\ \ \ \mu_{1 d f}(p_a) = \frac{\mathbb{E}[Y D  \vert  Z=0]-p_a \mu_{1 a}(p_a)}{\mathbb{E}[D  \vert  Z=0]-p_a},
\end{aligned}
\end{eqnarray*}
where $p_a=\mathbb{E}[D  \vert  Z=0]-p_{d f}$, $p_{df} \in \mathring{\Theta}_I(p_{df}),$ and $\mu_{1a}(p_a) \in \left[\mu_{F_{1a}^{UB}}(p_a),\mu_{F_{1a}^{LB}}(p_a)\right]$. 

Similarly, sharp bounds for $\mu_{0 c}$ and $\mu_{0 d f}$ are given by:
\begin{eqnarray*}
\begin{aligned}
& \mu_{0 c}(p_n) = \frac{\mathbb{E}[Y(1-D)  \vert  Z=0]-p_n \mu_{0 n}(p_n)}{\mathbb{E}[1-D  \vert  Z=0]-p_n},\ \ \ \mu_{0 d f}(p_n) = \frac{\mathbb{E}[Y(1-D)  \vert  Z=1]-p_n \mu_{0 n}(p_n)}{\mathbb{E}[1-D  \vert  Z=1]-p_n}, 
\end{aligned}
\end{eqnarray*}
where $p_n=$ $\mathbb{E}[1-D  \vert  Z=1]-p_{d f},$ $p_{df} \in \mathring{\Theta}_I(p_{df})$, $\mu_{0n}(p_n)\in \left[\mu_{F_{0n}^{UB}}(p_{n}),\mu_{F_{0n}^{LB}}(p_{n})\right].$ 

Sharp bounds for $\theta_{0c}$, $\theta_{1c}$, $\theta_{0df}$, $\theta_{1df}$ are given by:
\begin{eqnarray*}
    \theta_{0c}(p_c)=\theta_{1c}(p_c)&=& \mu_{1c}(p_c)-\mu_{0c}(p_c),\ \ \ \theta_{0df}(p_{df})=\theta_{1df}(p_{df})= \mu_{1df}(p_{df})-\mu_{0df}(p_{df}),
\end{eqnarray*}
where $p_c=1-p_{df}-p_n-p_a=\mathbb{E}[D  \vert  Z=1] - \mathbb{E}[D  \vert  Z=0] +p_{df}$, and $p_{df} \in \mathring{\Theta}_I(p_{df})$. 
\end{theorem}

\begin{proof}
    See detailed proofs in Appendix \ref{app:proof_prop2_thm1}.
\end{proof}
Theorem \ref{thm1} provides sharp bounds on the LATEs for compliers and defiers under random assignment and exclusion restriction, without requiring any kind of monotonicity assumption. The bounds differ from those in \cite{Huber2017SharpNoncompliance} in two ways. First, they are derived under a stronger assumption (random assignment instead of mean independence assumption). Second, as we previously discuss, these bounds take into account the testable implication \eqref{test:imp}, while the \cite{Huber2017SharpNoncompliance} bounds do not take into account the mean independence version of this inequality. As explained in \cite{li2024discordant}, these kinds of outer bounds could be misleading. 

Proposition \ref{prop:distbound} and Theorem \ref{thm1} differ from the results in \cite{Kitagawa2021TheIndependence}. While \cite{Kitagawa2021TheIndependence} studies the identified sets of potential outcome distributions and derives sharp bounds for the average treatment effect, our focus is on the identified set of the vector parameter $\Gamma$, as it is a building block for many commonly used parameters, including the average and local average treatment effects.
The proposed approach in Theorem \ref{thm1} also differs from other existing papers \citep{deChaisemartin2017ToleratingMonotonicity, Noack2021SensitivityAssumption, dahl2023never} as it does not assume any weaker version of the monotonicity assumption \ref{ass:mon}.

The identified set for $\Gamma$ is 
\begin{eqnarray*}
\Theta_I(A_2)=\left\{\begin{array}{l}
        \Theta_I^1(A_2),\ \text{ if }\ p_{df} \in \Theta_I(p_{df}) \cap (0,\min\{\mathbb{E}[D  \vert  Z=0], \mathbb{E}[1-D  \vert  Z=1]\}),\\
       \Theta_I^2(A_2)=\Theta_I(A_1),\ \text{ if }\ p_{df}=0\in \Theta_I(p_{df}),\\
       \Theta_I^3(A_2),\ \text{ if }\ p_{df}=\min\{\mathbb{E}[D  \vert  Z=0], \mathbb{E}[1-D  \vert  Z=1]\},\\
       \emptyset,\ \text{ if }\ p_{df} \not\in \Theta_I(p_{df}),        \end{array}\right.
\end{eqnarray*}
where
\begin{eqnarray*}
    \Theta_I^1(A_2) &=& \Bigg \{\Gamma \in \mathbb R^{20}: \theta_{0a}(p_a)=\theta_{1a}(p_a) = \mu_{1a}(p_a) - \mu_{0a}, \theta_{0n}(p_n) = \theta_{1n}(p_n) = \mu_{1n} - \mu_{0n}(p_n),\\
   && \qquad \theta_{0c}(p_{c})=\theta_{1c}(p_{c})=\frac{\mathbb{E}[Y D  \vert  Z=1]-p_a \mu_{1 a}(p_a)}{\mathbb{E}[D  \vert  Z=1]-p_a}-\frac{\mathbb{E}[Y(1-D)  \vert  Z=0]-p_n \mu_{0 n}(p_n)}{\mathbb{E}[1-D  \vert  Z=0]-p_n},\\   
   &&\qquad \theta_{0df}(p_{df})=\theta_{1df}(p_{df}) = \frac{\mathbb{E}[Y D  \vert  Z=0]-p_a \mu_{1 a}(p_a)}{\mathbb{E}[D  \vert  Z=0]-p_a}- \frac{\mathbb{E}[Y(1-D)  \vert  Z=1]-p_n \mu_{0 n}(p_n)}{\mathbb{E}[1-D  \vert  Z=1]-p_n},\\   
   && \qquad \delta_{0a}=\delta_{1a}=\delta_{0c}=\delta_{1c}=\delta_{0df}=\delta_{1df}=\delta_{0n}=\delta_{1n}=0,\\
   && \qquad p_a=\mathbb{E}[D  \vert  Z=0]-p_{d f}, p_c=\mathbb{E}[D  \vert  Z=1] - \mathbb{E}[D  \vert  Z=0] +p_{df},\\
   && \qquad p_{df}\in \Theta_I(p_{df}) \cap (0,\min\{\mathbb{E}[D  \vert  Z=0], \mathbb{E}[1-D  \vert  Z=1]\}), p_n=\mathbb{E}[1-D  \vert  Z=1]-p_{d f},\\
    && \qquad \mu_{1a}(p_a) \in \left[\mu_{F_{1a}^{UB}}(p_a),\mu_{F_{1a}^{LB}}(p_a)\right], 
    \mu_{0n}(p_n)\in \left[\mu_{F_{0n}^{UB}}(p_{n}),\mu_{F_{0n}^{LB}}(p_{n})\right],  \\
    && \qquad \mu_{1n}, \mu_{0a} \in \left[\inf \mathcal{Y}, \sup \mathcal{Y}\right]
   \Bigg \}, 
\end{eqnarray*}
which arises when $p_{df}>0$ is in the interior of the identified set for the proportion of defiers, $\Theta_I(p_{df})$. When $p_{df}$ is equal to 0, the identified set for $\Gamma$ coincides with $\Theta_I(A_1)$, which corresponds to the setting where all assumptions in $A_1$ (RA, ER, MON) hold. When $p_{df}$ is equal to the upper bound of $\Theta_I(p_{df})$, the identified set for $\Gamma$ under $A_2$ is $\Theta_I^3(A_2)$, which is formally defined in Appendix \ref{app:proof_prop2_thm1}.

\subsection{Numerical Illustration}\label{simu}
In this section, we illustrate how informative our proposed bounds can be in the context of a double hurdle model for a given data generating process. This is an example where our bounds correctly identify the sign of the LATEs for compliers and defiers separately while the IV estimand yields the wrong sign. Consider the following data-generating process\footnote{Note that this a version of the double hurdle model considered in the introduction, since we can equivalently write $D$ as follows:\\ $D = \mathbbm{1}\{V_1\leq 2 Z, -V_2< -Z\}=\mathbbm{1}\{\Phi(V_1)\leq \Phi(2 Z), \Phi(-V_2) < \Phi(-Z)\}=\mathbbm{1}\{\tilde{V}_1\leq Q_1(Z), \tilde{V}_2 < Q_2(Z)\}$, where $\tilde{V_1}\equiv \Phi(V_1)$, $\tilde{V}_2\equiv \Phi(-V_2)$, $Q_1(Z)\equiv \Phi(2Z)$, $Q_2(Z) \equiv \Phi(-Z)$.}
\begin{equation*}
    \left\{\begin{array}{lll}
Y&= &\beta D+U \\
D&= &\mathbbm{1}\{V_1\leq 2 Z, V_2> Z\} \\
Z&= &\mathbbm{1}\{\epsilon>0\}
\end{array}\right.
\end{equation*}
where $\beta=5 \Phi(2V_1+V_2), U=\frac{1}{2} (V_1+V_2), (V_1, V_2, \varepsilon)^{\prime} \sim N(0, I)$, and $\Phi(\cdot)$ is the standard normal cdf.

The potential treatments are then defined as
\begin{equation*}
    \left\{\begin{array}{lll}
D_0&= &\mathbbm{1}\{V_1 \leq 0, V_2>0\}, \\
D_1&= &\mathbbm{1}\{V_1 \leq 2, V_2>1\}. 
\end{array}\right.
\end{equation*}

In Table \ref{dgp3}, we can compute the following some important parameters from this DGP.
\begin{table}
	\centering
	\begin{tabular}{lll}
		\toprule
		Parameters & Formulas   & Value                \\ \midrule
		$p_{df}$          & $P(D_0=1, D_1=0)$      &     0.1708           \\
		$p_a$          & $P(D_0=1, D_1=1)$    &  0.0793  \\
		$p_c$          & $P(D_0=0, D_1=0)$ &  0.0756          \\
		$p_n$          & $P(D_0=0, D_1=0)$   & 0.6743\\
  	$LATE_c$     & $E[Y_1 \vert T=c] - E[Y_0 \vert T=c] $   &   4.9253 \\
		$LATE_{df}$      & $E[Y_1 \vert T=df] - E[Y_0 \vert T=df] $  &  1.2317   \\
              &&\\
& $E[D \vert Z=0]-E[D \vert Z=1]$ & 0.0955 \\
$IV$ estimand & $\frac{E[Y \vert Z=1]-E[Y \vert Z=0]}{E[D \vert Z=1]-E[D \vert Z=0]}$ &  -1.6874 \\
 \bottomrule 
	\end{tabular}
\caption{Numerical results (sample size $n = 10,000,000$)}
\label{dgp3}
\end{table}
In this DGP, the true causal effects for compliers and defiers are $L A T E_c \equiv$ $\mathbb{E}\left[Y_1-Y_0  \vert  T=c\right]=4.93$ and $L A T E_{d f} \equiv \mathbb{E}\left[Y_1-Y_0  \vert  T=d f\right]=1.23$. Our lower bounds for compliers and defiers from Figure~\ref{f11-1}  are strictly larger than zero, showing our bound can clearly identify the positive sign and reasonable regions of LATE for two subgroups of people. However, when we add the IV estimand to Figure~\ref{f11-1}, we can see that the IV estimand is worse than our bounds under this DGP in two dimensions: magnitude and sign. The IV estimand is lower than the lower bounds. Also, compared to the positive true LATE for defiers ($LATE_{df}= 1.23$), the sign of the IV estimand ($LATE_{IV}=  -1.69$) is negative, which is misleading to policymakers. Indeed, when there exist both compliers and defiers in the population, the IV estimand is a weighted average of the LATEs for compliers and defiers, with negative weights.\footnote{IV estimand $=\frac{p_c LATE_c -p_{df} LATE_{df}}{p_c-p_{df}}$.} These negative weights lead to sign reversal of the IV estimand in this example.

\begin{figure}
	\centering
	\includegraphics[scale=0.7]{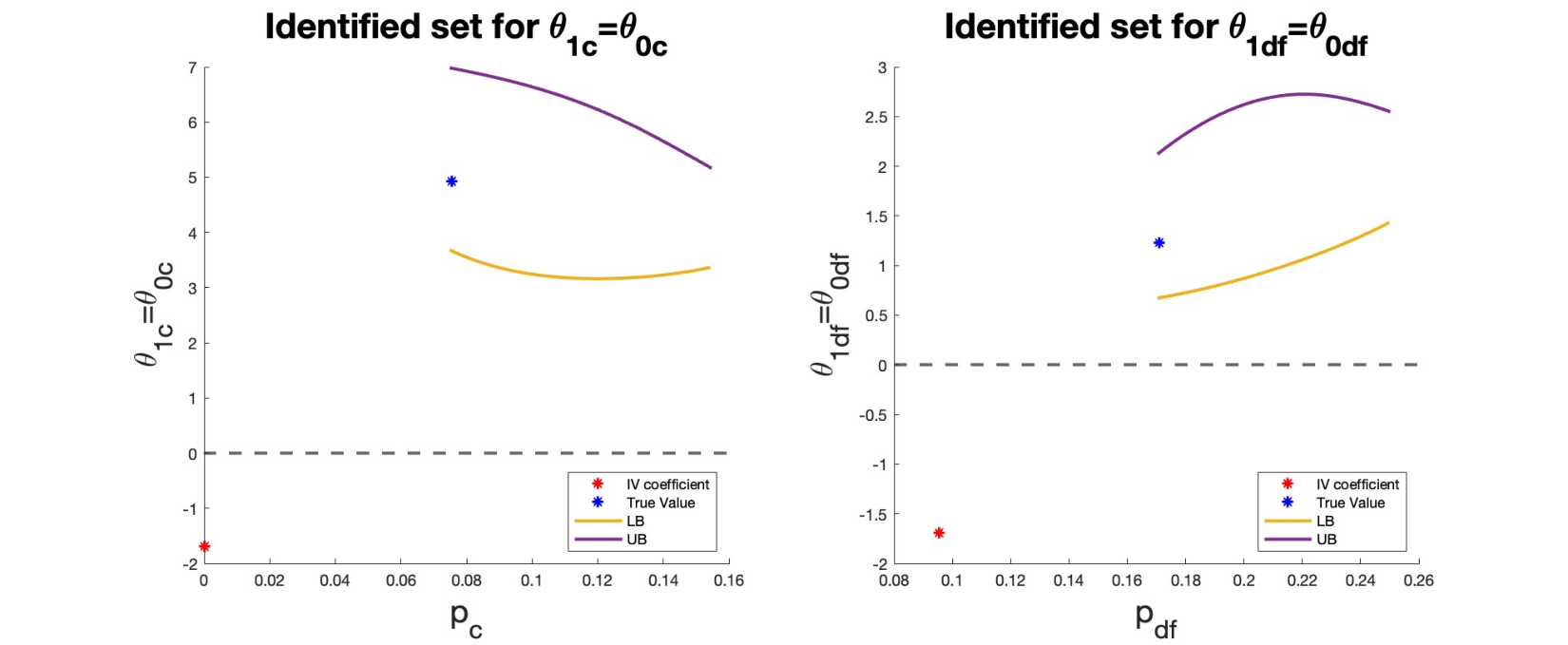}
	\caption{Estimated bounds for LAICDEs, true values for LAICDEs, and IV estimand}
	\label{f11-1}
\end{figure}

\section{Identification under random assignment and monotonicity}\label{sec:ra_mon:ass}

In this section, the potential outcome $Y_{dz}$ is allowed to vary with $z$ in such a way that we may have $\theta_{0t}\neq \theta_{1t}$ and $\delta_{dt}\neq 0$. This section demonstrates how researchers can use a randomly assigned instrument that may violate the exclusion assumption \ref{ass:er}. We maintain the monotonicity assumption \ref{ass:mon}. To proceed, as in the previous section, we write the identified probabilities in each of the observed four subgroups $\{(D=d,Z=z): z, d \in \{0,1\}\}$ as a mixture of the potential outcome distributions for the unobserved types. More precisely, under Assumption \ref{ass:RA}, Equations \eqref{important3}-\eqref{important2} become
\begin{eqnarray}
\label{eq:noer1} 
\mathbb{P}(Y \in A, D=1  \vert  Z=1) &=& p_c \mathbb{P}(Y_{11} \in A  \vert  T=c)+p_a \mathbb{P}(Y_{11} \in A  \vert  T=a), \\
\label{eq:noer2}
\mathbb{P}(Y \in A, D=1  \vert  Z=0) &=& p_{df} \mathbb{P}(Y_{10} \in A  \vert  T=df)+p_a \mathbb{P}(Y_{10} \in A  \vert  T=a),\\
\label{eq:noer3}
\mathbb{P}(Y \in A, D=0  \vert  Z=1) &=& p_{d f} \mathbb{P}(Y_{01} \in A  \vert  T=d f)+p_n \mathbb{P}(Y_{01} \in A  \vert  T=n), \\
\label{eq:noer4}
\mathbb{P}(Y \in A, D=0  \vert  Z=0) &=& p_c \mathbb{P}(Y_{00} \in A  \vert  T=c)+p_n \mathbb{P}(Y_{00} \in A  \vert  T=n).
\end{eqnarray}
Suppose first that $\mathbb{E}[D  \vert  Z=1] - \mathbb{E}[D  \vert  Z=0]> 0$. Then, under Assumption \ref{ass:mon}, there are no defiers, i.e., $p_{df}=0$. As a result, $p_c=\mathbb{E}[D  \vert  Z=1] - \mathbb{E}[D  \vert  Z=0]$, $p_a=\mathbb E[D\vert Z=0]$, and $p_n=\mathbb E[1-D\vert Z=1]$.
Equation \eqref{eq:noer1} implies 
\begin{eqnarray*}
    \mathbb P(Y_{11} \in A \vert T=c) &=& \frac{\mathbb P(Y\in A, D=1 \vert Z=1)- p_a \mathbb P(Y_{11} \in A \vert T=a)}{p_c}.
\end{eqnarray*}
Since $\mathbb P(Y_{11} \in A \vert T=c) \in [0,1],$ and $\mathbb P(Y_{11} \in A \vert T=a) \in [0,1]$, we have
\begin{eqnarray*}
   \max\left\{\frac{\mathbb P(Y\in A, D=1 \vert Z=1)- p_c}{p_a},0\right\} \leq \mathbb P(Y_{11} \in A \vert T=a) \leq \min\left\{\frac{\mathbb P(Y\in A, D=1 \vert Z=1)}{p_a},1\right\}.  
\end{eqnarray*}
Equation \eqref{eq:noer2} implies $\mathbb P(Y_{10}\in A \vert T=a)=\mathbb P(Y \in A \vert D=1, Z=0)$, while Equation \eqref{eq:noer3} implies $\mathbb P(Y_{01}\in A \vert T=n)=\mathbb P(Y \in A \vert D=0, Z=1)$. Finally, Equation \eqref{eq:noer4} implies
\begin{eqnarray*}
   \max\left\{\frac{\mathbb P(Y\in A, D=0 \vert Z=0)- p_c}{p_n},0\right\} \leq \mathbb P(Y_{00} \in A \vert T=n) \leq \min\left\{\frac{\mathbb P(Y\in A, D=0 \vert Z=0)}{p_n},1\right\},  
\end{eqnarray*}
and 
\begin{eqnarray*}
    \mathbb P(Y_{00} \in A \vert T=c) &=& \frac{\mathbb P(Y\in A, D=0 \vert Z=0)- p_n \mathbb P(Y_{00} \in A \vert T=n)}{p_c}.
\end{eqnarray*}
The above bounding approach builds on \cite{Horowitz1995IdentificationData}. We can then suitably take the expectations of the bounds and obtain the so-called \cite{Lee2009TrainingEffects} bounds. Let $\mu_{dzt}\equiv \mathbb E[Y_{dz} \vert T=t]$, $F_{dzt}\equiv F_{Y_{dz}\vert T=t}$. The following proposition \ref{prop:noer} derives sharp bounds on the distributions $F_{dzt}$.

\begin{proposition}\label{prop:noer}
Suppose Assumptions \ref{ass:RA} and \ref{ass:mon} hold, and $\mathbb{E}[D  \vert  Z=1] - \mathbb{E}[D  \vert  Z=0]> 0$. Then,
\begin{eqnarray*}
    && p_c=\mathbb{E}[D  \vert  Z=1] - \mathbb{E}[D  \vert  Z=0],\ p_a=\mathbb E[D\vert Z=0],\ p_n=\mathbb E[1-D\vert Z=1], p_{df}=0,\\
    && F_{11a}^{LB}(y)\equiv \max\left\{\frac{\mathbb P(Y\leq y, D=1 \vert Z=1)- p_c}{p_a},0\right\} \leq F_{11a}(y)\\
    && \qquad \qquad \qquad \qquad \qquad \qquad \qquad \qquad \qquad \qquad \qquad \leq \min\left\{\frac{\mathbb P(Y\leq y, D=1 \vert Z=1)}{p_a},1\right\} \equiv F_{11a}^{UB}(y),\\
    &&F_{00n}^{LB}(y) \equiv \max\left\{\frac{\mathbb P(Y\leq y, D=0 \vert Z=0)- p_c}{p_n},0\right\} \leq F_{00n}(y)\\
    && \qquad \qquad \qquad \qquad \qquad \qquad \qquad \qquad \qquad \qquad \qquad \leq \min\left\{\frac{\mathbb P(Y\leq y, D=0 \vert Z=0)}{p_n},1\right\} \equiv F_{00n}^{UB}(y),\\
    && F_{11c}(y) = \frac{\mathbb P(Y\leq y, D=1 \vert Z=1)- p_a F_{11a}(y)}{p_c},\\
    && F_{00c}(y) = \frac{\mathbb P(Y\leq y, D=0 \vert Z=0)- p_n F_{00n}(y)}{p_c},\\
    && F_{10a(y)}=\mathbb P(Y \leq y \vert D=1, Z=0),\ F_{01n(y)}=\mathbb P(Y \leq y \vert D=0, Z=1).
\end{eqnarray*}
These above bounds are pointwise sharp.   
\end{proposition}
Proposition \ref{prop:noer} provides important results that help derive sharp bounds on the local average treatment-controlled direct effects for the always-takers and never-takers. A direct implication of this proposition is Corollary \ref{cor:noer1} in Appendix \ref{propC}
, which derives sharp bounds on the parameters $\mu_{dzt}$. Building on the results in Corollary \ref{cor:noer1}, Corollary \ref{cor:noer2} below provides sharp bounds on the parameters $\delta_{1a}$, $\delta_{0n}$, and $\theta_{1c}+\delta_{0c}=\theta_{0c}+\delta_{1c}$. To the best of our knowledge, this is the first paper to formally derive sharp bounds on these parameters under the random assignment and monotonicity assumptions. Note however that $\delta_{1a}$ and $\delta_{0n}$ coincide with the local net average treatment effect (LNATE) for the always-takers and never-takers, respectively, defined in \cite{flores2013partial}. In general, $LNATE_t\equiv \mathbb E[Y_{D_0 1}-Y_{D_0 0}\vert T=t]$ is different from our LATCDE parameter $\delta_{dt}$.

\begin{corollary}\label{cor:noer2}
   Suppose Assumptions \ref{ass:RA} and \ref{ass:mon} hold, and $\mathbb{E}[D  \vert  Z=1] - \mathbb{E}[D  \vert  Z=0]> 0$. Then,
\begin{eqnarray*}
    && p_c=\mathbb{E}[D  \vert  Z=1] - \mathbb{E}[D  \vert  Z=0],\ p_a=\mathbb E[D\vert Z=0],\ p_n=\mathbb E[1-D\vert Z=1], p_{df}=0,\\
    && \mu_{F_{11a}^{UB}}-\mathbb E[Y \vert D=1, Z=0] \leq \delta_{1a} \leq \mu_{F_{11a}^{LB}}-\mathbb E[Y \vert D=1, Z=0],\\
    &&\mathbb E[Y \vert D=0, Z=1]-\mu_{F_{00n}^{LB}} \leq \delta_{0n} \leq \mathbb E[Y \vert D=0, Z=1]-\mu_{F_{00n}^{UB}},\\
    && \delta_{0c}+\theta_{1c}=\delta_{1c}+\theta_{0c} = \frac{\mathbb E[Y D \vert Z=1]- p_a \mu_{11a}}{p_c}-\frac{\mathbb E[Y (1-D) \vert Z=0]- p_n \mu_{00n}}{p_c},\\
    && \mu_{F_{11a}^{UB}} \leq \mu_{11a} \leq \mu_{F_{11a}^{LB}},\ \mu_{F_{00n}^{UB}} \leq \mu_{00n} \leq \mu_{F_{00n}^{LB}}.
\end{eqnarray*}   
These bounds are sharp.  

\end{corollary}
Proposition \ref{prop:noer} and Corollaries \ref{cor:noer2}-\ref{cor:noer1} are intermediate results that are key to deriving the identified set for the vector of parameters $\Gamma$. 
The identified set for $\Gamma$ is 
\begin{eqnarray*} 
\Theta_I(A_3)=\left\{\begin{array}{l}
        \Theta_I^1(A_3),\ \text{ if }\ \mathbb{E}[D  \vert  Z=1] - \mathbb{E}[D  \vert  Z=0]> 0,\\
        \Theta_I^2(A_3),\ \text{ if }\ \mathbb{E}[D  \vert  Z=1] - \mathbb{E}[D  \vert  Z=0]< 0,\\
       \Theta_I^3(A_3),\ \text{ if }\ \mathbb{E}[D  \vert  Z=1] - \mathbb{E}[D  \vert  Z=0]= 0,
        \end{array}\right.
\end{eqnarray*} 
where
\begin{eqnarray*}
    \Theta_I^1(A_3) &=& \Bigg \{\Gamma \in \mathbb R^{20}: \theta_{zt}=\mu_{1zt}-\mu_{0zt},\ t\in \{a,c,df,n\},\ z\in \{0,1\},\\ 
    && \qquad \delta_{dt} = \mu_{d1t}-\mu_{d0t},\ t\in \{a,c,df,n\},\ d\in\{0,1\},\\
   && \qquad \mu_{10a}=\mathbb E[Y\vert D=1, Z=0],\ \mu_{01n}=\mathbb E[Y\vert D=0,Z=1],\\
   && \qquad \mu_{F_{11a}^{UB}} \leq \mu_{11a} \leq \mu_{F_{11a}^{LB}},\ \mu_{F_{00n}^{UB}} \leq \mu_{00n} \leq \mu_{F_{00n}^{LB}},\\
   && \qquad \mu_{11c} = \frac{\mathbb E[Y D \vert Z=1]- p_a \mu_{11a}}{p_c},\ \mu_{00c} = \frac{\mathbb E[Y (1-D) \vert Z=0]- p_n \mu_{00n}}{p_c},\\
   && \qquad \mu_{10n}, \mu_{11n}, \mu_{01a}, \mu_{00a}, \mu_{10c}, \mu_{01c}, \mu_{11df}, \mu_{10df}, \mu_{01df}, \mu_{00df} \in [\inf \mathcal Y, \sup \mathcal Y],\\
   && \qquad p_a=\mathbb{E}[D  \vert  Z=0],\ p_c=\mathbb{E}[D  \vert  Z=1] - \mathbb{E}[D  \vert  Z=0],\ p_{df}=0,\  p_n=\mathbb{E}[1-D  \vert  Z=1]  
   \Bigg \},
\end{eqnarray*}
$\Theta_I^2(A_3)=\tilde{\Theta}_I^1(A_3),$ where $\tilde{\Theta}_I^1(A_3)$ is the corresponding $\Theta_I^1(A_3)$ defined for $\tilde{\Gamma}$ based on $(Y,D,\tilde{Z})$ and $\tilde{T}$, 
and
\begin{eqnarray*}
    \Theta_I^3(A_3) &=& \Bigg \{\Gamma \in \mathbb R^{20}: \theta_{zt}=\mu_{1zt}-\mu_{0zt},\ t\in \{a,c,df,n\},\ z\in \{0,1\},\\ 
    && \qquad \delta_{dt} = \mu_{d1t}-\mu_{d0t},\ t\in \{a,c,df,n\},\ d\in\{0,1\},\\ 
   && \qquad \mu_{10a}=\mathbb E[Y\vert D=1, Z=0],\ \mu_{01n}=\mathbb E[Y\vert D=0,Z=1],\\
   &&\qquad \mu_{11a} = \mathbb{E}[Y \mid D = 1, Z = 1],\ \mu_{00n} = \mathbb{E}[Y \mid D = 0, Z = 0],\\
   && \qquad \mu_{10n}, \mu_{11n}, \mu_{01a}, \mu_{00a}, \mu_{10c}, \mu_{01c}, \mu_{11c}, \mu_{00c}, \mu_{11df}, \mu_{10df}, \mu_{01df}, \mu_{00df} \in [\inf \mathcal Y, \sup \mathcal Y],\\
   && \qquad p_a=\mathbb{E}[D],\ p_c=p_{df}=0,\  p_n=\mathbb{E}[1-D]  
   \Bigg \}.
\end{eqnarray*}

\begin{remark}
We can express $\mathbb{E}[Y \mid Z = 1]$ as 
\begin{equation*}
    \begin{aligned}
        \mathbb{E}[Y \mid Z = 1] =& \mathbb{E}[Y_{11} D_1 + Y_{01} (1 - D_1) \mid Z = 1] \\
        =& \mathbb{E}[Y_{11} D_1 + Y_{01} (1 - D_1)] \\
        =& \mathbb{E}[Y_{11} \mid T = a] p_a + \mathbb{E}[Y_{11} \mid T = c] p_c + \mathbb{E}[Y_{01} \mid T = n] p_n,
    \end{aligned}
\end{equation*}
where the second equality holds under Assumption \ref{ass:RA}, and we can apply the law of iterated expectation to obtain the third equality. Using the similar reasoning, we can also write $\mathbb{E}\left[Y \mid Z = 0\right]$ as
\begin{equation*}
    \begin{aligned}
        \mathbb{E}[Y \mid Z = 0] =& \mathbb{E}[Y_{10} D_0 + Y_{00} (1 - D_0) \mid Z = 0] \\
        =& \mathbb{E}[Y_{10} D_0 + Y_{00} (1 - D_0)] \\
        =& \mathbb{E}[Y_{10} \mid T = a] p_a + \mathbb{E}[Y_{00} \mid T = c] p_c + \mathbb{E}[Y_{00} \mid T = n] p_n.
    \end{aligned}
\end{equation*}



As is customary, define the intent-to-treat (ITT) parameter as $\operatorname{ITT} \equiv \mathbb{E}[Y \mid Z = 1] - \mathbb{E}[Y \mid Z = 0]$. Then 
\begin{equation*}
    \begin{aligned}
        \operatorname{ITT} =& \mathbb{E}[Y \mid Z = 1] - \mathbb{E}[Y \mid Z = 0] \\
        =& \mathbb{E}[Y_{11} - Y_{10} \mid T = a] p_a + \mathbb{E}[Y_{01} - Y_{00} \mid T = n] p_n + \mathbb{E}[Y_{11} - Y_{00} \mid T = c] p_c \\
        =& \delta_{1a} p_a + \delta_{0n} p_n + \left[ \delta_{1c}+\theta_{0c} \right] p_c.
    \end{aligned}
\end{equation*}
Since $p_a, p_n, p_c \geq 0$, and $p_a + p_n + p_c = 1$ under Assumption \ref{ass:mon}, we have shown that $\operatorname{ITT}$ can be expressed as a convex combination of $\delta_{1a}$, $\delta_{0n}$, and $\delta_{1c}+\theta_{0c}$.
\end{remark}

\begin{remark}(Testable implication)
    In Corollary \ref{cor:noer1}, we point identify $\mathbb{E}[Y_{10} \mid T = a]$ and $\mathbb{E}[Y_{01} \mid T = n]$, and partially identify $\mathbb{E}[Y_{11} \mid T = a]$ and $\mathbb{E}[Y_{00} \mid T = n]$. Those identification results can be exploited as a testable implication for the exclusion restriction. Since the exclusion restriction states that $Y_{11} = Y_{10}$ and $Y_{01} = Y_{00}$, it implies that $\mathbb{E}[Y_{11} \mid T = a] = \mathbb{E}[Y_{10} \mid T = a]$ and $\mathbb{E}[Y_{01} \mid T = n] = \mathbb{E}[Y_{00} \mid T = n]$. Therefore, if the point identified $\mathbb{E}[Y_{10} \mid T = a]$ is not in the identified set of $\mathbb{E}[Y_{11} \mid T = a]$, or the point identified $\mathbb{E}[Y_{01} \mid T = n]$ is not contained in the identified set of $\mathbb{E}[Y_{00} \mid T = n]$, we have enough evidence to reject the exclusion restriction assumption \ref{ass:er}. However, this testable implication is not sharp for testing the exclusion restriction assumption \ref{ass:er}. \citet[Proposition 3.2 (ii)]{Kitagawa2021TheIndependence} showed that inequality~\eqref{test:imp} is the sharp testable implication for the random assignment and exclusion restriction assumptions \ref{ass:RA} and \ref{ass:er} together. This result is confirmed in the sharpness proof of Proposition \ref{prop_pdf} in this paper. In a randomized experiment setting like the one considered in this paper, the testable inequality~\eqref{test:imp} can be interpreted as a testable implication for only the exclusion restriction assumption. 
\end{remark}

\begin{remark}
Notice that unlike $\Theta_I(A_1)$ and $\Theta_I(A_2)$ which could be empty, $\Theta_I(A_3)$ is never empty. This implies that Assumptions \ref{ass:RA} and \ref{ass:mon} do not have any testable implication. 
\end{remark}

\section{Misspecification robust bounds}\label{sec:mrb}
In this section, we follow \cite{li2024discordant} to derive a misspecification robust bound for the model and the various assumptions considered in this paper. This is a discrete way of salvaging the model when it is refuted by the data. This relaxation is different from the continuous relaxation approach introduced in \cite{Masten2021SalvagingModels}. We find this discrete relaxation easier to interpret than the continuous one. Before we proceed to show the misspecification robust bound for the collection of assumptions $A$, we discuss when each of the assumptions $A_1$, $A_2$, and $A_3$ is data-consistent. According to \cite{li2024discordant}, an assumption $A_k$ is data-consistent if its identified set is nonempty. Therefore, $A_1$ is data-consistent if $\Theta_I^1(A_1)$ is nonempty when $\mathbb E[D\mid Z=1]-\mathbb E[D\mid Z=0]>0$, $\Theta_I^2(A_1)$ is nonempty when $\mathbb E[D\mid Z=1]-\mathbb E[D\mid Z=0]<0$, and $\Theta_I^3(A_1)$ is nonempty when $\mathbb E[D\mid Z=1]-\mathbb E[D\mid Z=0]=0$. Assumption $A_2$ is data-consistent if inequality \eqref{test:imp} holds. Assumption $A_3$ is always data-consistent as its identified set is never empty. 
\begin{definition}[\citeauthor{li2024discordant}, \citeyear{li2024discordant}]
An assumption $A_k\in A$ is a minimum data-consistent relaxation of $A$ if $\Theta_I(A_k)\neq \emptyset$ and $\Theta_I(A_k\cup A_{k'})=\emptyset$ for any $A_{k'} \neq A_k: A_{k'} \in A$. 
\end{definition}
The misspecification robust bound $\Theta_I^*(A)$ is defined as the union of the identified sets of \textit{all} minimum data-consistent relaxations of $A$. Note here that $A_2 \cup A_3=A_1$. Then $\Theta_I(A)=\Theta_I(A_1\cup A_2\cup A_3)=\Theta_I(A_1).$ If $\Theta(A)\neq \emptyset$, there exists a unique minimum data-consistent relaxation by definition, it is $A_1$. But, if $\Theta_I(A_1)=\emptyset$ and $\Theta_I(A_2)\neq \emptyset$, we have two minimum data-consistent relaxations, $A_2$ and $A_3$. Finally, if $\Theta(A_2)=\emptyset$ (i.e., inequality~\eqref{test:imp} is violated), then $A_3$ will be the only minimum data-consistent relaxation. Hence, we have
\begin{eqnarray*}
\Theta_I^{*}(A)&=&\left\{\begin{array}{l}
       \Theta_I(A_1)\ \ \text{ if } \Theta_I(A_1)\neq \emptyset,\\
\Theta_I(A_{2})\cup\Theta_I(A_{3})\ \ \text{ if } \Theta_I(A_1) = \emptyset \text{ and } \Theta_I(A_2) \neq \emptyset,\\
\Theta_I(A_3)\ \ \text{ if } \Theta_I(A_2) = \emptyset
        \end{array}\right.\\ \\
&=&\left\{\begin{array}{l}
\Theta_I(A_1)\ \ \text{ if } \Theta_I(A_1)\neq \emptyset,\\
\Theta_I(A_{2})\cup \Theta_I(A_{3})\ \ \text{ if } \Theta_I(A_1) = \emptyset.
\end{array}\right.
\end{eqnarray*}
In the next section, we illustrate the theoretical results with a real-world application. 
\section{Empirical illustration}\label{sec:emp}


In this section, we illustrate our methodology by revisiting \cite{bursztyn2020misperceived}, which is also re-examined by \cite{kwon2024testing} in a mechanism framework. \cite{bursztyn2020misperceived} study how misperceived social norms affect women's labor force participation in Saudi Arabia. They find that while many men privately support female employment, they incorrectly believe that their peers do not share this view. To study the effects of correcting these misperceptions, the authors design a randomized experiment in which men are randomly informed that a majority of their peers actually support female employment. The results show that women whose husbands receive the information are significantly more likely to register for a job-search service. In addition, the study finds persistent effects on women's employment behavior, as women from treated households are more likely to be employed several months after the intervention.

To align this study with our framework, we define the women long-term labor supply outcome\footnote{\cite{bursztyn2020misperceived} construct an index by pooling six measures of women's long-term labor supply outcomes, which we use as the outcome in our analysis.} as the outcome $Y$, the decision to sign up for a job-search service as the binary treatment variable $D$, and the information regarding social norms as a randomized experiment $Z$. Given that the information experiments in this study are randomly assigned, we maintain Assumption \ref{ass:RA} throughout the analysis. The dataset applied here comes from the online appendix of \cite{bursztyn2020misperceived}. Our sample is restricted to observations with non-missing data on women's long-term labor supply outcomes, resulting in a final sample size of 381 observations. Some summary statistics are provided in Table~\ref{cardsum}. We observe that households assigned to the treatment group ($Z=1$) are more likely to enroll in job-search services compared to those in the control group ($Z=0$). Additionally, treated households, on average, have better long-term female labor force outcomes than the control households.

The results presented in this section are entirely based on estimation. We provide a brief discussion of the inference methods and present the corresponding inference results in Appendix \ref{esti_inf}.

\begin{table}
	\centering
 \caption{Summary Statistics} 
\begin{tabular}{lccc}
\toprule
& Total & $Z = 1$ & $Z = 0$ \\
\midrule
Observations & 381 & 190 & 191 \\
$D$  (job-search service) & $0.294\ (0.445)$ & $0.340\ (0.475)$ & $0.247\ (0.433)$ \\
$Z$  (exposure to information) & $0.501\ (0.500)$ & $1.000\ (0.000)$ & $0.000\ (0.000)$ \\
$Y$  (women labor force index) & $0.006\ (0.499)$ & $0.130\ (0.571)$ & $-0.119\ (0.374)$ \\
\bottomrule 
Average (standard deviation).
\end{tabular}
\label{cardsum}
\end{table}

\subsection{Identified set under $A_1$}
We estimate the difference in probabilities, $\mathbb{P}(D=1 \mid Z=1)-\mathbb{P}(D=1 \mid Z=0)$, and obtain a result of $\hat{\mathbb{P}}(D=1 \mid Z=1)-\hat{\mathbb{P}}(D=1 \mid Z=0) = 0.0929 > 0$. Given this positive estimate, we want to test the inequalities \eqref{eq:testimpl1}-\eqref{eq:testimpl2}, which are the sharp testable implications of the joint assumptions in $A_1 = \{RA, ER, MON\}$. We apply the testing method proposed in \cite{Mourifie2017TestingAssumptions} and reject the validity of inequalities \eqref{eq:testimpl1}-\eqref{eq:testimpl2} at the 1\% level. This rejection implies that the identified set $\Theta_I(A_1)$ is empty, suggesting that the assumptions in $A_1$ are not consistent with the dataset used in \cite{bursztyn2020misperceived}.

The application in \cite{kwon2024testing} differs from ours. 
First, we use a continuous outcome variable constructed in \cite{bursztyn2020misperceived}, which is an index summarizing six measures of female labor force participation, including an indicator for whether the wife applies for jobs outside the home. In contrast, \cite{kwon2024testing} use only this binary indicator as their outcome variable.
Second, in our  notation, the primary objective of \cite{kwon2024testing} is to test the validity of the exclusion restriction  assumption \ref{ass:er}. For their implementation, they assume monotonicity (Assumption \ref{ass:mon}) and find that the set of assumptions $A_1$ is rejected. So, the rejection of $A_1$ implies that the random assignment $Z$ violates the exclusion restriction if monotonicity holds. They also estimate the average direct effects of the random assignment on the outcome of always-takers and never-takers. Finally, they perform a sensitivity analysis by gradually relaxing the monotonicity assumption, assuming a certain proportion of defiers to exist.
In contrast, our method constructs misspecification robust bounds, as discussed in Section \ref{sec:mrb}, by applying the minimum data-consistent relaxation of the LATE assumptions. Unlike \cite{kwon2024testing}, we do not impose any prior restrictions on the proportion of defiers. Instead, we discretely relax the LATE assumptions in a manner that ensures they remain consistent with the observed data.

\subsection{Identified set under $A_2$}

Next, we relax Assumption \ref{ass:mon} and consider identification under the assumption set $A_2 = \{RA, ER\}$. Based on Proposition \ref{prop_pdf}, we estimate the identified set for the proportion of defiers as
\begin{equation*}
    \widehat{\Theta}_I(p_{df}) = [0.1038, 0.2474].
\end{equation*}
Since this identified set is nonempty, Proposition \ref{prop_pdf} implies that we fail to reject the inequality \eqref{test:imp}. Therefore, the combination of assumptions in $A_2$ cannot be rejected. Table~\ref{tab:type_prop_no_mon} reports the estimated identified sets for the proportions of types under the assumption set~$A_2$.

\begin{table}
    \centering
    \caption{Estimated proportion of each type under $A_2$}
      \begin{tabular}{cr}
      \toprule
      \multicolumn{1}{c}{Parameter} & \multicolumn{1}{c}{Estimated identified set} \\
      \midrule
      $p_a$ & [0, 0.1436] \\
      $p_c$ & [0.1967, 0.3403] \\
      $p_n$ & [0.4123, 0.5559] \\
       $p_{df}$    & [0.1038, 0.2474]\\
      Observations & 381 \\
      \bottomrule
      \end{tabular}%
    \label{tab:type_prop_no_mon}%
\end{table}%

Building on the analysis in Section \ref{sec:ra_er:ass}, we derive informative identified sets for the local average instrument-controlled direct effects for compliers, $\theta_{1c} = \theta_{0c}$, and for defiers, $\theta_{1df} = \theta_{0df}$. Figure \ref{fig:est_id_sets} presents the estimated bounds for these identified sets.

\begin{figure}
	\centering
	\includegraphics[scale=0.7]{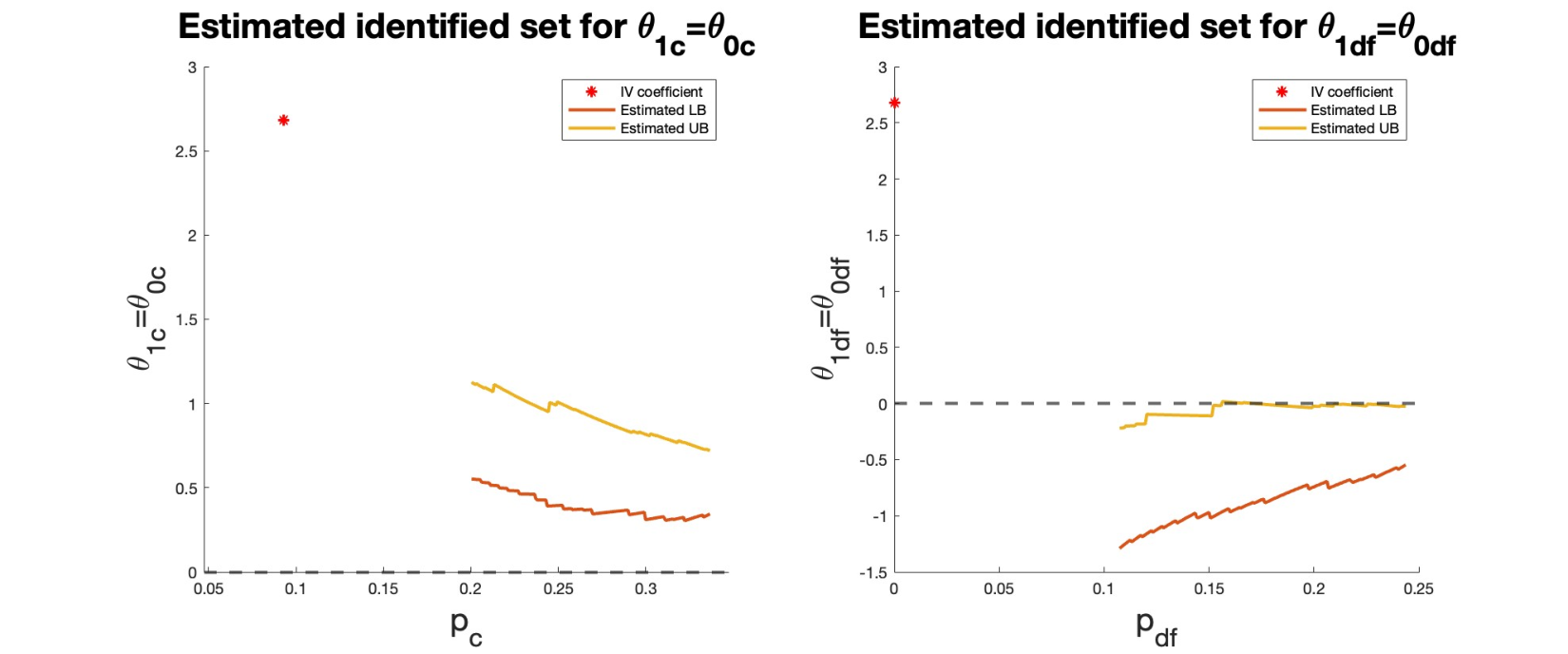}
	\caption{Estimated bounds for LAICDEs under $A_2$ and IV estimate}
	\label{fig:est_id_sets}
\end{figure}

Figure \ref{fig:est_id_sets} shows the heterogeneous effects of job-search services across different household types. Specifically, we observe positive treatment effects for compliers across all values of $p_c$ within the identified set. In complier households, individuals enroll in job-search services only when they are informed about the social norms. Upon learning that more women are employed than they initially thought, they may perceive an increased likelihood of getting a job, as a larger female workforce indicates greater opportunities in the labor market, which motivates them to comply. These characteristics of complier households align with the increase effects on long-term female labor force participation following their enrollment in job-search services. 

In contrast, we observe negative treatment effects for defiers across most values of $p_{df}$ within the identified set. Unlike compliers, women in defier households may regard increased female workforce as a signal of intensified competition. This may discourage them from seeking jobs, particularly if their relative advantage has diminished. Consequently, they choose not to enroll in job-search services once they are informed about the social norms. These characteristics of defier households are consistent with the decreased effects on long-term female labor force participation following enrollment in job-search services, as they may face greater competition in the female labor market.

Notably, the IV estimate $(2.68)$ falls outside our estimated bounds. It overestimates the effect of job-search services for compliers relative to our identified set.

\subsection{Identified set under $A_3$}
Finally, we relax Assumption \ref{ass:er} while maintaining Assumption \ref{ass:mon} and study identification under the assumption set $A_3 = \{RA, MON\}$. As discussed in Section \ref{sec:ra_mon:ass}, the identified set under $A_3$ is always non-empty, implying that the joint assumptions in $A_3$ can never be rejected.

With Assumption \ref{ass:mon} imposed, defiers cannot exist and all proportions of types can be point identified in this framework. Table \ref{tab:type_prop} presents the estimations of type proportions under $A_3$. The results show that never-takers constitute the largest proportion of the population, accounting for 70.68\% of all observations. This finding suggests that the majority of households in the dataset would not sign up job-search services for women, regardless of whether or not they receive the information about correct social norms.

\begin{table}
    \centering
    \caption{Estimated proportion of each type under $A_3$}
      \begin{tabular}{cr}
      \toprule
      \multicolumn{1}{c}{Parameter} & \multicolumn{1}{c}{Estimated proportion} \\
      \midrule
      $p_a$ & 0.2474 \\
      $p_c$ & 0.0929 \\
      $p_n$ & 0.7068 \\
       $p_{df}$    & 0 \\
      Observations & 381 \\
      \bottomrule
      \end{tabular}%
    \label{tab:type_prop}%
\end{table}%

According to the identified set $\Theta_I(A_3)$, we have informative identified sets for the local average treatment-controlled direct effects for always-takers and never-takers, $\delta_{1a}$ and $\delta_{0n}$, and the total effects for compliers, $\delta_{1c} + \theta_{0c} = \delta_{0c} + \theta_{1c}$. Table \ref{tab:est_bounds} presents the estimated results for these identified sets.

\begin{table}
    \centering
    \caption{Estimated informative identified sets under $A_3$}
      \begin{tabular}{lrr}
      \toprule
        Parameter    & \multicolumn{1}{c}{Estimated Bounds} \\
      \midrule
      $\delta_{1a}$ \quad \quad \quad \quad & [0.0557,0.7155] \\
      $\delta_{0n}$ \quad \quad \quad \quad & [0.0829,0.2505] \\
      $\delta_{1c}+\theta_{0c}$ \quad \quad \quad \quad& [-0.7140,1.9242] \\
      \bottomrule
      \end{tabular}%
    \label{tab:est_bounds}%
\end{table}%

The results in Table \ref{tab:est_bounds} indicate that the information experiment has positive direct effects on long-term female labor force outcomes for both always-takers and never-takers, who together account for nearly 90\% of the population. This finding implies that, even though the information experiment does not affect job-search service enrollment decisions in these households, exposure to correct social norm information directly benefits women's long-term labor force participation.

\subsection*{Misspecification robust bound} Since $\widehat{\Theta}_I(A_2)$ and $\widehat{\Theta}_I(A_3)$ are nonempty, the estimated misspecification robust bound in this empirical example is $\widehat{\Theta}^*_I(A)=\widehat{\Theta}_I(A_2) \cup \widehat{\Theta}_I(A_3)$. We explained earlier in Section \ref{sec:mrb} how this result is achieved. 

\section{Conclusion}
\label{sec:summary}
  In this paper, we propose a robust identification approach in a randomized experiment setting with imperfect compliance where the LATE assumptions could potentially be rejected by the data. We propose bounds that are never empty and robust to detectable violations of the LATE assumptions. The bounds are also robust to violations of the monotonicity or exclusion restriction assumptions. 
  We introduce two causal parameters: the local average treatment-controlled direct effect, and the local average instrument-controlled direct effect. We apply our methods to revisit \cite{bursztyn2020misperceived} and study the effect of randomly providing accurate social norm information on the decision to enroll women in job-search services, which may affect women's labor force participation in Saudi Arabia. We show that the LATE assumptions are jointly rejected in this context. However, the assumptions cannot be rejected when either the monotonicity or exclusion restriction is relaxed. Under the assumptions of random assignment and the exclusion restriction, we find that job-search service enrollment has a positive effect on long-term female labor force outcomes for compliers, whereas the effect is negative for defiers. Under the assumptions of random assignment and monotonicity, our results suggest that the information experiment itself has positive direct effects on long-term female labor force outcomes for both always-takers and never-takers, who together account for more than 90\% of the population.
  

  The approach developed in this paper works well for a binary instrument. An interesting area for future research would consist of extending the proposed method to a multivalued discrete or continuous instrument in the context of marginal treatment effects introduced in \cite{Heckman2001Policy-RelevantEffects, Heckman2005StructuralEvaluation}. One challenge is that the number of causal parameters would be infinite in this setting. Recently, \cite{sun2022pairwise} proposed a pairwise validity approach for multivalued discrete instruments in the LATE framework. A combination of their method with the one developed in this paper could be an interesting area for future research. 
  Another consideration for future research could be how to incorporate exogenous covariates in the analysis. While the approach still works conditional on covariates, it may suffer from the curse of dimensionality as the number of parameters grows with the cardinality of the support of the covariates. Semiparametric assumptions could also enter the set $A$ of assumptions considered in the paper. 


\onehalfspacing
\bibliographystyle{jpe}
\renewcommand\refname{Reference}
\bibliography{Template_Bib}

\clearpage
\appendix

\section{Proof of Identified Set Under Random Assignment and Exclusion Restriction}

\subsection{Proof of Proposition \ref{prop_pdf}} \label{app:proof_prop1}   

\begin{proof}
   First, we demonstrate that Equation \eqref{test:imp} is a testable implication of Assumptions \ref{ass:RA} and \ref{ass:er}. Given $d \in \{0, 1\}$, we have
\begin{equation*}
    \begin{aligned}
        \int_{\mathcal Y} \sup _z f_{Y, D \mid Z}(y, d \vert z) d \mu(y) &= \int_{\mathcal Y} \sup _z f_{Y_d, D_z \mid Z}(y, d \vert z) d \mu(y) \\
        & \leq \int_{\mathcal Y} \sup _z \{f_{Y_d, D_z \mid Z}(y, d \vert z)+f_{Y_d, D_z \mid Z}(y, 1-d \vert z)\} d \mu(y)\\
        &= \int_{\mathcal Y} \sup _z f_{Y_d \mid Z}(y \vert z) d \mu(y) \\
        &= \int_{\mathcal Y} f_{Y_d}(y) d \mu(y) = 1,
    \end{aligned}
\end{equation*}
where the first equality holds from  Assumption \ref{ass:er}, the first inequality holds because $f_{Y_d, D_z \mid Z}(y, 1-d \mid z) \geq 0$ for any $y$ and $z$, the second equality holds because $f_{Y_d \mid Z}(y \vert z)=f_{Y_d, D_z \mid Z}(y, d \vert z)+f_{Y_d, D_z \mid Z}(y, 1-d \vert z)$ by the law of total probability, and the third equality follows from Assumption \ref{ass:RA}, and the last equality follows from the fact that a density function integrates to 1. Therefore, Assumptions \ref{ass:RA} and \ref{ass:er} imply the inequality in Equation \eqref{test:imp}, and the identified set of $p_{df}$ is empty if Equation \eqref{test:imp} is violated.

For any Borel set $A \subseteq \mathcal{Y}$, under Assumptions \ref{ass:RA} and \ref{ass:er}, we have
\begin{equation*}
    \begin{aligned}
        & \mathbb{P}\left(Y \in A, D = 1 \mid Z = 0\right) - \mathbb{P}\left(Y \in A, D = 1 \mid Z = 1\right) \\
        =& \mathbb{P}\left(Y_1 \in A, D_0 = 1\right) - \mathbb{P}\left(Y_1 \in A, D_1 = 1\right) \\
        =& \mathbb{P}\left(Y_1 \in A, T = a\right) + \mathbb{P}\left(Y_1 \in A, T = df\right) - \mathbb{P}\left(Y_1 \in A, T = a\right) - \mathbb{P}\left(Y_1 \in A, T = c\right) \\
        =& \mathbb{P}\left(Y_1 \in A, T = df\right) - \mathbb{P}\left(Y_1 \in A, T = c\right) \\
        \leq& \mathbb{P}\left(Y_1 \in A, T = df\right) \leq \mathbb{P}\left(T = df\right),
    \end{aligned}
\end{equation*}
and 
\begin{equation*}
    \begin{aligned}
        & \mathbb{P}\left(Y \in A, D = 0 \mid Z = 1\right) - \mathbb{P}\left(Y \in A, D = 0 \mid Z = 0\right) \\
        =& \mathbb{P}\left(Y_0 \in A, D_1 = 0\right) - \mathbb{P}\left(Y_0 \in A, D_0 = 0\right) \\
        =& \mathbb{P}\left(Y_0 \in A, T = n\right) + \mathbb{P}\left(Y_0 \in A, T = df\right) - \mathbb{P}\left(Y_0 \in A, T = n\right) - \mathbb{P}\left(Y_0 \in A, T = c\right) \\
        =& \mathbb{P}\left(Y_0 \in A, T = df\right) - \mathbb{P}\left(Y_0 \in A, T = c\right) \\
        \leq& \mathbb{P}\left(Y_0 \in A, T = df\right) \leq \mathbb{P}\left(T = df\right).
    \end{aligned}
\end{equation*}
Combined with the property probability that $p_{df} \geq 0$, we have 
\begin{equation*}
    p_{df} \geq \max \left\{\max _{s \in\{0,1\}}\left\{\sup _A\{\mathbb{P}(Y \in A, D=s \mid Z=1-s)-\mathbb{P}(Y \in A, D=s \mid Z=s)\}\right\}, 0\right\}.
\end{equation*}

In addition, according to the definitions and the property that $p_a, p_n \geq 0$,
\begin{equation*}
    \mathbb{P}\left(D = 1 \mid Z = 0\right) = p_a + p_{df} \geq p_{df}, \ \ \mathbb{P}\left(D = 0 \mid Z = 1\right) = p_n + p_{df} \geq p_{df}.
\end{equation*}

To summarize, an identified bound of $p_{df}$ is given by
\small \begin{equation*}
    \Theta_{I}(p_{df}) = \left\{\begin{array}{lll}
    & \bigg[\max \left\{\max _{s \in\{0,1\}}\left\{\sup _A\{\mathbb{P}(Y \in A, D=s  \vert  Z=1-s)-\mathbb{P}(Y \in A, D=s  \vert  Z=s)\}\right\}, 0\right\}, \\
    & \qquad \qquad \qquad \min \{\mathbb{E}[D  \vert  Z=0], \mathbb{E}[1-D  \vert  Z=1]\}\bigg], \text{ if inequality \eqref{test:imp} } \text{ holds } \\
    & \emptyset, \text{ otherwise}.
    \end{array}\right.
    \end{equation*}

The next step is to prove the sharpness of the $\Theta_I(p_{df})$. We have to construct a joint distribution of $(\tilde{Y}_{11}, \tilde{Y}_{10}, \tilde{Y}_{01}, \tilde{Y}_{00}, \tilde{T}, Z)$, $\tilde{T} \in \{a, n, c, df\}, Z \in \{0, 1\}$, such that each point in $\Theta_I(p_{df})$ can be achieved, and the constructed distribution is compatible with imposed assumptions and satisfies

\begin{eqnarray}
\label{eq:proof1} 
\mathbb{P}(Y \in A, D=1  \vert  Z=1) &=& \tilde{p}_c \mathbb{P}(\tilde{Y}_1 \in A  \vert  \tilde{T}=c)+\tilde{p}_a \mathbb{P}(\tilde{Y}_1 \in A  \vert  \tilde{T}=a), \\
\label{eq:proof2}
\mathbb{P}(Y \in A, D=1  \vert  Z=0) &=& \tilde{p}_{df} \mathbb{P}(\tilde{Y}_1 \in A  \vert  \tilde{T}=df)+\tilde{p}_a \mathbb{P}(\tilde{Y}_1 \in A  \vert  \tilde{T}=a),\\
\label{eq:proof3}
\mathbb{P}(Y \in A, D=0  \vert  Z=1) &=& \tilde{p}_{d f} \mathbb{P}(\tilde{Y}_0 \in A  \vert  \tilde{T}=d f)+\tilde{p}_n \mathbb{P}(\tilde{Y}_0 \in A  \vert  \tilde{T}=n), \\
\label{eq:proof4}
\mathbb{P}(Y \in A, D=0  \vert  Z=0) &=& \tilde{p}_c \mathbb{P}(\tilde{Y}_0 \in A  \vert  \tilde{T}=c)+\tilde{p}_n \mathbb{P}(\tilde{Y}_0 \in A  \vert  \tilde{T}=n).
\end{eqnarray}

Define $\tilde{Y}_{dz}=\tilde{Y}_d$ for all $d$, $z$, $\tilde{p}_t\equiv \mathbb P(\tilde{T}=t)=\mathbb P(\tilde{T}=t \mid Z=z)$, $\mathbb P(\tilde{Y}_d \leq y_d \mid \tilde{T}=t, Z=z)=\mathbb P(\tilde{Y}_d \leq y_d \mid \tilde{T}=t),$ and $\mathbb P(\tilde{Y}_1 \leq y_1,\tilde{Y}_0 \leq y_0, \tilde{T}=t \mid Z=z)=\tilde{p}_t \mathbb P(\tilde{Y}_1 \leq y_1 \mid \tilde{T}=t) \mathbb P(\tilde{Y}_0 \leq y_0 \mid \tilde{T}=t)$.

For each $\tilde{p}_{df} \in \Theta_I(p_{df})$, we need to find $\mathbb P(\tilde{Y}_d \leq y_d \mid \tilde{T}=t)$ such that Equations \eqref{eq:proof1}-\eqref{eq:proof4} hold. We will do this case by case.

\subsubsection*{case 1: $\max_{s \in\{0,1\}}\left\{\sup _A\{\mathbb{P}(Y \in A, D=s  \vert  Z=1-s)-\mathbb{P}(Y \in A, D=s  \vert  Z=s)\}\right\} \leq 0$} 
In this case, $\Theta_I(p_{df})=[0,\min \{\mathbb{E}[D  \vert  Z=0], \mathbb{E}[1-D  \vert  Z=1]\}]$. Note that inequality \eqref{test:imp} holds in this case, as $\sup_z f_{Y,D\mid Z}(y,d \mid z)=f_{Y,D\mid Z}(y,d \mid d)$, and $\int_{\mathcal Y} f_{Y,D\mid Z}(y,d \mid d) d \mu(y)=\mathbb{P}(D=d\mid Z=d) \leq 1$ for all $d$. 

\begin{enumerate}
    \item \label{subcase1} $\tilde{p}_{df}=0$. 

We have $\mathbb E[D\mid Z=1]-\mathbb E[D \mid Z=0]=\mathbb P(Y\in \mathcal Y, D=1\mid Z=1)-\mathbb P(Y\in \mathcal Y, D=1\mid Z=0) \geq 0$. First, suppose  $\mathbb E[D\mid Z=1]-\mathbb E[D \mid Z=0] > 0$, and define 
\begin{eqnarray*}
\tilde{p}_a &=& \mathbb E[D\mid Z=0], \ \tilde{p}_n=\mathbb E[1-D \mid Z=1], \ \tilde{p}_c=\mathbb E[D\mid Z=1]-\mathbb E[D \mid Z=0],\\
\mathbb P(\tilde{Y}_1 \leq y_1\mid \tilde{T}=a) &=& \mathbb P(Y \leq y_1 \mid D=1, Z=0),\\
\mathbb P(\tilde{Y}_0 \leq y_0\mid \tilde{T}=n) &=& \mathbb P(Y \leq y_0 \mid D=0, Z=1),\\
\mathbb P(\tilde{Y}_1 \leq y_1\mid \tilde{T}=c) &=& \frac{\mathbb P(Y \leq y_1,D=1 \mid Z=1)-\mathbb P(Y \leq y_1,D=1 \mid Z=0)}{\mathbb E[D\mid Z=1]-\mathbb E[D \mid Z=0]},\\
\mathbb P(\tilde{Y}_0 \leq y_0\mid \tilde{T}=c) &=& \frac{\mathbb P(Y \leq y_0,D=0 \mid Z=0)-\mathbb P(Y \leq y_0,D=0 \mid Z=1)}{\mathbb E[D\mid Z=1]-\mathbb E[D \mid Z=0]},\\
\mathbb P(\tilde{Y}_0 \leq y_0\mid \tilde{T}=a) &=& \mathbb P(Y \leq y_0 \mid D=1, Z=0),\\
\mathbb P(\tilde{Y}_1 \leq y_1\mid \tilde{T}=n) &=& \mathbb P(Y \leq y_1 \mid D=0, Z=1)
\end{eqnarray*}
Since $\mathbb{P}(Y \in A, D=d  \vert  Z=d)-\mathbb{P}(Y \in A, D=d  \vert  Z=1-d) \geq 0$ for all Borel set $A$ and all $d$, we can easily verify that the proposed distributions $\mathbb P(\tilde{Y}_d \leq y_d \mid \tilde{T}=t)$ are well-defined cdf. We can also verify that Equations \eqref{eq:proof1}-\eqref{eq:proof4} hold.   

Now, suppose $\mathbb E[D\mid Z=1]-\mathbb E[D \mid Z=0]=0$. Then, 
\begin{eqnarray*}
&& \max _{s \in\{0,1\}}\left\{\sup _A\{\mathbb{P}(Y \in A, D=s  \vert  Z=1-s)-\mathbb{P}(Y \in A, D=s  \vert  Z=s)\}\right\} \\ 
&& \qquad \qquad \leq 0=\mathbb P(Y\in \mathcal Y, D=1\mid Z=0)-\mathbb P(Y\in \mathcal Y, D=1\mid Z=1).
\end{eqnarray*}
Therefore, 
\begin{eqnarray*}
&& \max _{s \in\{0,1\}}\left\{\sup _A\{\mathbb{P}(Y \in A, D=s  \vert  Z=1-s)-\mathbb{P}(Y \in A, D=s  \vert  Z=s)\}\right\} \\ 
&& \qquad \qquad = \mathbb P(Y\in \mathcal Y, D=1\mid Z=0)-\mathbb P(Y\in \mathcal Y, D=1\mid Z=1)=0,
\end{eqnarray*}
which implies $\mathbb{P}(Y \in A, D=s  \vert  Z=1-s)=\mathbb{P}(Y \in A, D=s  \vert  Z=s)=\mathbb{P}(Y \in A, D=s)$ for all Borel set $A$ and $s$, i.e., $Z \indep (Y,D)$. Define
\begin{eqnarray*}
\tilde{p}_a &=& \mathbb E[D], \ \tilde{p}_n=\mathbb E[1-D], \ \tilde{p}_c=0,\\
\mathbb P(\tilde{Y}_1 \leq y_1\mid \tilde{T}=a) &=& \mathbb P(Y \leq y_1 \mid D=1),\\
\mathbb P(\tilde{Y}_0 \leq y_0\mid \tilde{T}=n) &=& \mathbb P(Y \leq y_0 \mid D=0),\\
\mathbb P(\tilde{Y}_0 \leq y_0\mid \tilde{T}=a) &=& \mathbb P(Y \leq y_0 \mid D=1),\\
\mathbb P(\tilde{Y}_1 \leq y_1\mid \tilde{T}=n) &=& \mathbb P(Y \leq y_1 \mid D=0).
\end{eqnarray*}
We only have two types $a$ and $n$. It is straightforward that Equations \eqref{eq:proof1}-\eqref{eq:proof4} hold. 

    \item \label{subcase2} $\tilde{p}_{df}=\min \{\mathbb{E}[D  \vert  Z=0], \mathbb{E}[1-D  \vert  Z=1]\}$.

    First, suppose $\tilde{p}_{df}=\mathbb E[D \mid Z=0]$. Then $\mathbb{E}[D  \vert  Z=0] \leq \mathbb{E}[1-D  \vert  Z=1]$, and $\tilde{p}_a=0$. Define
\begin{eqnarray*}
\tilde{p}_a &=& 0, \ \tilde{p}_n=\mathbb E[1-D \mid Z=1]-\mathbb E[D \mid Z=0], \ \tilde{p}_c=\mathbb E[D\mid Z=1],\\
\mathbb P(\tilde{Y}_1 \leq y_1\mid \tilde{T}=c) &=& \mathbb P(Y \leq y_1 \mid D=1, Z=1),\\
\mathbb P(\tilde{Y}_1 \leq y_1\mid \tilde{T}=df) &=& \mathbb P(Y \leq y_1 \mid D=1, Z=0),\\
\mathbb P(\tilde{Y}_0 \leq y_0\mid \tilde{T}=n) &=& \mathbb P(Y \leq y_0 \mid D=0, Z=1),\\
\mathbb P(\tilde{Y}_0 \leq y_0\mid \tilde{T}=df) &=& \mathbb P(Y \leq y_0\mid D=0, Z=1),\\
\mathbb P(\tilde{Y}_0 \leq y_0\mid \tilde{T}=c) &=& \frac{\mathbb P(Y \leq y_0,D=0 \vert Z=0)-(\mathbb E[1-D \vert Z=1]-\mathbb E[D \vert Z=0])\mathbb P(Y \leq y_0\vert D=0, Z=1)}{\mathbb E[D\vert Z=1]},\\
\mathbb P(\tilde{Y}_1 \leq y_1\mid \tilde{T}=n) &=& \mathbb P(Y \leq y_1 \mid D=0, Z=1)
\end{eqnarray*}
It is easy to verify that Equations \eqref{eq:proof1}-\eqref{eq:proof4} hold. It is obvious that $\mathbb P(\tilde{Y}_1 \leq y_1\mid \tilde{T}=c)$, $\mathbb P(\tilde{Y}_1 \leq y_1\mid \tilde{T}=df)$, $\mathbb P(\tilde{Y}_0 \leq y_0\mid \tilde{T}=n)$, $\mathbb P(\tilde{Y}_0 \leq y_0\mid \tilde{T}=df)$, and $\mathbb P(\tilde{Y}_1 \leq y_1\mid \tilde{T}=n)$ are well-defined cdfs. We are going to check $\mathbb P(\tilde{Y}_0 \leq y_0\mid \tilde{T}=c)$ is also a well-defined cdf. It is clear that $\lim_{y_0 \rightarrow -\infty}\mathbb P(\tilde{Y}_0 \leq y_0\mid \tilde{T}=c)=0$, and $\lim_{y_0 \rightarrow \infty}\mathbb P(\tilde{Y}_0 \leq y_0\mid \tilde{T}=c)=1$. Right-continuity of $\mathbb P(\tilde{Y}_0 \leq y_0\mid \tilde{T}=c)$ holds by definition. It remains to check that $\mathbb P(\tilde{Y}_0 \leq y_0\mid \tilde{T}=c)$ is nondecreasing. We have
$$\mathbb P(\tilde{Y}_0 \leq y_0\mid \tilde{T}=c) = \frac{\mathbb P(Y \leq y_0,D=0 \vert Z=0)-\mathbb P(Y \leq y_0, D=0 \vert Z=1)+\mathbb E[D \vert Z=0]\mathbb P(Y \leq y_0\vert D=0, Z=1)}{\mathbb E[D\vert Z=1]}.$$
Since $\mathbb{P}(Y \in A, D=d  \vert  Z=d)-\mathbb{P}(Y \in A, D=d  \vert  Z=1-d) \geq 0$ for all Borel set $A$ and all $d$, the function $\mathbb P(Y \leq y_0,D=0 \vert Z=0)-\mathbb P(Y \leq y_0, D=0 \vert Z=1)$ is nondecreasing. The function $\mathbb E[D \vert Z=0]\mathbb P(Y \leq y_0\vert D=0, Z=1)$ is also nondecreasing, because the cdf $\mathbb P(Y \leq y_0\vert D=0, Z=1)$ is. Because the summation of two nondecreasing functions is nondecreasing, we conclude that $\mathbb P(\tilde{Y}_0 \leq y_0\mid \tilde{T}=c)$ as defined above is nondecreasing. 

Now, suppose $\tilde{p}_{df}=\mathbb E[1-D \mid Z=1]$. Then $\mathbb{E}[D  \vert  Z=0] \geq \mathbb{E}[1-D  \vert  Z=1]$, and $\tilde{p}_n=0$. The argument is similar to the one above if we replace $D$ and $Z$ by $\tilde{D}=1-D$ and $\tilde{Z}=1-Z$, respectively. 

    \item $\tilde{p}_{df}=\lambda\times 0 + (1-\lambda) \min \{\mathbb{E}[D  \vert  Z=0], \mathbb{E}[1-D  \vert  Z=1]\}$, where $\lambda \in (0,1)$. As we can see, $\tilde{p}_{df}$ is defined as a mixture of the previous sub-cases \ref{subcase1} and \ref{subcase2}. Hence, It suffices to define the joint distribution based on the corresponding mixture distributions in the previous two points \ref{subcase1} and \ref{subcase2}.
\end{enumerate}

\subsubsection*{case 2: $\max_{s \in\{0,1\}}\left\{\sup _A\{\mathbb{P}(Y \in A, D=s  \vert  Z=1-s)-\mathbb{P}(Y \in A, D=s  \vert  Z=s)\}\right\} > 0$}
In this case, 
\small \begin{equation*}
    \Theta_{I}(p_{df}) = \left\{\begin{array}{lll}
    & \bigg[\max _{s \in\{0,1\}}\left\{\sup _A\{\mathbb{P}(Y \in A, D=s  \vert  Z=1-s)-\mathbb{P}(Y \in A, D=s  \vert  Z=s)\}\right\}, \\
    & \qquad \qquad \qquad \min \{\mathbb{E}[D  \vert  Z=0], \mathbb{E}[1-D  \vert  Z=1]\}\bigg], \text{ if inequality \eqref{test:imp} } \text{ holds } \\
    & \emptyset, \text{ otherwise}.
    \end{array}\right.
    \end{equation*}

\begin{enumerate}
\item \label{subcase4} $\tilde{p}_{df} =\max_{s \in\{0,1\}}\left\{\sup _A\{\mathbb{P}(Y \in A, D=s  \vert  Z=1-s)-\mathbb{P}(Y \in A, D=s  \vert  Z=s)\}\right\}$.

Define 
\begin{eqnarray*}
    \tilde{p}_{df} &=&\max_{s \in\{0,1\}}\left\{\sup _A\{\mathbb{P}(Y \in A, D=s  \vert  Z=1-s)-\mathbb{P}(Y \in A, D=s  \vert  Z=s)\}\right\},\\
    \tilde{p}_a &=& \mathbb E[D \vert Z=0]-\tilde{p}_{df}, \\ \tilde{p}_{c}&=&\mathbb E[D \vert Z=1]-\tilde{p}_a=\mathbb E[D\vert Z=1]-\mathbb E[D \vert Z=0]+\tilde{p}_{df},\\
    \tilde{p}_n &=& \mathbb E[1-D \vert Z=1]-\tilde{p}_{df},\\
    \mathbb P(\tilde{Y}_1 \leq y_1 \mid \tilde{T}=a) &=& \frac{\int^{y_1}_{-\infty}\min\{f_{Y,D\vert Z}(y,1\vert 1),f_{Y,D\vert Z}(y,1\vert 0)\}d\mu(y)}{\int^{\infty}_{-\infty}\min\{f_{Y,D\vert Z}(y,1\vert 1),f_{Y,D\vert Z}(y,1\vert 0)\}d\mu(y)},\\
    \mathbb P(\tilde{Y}_0 \leq y_0 \mid \tilde{T}=n) &=& \frac{\int^{y_0}_{-\infty}\min\{f_{Y,D\vert Z}(y,0\vert 1),f_{Y,D\vert Z}(y,0\vert 0)\}d\mu(y)}{\int^{\infty}_{-\infty}\min\{f_{Y,D\vert Z}(y,0\vert 1),f_{Y,D\vert Z}(y,0\vert 0)\}d\mu(y)},\\
    \mathbb P(\tilde{Y}_1 \leq y_1 \vert \tilde{T}=c) &=& \frac{\mathbb P(Y \leq y_1, D=1 \vert Z=1)-\tilde{p}_a \mathbb P(\tilde{Y}_1 \leq y_1 \mid \tilde{T}=a)}{\tilde{p}_c},\\
    \mathbb P(\tilde{Y}_1 \leq y_1 \vert \tilde{T}=df) &=& \frac{\mathbb P(Y \leq y_1, D=1 \vert Z=0)-\tilde{p}_a \mathbb P(\tilde{Y}_1 \leq y_1 \mid \tilde{T}=a)}{\tilde{p}_{df}},\\
    \mathbb P(\tilde{Y}_0 \leq y_0 \vert \tilde{T}=c) &=& \frac{\mathbb P(Y \leq y_0, D=0 \vert Z=0)-\tilde{p}_n \mathbb P(\tilde{Y}_0 \leq y_0 \mid \tilde{T}=n)}{\tilde{p}_c},\\
    \mathbb P(\tilde{Y}_0 \leq y_0 \vert \tilde{T}=df) &=& \frac{\mathbb P(Y \leq y_0, D=0 \vert Z=1)-\tilde{p}_n \mathbb P(\tilde{Y}_n \leq y_n \mid \tilde{T}=n)}{\tilde{p}_{df}},\\
   \mathbb P(\tilde{Y}_0 \leq y_0 \mid \tilde{T}=a) &=& \frac{\int^{y_0}_{-\infty}\min\{f_{Y,D\vert Z}(y,1\vert 1),f_{Y,D\vert Z}(y,1\vert 0)\}d\mu(y)}{\int^{\infty}_{-\infty}\min\{f_{Y,D\vert Z}(y,1\vert 1),f_{Y,D\vert Z}(y,1\vert 0)\}d\mu(y)},\\
    \mathbb P(\tilde{Y}_1 \leq y_1 \mid \tilde{T}=n) &=& \frac{\int^{y_1}_{-\infty}\min\{f_{Y,D\vert Z}(y,0\vert 1),f_{Y,D\vert Z}(y,0\vert 0)\}d\mu(y)}{\int^{\infty}_{-\infty}\min\{f_{Y,D\vert Z}(y,0\vert 1),f_{Y,D\vert Z}(y,0\vert 0)\}d\mu(y)}. 
\end{eqnarray*}
We need to show that $\mathbb P(\tilde{Y}_1 \leq y_1 \vert \tilde{T}=c)$, $\mathbb P(\tilde{Y}_1 \leq y_1 \vert \tilde{T}=df)$, $\mathbb P(\tilde{Y}_0 \leq y_0 \vert \tilde{T}=c)$, and $\mathbb P(\tilde{Y}_0 \leq y_0 \vert \tilde{T}=df)$ as defined above are well-defined cdfs. To do so, we are going to show that their density versions are nonnegative and integrate to 1.
Let $f_{\tilde{Y}_d\mid \tilde{T}}(y_d\mid t)$ denote the density of $\tilde{Y}_d$ conditional on $\tilde{T}=t$, that is absolutely continuous with respect to the dominating measure $\mu$. We have
\begin{eqnarray*}
    f_{\tilde{Y}_1 \mid \tilde{T}}(y_1\mid a) &=& \frac{\min\{f_{Y,D\vert Z}(y_1,1\vert 1),f_{Y,D\vert Z}(y_1,1\vert 0)\}}{\int^{\infty}_{-\infty}\min\{f_{Y,D\vert Z}(y,1\vert 1),f_{Y,D\vert Z}(y,1\vert 0)\}d\mu(y)},\\
    f_{\tilde{Y}_0 \mid \tilde{T}}(y_0 \mid n) &=& \frac{\min\{f_{Y,D\vert Z}(y_0,0\vert 1),f_{Y,D\vert Z}(y_0,0\vert 0)\}}{\int^{\infty}_{-\infty}\min\{f_{Y,D\vert Z}(y,0\vert 1),f_{Y,D\vert Z}(y,0\vert 0)\}d\mu(y)},\\
    f_{\tilde{Y}_1 \vert \tilde{T}}(y_1\vert c) &=& \frac{f_{Y, D \vert Z}(y_1,1\vert 1)-\tilde{p}_a f_{\tilde{Y}_1 \mid \tilde{T}}(y_1\mid a)}{\tilde{p}_c},\\
    f_{\tilde{Y}_1 \vert \tilde{T}}(y_1 \vert df) &=& \frac{f_{Y, D \vert Z}(y_1,1\vert 0)-\tilde{p}_a f_{\tilde{Y}_1 \mid \tilde{T}}(y_1\mid a)}{\tilde{p}_{df}},\\
    f_{\tilde{Y}_0 \vert \tilde{T}}(y_0 \vert c) &=& \frac{f_{Y, D \vert Z}(y_0,0 \vert 0)-\tilde{p}_n f_{\tilde{Y}_0 \mid \tilde{T}}(y_0 \mid n)}{\tilde{p}_c},\\
    f_{\tilde{Y}_0 \vert \tilde{T}}(y_0 \vert df) &=& \frac{f_{Y, D \vert Z}(y_0,0 \vert 1)-\tilde{p}_n f_{\tilde{Y}_0 \mid \tilde{T}}(y_0 \mid n)}{\tilde{p}_{df}},\\
    f_{\tilde{Y}_0 \mid \tilde{T}}(y_0\mid a) &=& \frac{\min\{f_{Y,D\vert Z}(y_0,1\vert 1),f_{Y,D\vert Z}(y_0,1\vert 0)\}}{\int^{\infty}_{-\infty}\min\{f_{Y,D\vert Z}(y,1\vert 1),f_{Y,D\vert Z}(y,1\vert 0)\}d\mu(y)},\\
    f_{\tilde{Y}_1 \mid \tilde{T}}(y_1 \mid n) &=& \frac{\min\{f_{Y,D\vert Z}(y_1,0\vert 1),f_{Y,D\vert Z}(y_1,0\vert 0)\}}{\int^{\infty}_{-\infty}\min\{f_{Y,D\vert Z}(y,0\vert 1),f_{Y,D\vert Z}(y,0\vert 0)\}d\mu(y)}. 
\end{eqnarray*}
Clearly, $f_{\tilde{Y}_1 \mid \tilde{T}}(y_1\mid a)$, $f_{\tilde{Y}_0 \mid \tilde{T}}(y_0\mid a)$, $f_{\tilde{Y}_0 \mid \tilde{T}}(y_0 \mid n)$, and $f_{\tilde{Y}_1 \mid \tilde{T}}(y_1 \mid n)$ are nonnegative and integrate to 1. To show that $f_{\tilde{Y}_1 \vert \tilde{T}}(y_1\vert c) \geq 0$, it suffices to show that $f_{\tilde{Y}_1 \mid \tilde{T}}(y_1\mid a) \leq \frac{1}{\tilde{p}_a}f_{Y, D \vert Z}(y_1,1\vert 1)$.

We have $\min\{f_{Y,D\vert Z}(y_1,1\vert 1),f_{Y,D\vert Z}(y_1,1\vert 0)\} \leq f_{Y,D\vert Z}(y_1,1\vert 1)$. We just need to show that $\tilde{p}_a \leq \int^{\infty}_{-\infty}\min\{f_{Y,D\vert Z}(y,1\vert 1),f_{Y,D\vert Z}(y,1\vert 0)\}d\mu(y)$. Let $B\equiv \{y\in \mathcal Y: f_{Y,D\vert Z}(y,1\vert 1) \geq f_{Y,D\vert Z}(y,1\vert 0)\}$, and $\overline{B}\equiv \{y\in \mathcal Y: f_{Y,D\vert Z}(y,1\vert 1) < f_{Y,D\vert Z}(y,1\vert 0)\}$. On the one hand, we have
\begin{eqnarray*}
&& \int^{\infty}_{-\infty}\min\{f_{Y,D\vert Z}(y,1\vert 1),f_{Y,D\vert Z}(y,1\vert 0)\}d\mu(y)=\int_{\mathcal Y}\min\{f_{Y,D\vert Z}(y,1\vert 1),f_{Y,D\vert Z}(y,1\vert 0)\}d\mu(y)\\
&& \qquad = \int_{\overline{B}}\min\{f_{Y,D\vert Z}(y,1\vert 1),f_{Y,D\vert Z}(y,1\vert 0)\}d\mu(y) + \int_{B}\min\{f_{Y,D\vert Z}(y,1\vert 1),f_{Y,D\vert Z}(y,1\vert 0)\}d\mu(y)\\
&& \qquad = \int_{\overline{B}}f_{Y,D\vert Z}(y,1\vert 1)d\mu(y) + \int_{B}f_{Y,D\vert Z}(y,1\vert 0)d\mu(y),\\
&& =\mathbb{P}(Y \in \overline{B}, D=1  \vert  Z=1)+\mathbb{P}(Y \in B, D=1  \vert  Z=0)
\end{eqnarray*}
On the other hand, we have 
\begin{eqnarray*}
    \tilde{p}_a &=& \mathbb E[D \vert Z=0]-\max_{s \in\{0,1\}}\left\{\sup _A\{\mathbb{P}(Y \in A, D=s  \vert  Z=1-s)-\mathbb{P}(Y \in A, D=s  \vert  Z=s)\}\right\},\\
    &=& \min_{s \in\{0,1\}}\left\{\inf_A\{\mathbb E[D \vert Z=0]-\mathbb{P}(Y \in A, D=s  \vert  Z=1-s)+\mathbb{P}(Y \in A, D=s  \vert  Z=s)\}\right\},\\
    &\leq& \inf_A\{\mathbb E[D \vert Z=0]-\mathbb{P}(Y \in A, D=1  \vert  Z=0)+\mathbb{P}(Y \in A, D=1  \vert  Z=1)\} \ \ (\text{for } s=1),\\
    &=& \inf_A\{\mathbb{P}(Y \in \overline{A}, D=1  \vert  Z=0)+\mathbb{P}(Y \in A, D=1  \vert  Z=1)\},\ \text{ where } \overline{A} \text{ denotes the complement of } A, \\
    &\leq& \mathbb{P}(Y \in \overline{B}, D=1  \vert  Z=1)+\mathbb{P}(Y \in B, D=1  \vert  Z=0),\\
    &=&  \int^{\infty}_{-\infty}\min\{f_{Y,D\vert Z}(y,1\vert 1),f_{Y,D\vert Z}(y,1\vert 0)\}d\mu(y).    
\end{eqnarray*}
Similarly, we have $f_{\tilde{Y}_1 \mid \tilde{T}}(y_1\mid a) \leq \frac{1}{\tilde{p}_a}f_{Y, D \vert Z}(y_1,1\vert 0)$. Therefore, $f_{\tilde{Y}_1 \vert \tilde{T}}(y_1 \vert df) \geq 0$. To show that $f_{\tilde{Y}_0 \vert \tilde{T}}(y_0 \vert c)$ and $f_{\tilde{Y}_0 \vert \tilde{T}}(y_0 \vert df)$ are nonnegative, we are going to follow a similar reasoning as above. 

Let $C\equiv \{y\in \mathcal Y: f_{Y,D\vert Z}(y,0\vert 0) \geq f_{Y,D\vert Z}(y,0\vert 1)\}$, and $\overline{C}\equiv \{y\in \mathcal Y: f_{Y,D\vert Z}(y,0\vert 0) < f_{Y,D\vert Z}(y,0\vert 1)\}$.

We have
\begin{eqnarray*}
&& \int^{\infty}_{-\infty}\min\{f_{Y,D\vert Z}(y,0\vert 1),f_{Y,D\vert Z}(y,0\vert 0)\}d\mu(y)=\int_{\mathcal Y}\min\{f_{Y,D\vert Z}(y,0\vert 1),f_{Y,D\vert Z}(y,0\vert 0)\}d\mu(y),\\
&& \qquad = \int_{\overline{C}}\min\{f_{Y,D\vert Z}(y,0\vert 1),f_{Y,D\vert Z}(y,0\vert 0)\}d\mu(y) + \int_{C}\min\{f_{Y,D\vert Z}(y,0\vert 1),f_{Y,D\vert Z}(y,0\vert 0)\}d\mu(y),\\
&& \qquad = \int_{\overline{C}}f_{Y,D\vert Z}(y,0\vert 0)d\mu(y) + \int_{C}f_{Y,D\vert Z}(y,0\vert 1)d\mu(y),\\
&& =\mathbb{P}(Y \in \overline{C}, D=0 \vert Z=0)+\mathbb{P}(Y \in C, D=0 \vert Z=1).
\end{eqnarray*}
We also have 
\begin{eqnarray*}
    \tilde{p}_n &=& \mathbb E[1-D \vert Z=1]-\max_{s \in\{0,1\}}\left\{\sup _A\{\mathbb{P}(Y \in A, D=s  \vert  Z=1-s)-\mathbb{P}(Y \in A, D=s  \vert  Z=s)\}\right\},\\
    &=& \min_{s \in\{0,1\}}\left\{\inf_A\{\mathbb E[1-D \vert Z=1]-\mathbb{P}(Y \in A, D=s  \vert  Z=1-s)+\mathbb{P}(Y \in A, D=s  \vert  Z=s)\}\right\},\\
    &\leq& \inf_A\{\mathbb E[1-D \vert Z=1]-\mathbb{P}(Y \in A, D=0  \vert  Z=1)+\mathbb{P}(Y \in A, D=0  \vert  Z=0)\} \ \ (\text{for } s=0),\\
    &=& \inf_A\{\mathbb{P}(Y \in \overline{A}, D=0  \vert  Z=1)+\mathbb{P}(Y \in A, D=0  \vert  Z=0)\}, \\
    &\leq& \mathbb{P}(Y \in \overline{C}, D=0  \vert  Z=0)+\mathbb{P}(Y \in C, D=0 \vert Z=1),\\
    &=&  \int^{\infty}_{-\infty}\min\{f_{Y,D\vert Z}(y,0\vert 1),f_{Y,D\vert Z}(y,0\vert 0)\}d\mu(y).    
\end{eqnarray*}
$\min\{f_{Y,D\vert Z}(y,0\vert 1),f_{Y,D\vert Z}(y,0\vert 0)\} \leq f_{Y,D\vert Z}(y,0\vert 0)$ and $\tilde{p}_n \leq \int^{\infty}_{-\infty}\min\{f_{Y,D\vert Z}(y,0\vert 1),f_{Y,D\vert Z}(y,0\vert 0)\}d\mu(y)$ imply $f_{\tilde{Y}_0 \vert \tilde{T}}(y_0 \vert n) \leq \frac{1}{\tilde{p}_n} f_{Y,D\vert Z}(y,0\vert 0)$, therefore $f_{\tilde{Y}_0 \vert \tilde{T}}(y_0 \vert c) \geq 0$. A similarly argument proves that $f_{\tilde{Y}_0 \vert \tilde{T}}(y_0 \vert df) \geq 0$. 

Finally, given the definition of the types' probabilities $\tilde{p}_t$, it is easy to check that all the defined above densities integrate to 1. We can also check that the above proposed distributions satisfy Equations \eqref{eq:proof1}-\eqref{eq:proof4}.

To complete this step of the proof, we need to check that the DGP we provide requires that inequality~\eqref{test:imp} holds. We have 
\begin{eqnarray*}
    f_{\tilde{Y}_1}(y_1)&=&\sum_{t\in \{a,c,df,n\}} \tilde{p}_t f_{\tilde{Y}_1 \vert \tilde{T}}(y_1 \vert t),\\
    &=& \tilde{p}_a f_{\tilde{Y}_1 \vert \tilde{T}}(y_1\vert a) +f_{Y, D \vert Z}(y_1,1\vert 1)-\tilde{p}_a f_{\tilde{Y}_1 \vert \tilde{T}}(y_1\vert a) + f_{Y, D \vert Z}(y_1,1\vert 0)-\tilde{p}_a f_{\tilde{Y}_1 \vert \tilde{T}}(y_1\vert a)+\tilde{p}_n f_{\tilde{Y}_1 \vert \tilde{T}}(y_1\vert n),\\
    &=& f_{Y, D \vert Z}(y_1,1\vert 1) + (f_{Y, D \vert Z}(y_1,1\vert 0)-\tilde{p}_a f_{\tilde{Y}_1 \mid \tilde{T}}(y_1\mid a))+\tilde{p}_n f_{\tilde{Y}_1 \mid \tilde{T}}(y_1\mid n) \geq f_{Y, D \vert Z}(y_1,1\vert 1),\\
    &=&(f_{Y, D \vert Z}(y_1,1\vert 1)-\tilde{p}_a f_{\tilde{Y}_1 \mid \tilde{T}}(y_1\mid a)) + f_{Y, D \vert Z}(y_1,1\vert 0)+\tilde{p}_n f_{\tilde{Y}_1 \mid \tilde{T}}(y_1\mid n) \geq f_{Y, D \vert Z}(y_1,1\vert 0),\\
    &\geq& \max\left\{f_{Y, D \vert Z}(y_1,1\vert 1),f_{Y, D \vert Z}(y_1,1\vert 0)\right\}=\max_{z\in\{0,1\}}\left\{f_{Y, D \vert Z}(y_1,1\vert z)\right\}. 
\end{eqnarray*}
Therefore, $1=\int_{\mathcal Y} f_{\tilde{Y}_1}(y_1) d\mu(y_1) \geq \int_{\mathcal Y} \max_{z\in\{0,1\}}\left\{f_{Y, D \vert Z}(y_1,1\vert z)\right\} d\mu(y_1)$. 

Similarly, $1=\int_{\mathcal Y} f_{\tilde{Y}_0}(y_0) d\mu(y_0) \geq \int_{\mathcal Y} \max_{z\in\{0,1\}}\left\{f_{Y, D \vert Z}(y_0,0\vert z)\right\} d\mu(y_0)$. These last two inequalities combined are equivalent to inequality \eqref{test:imp}. 

\item \label{subcase5} $\tilde{p}_{df}=\min \{\mathbb{E}[D  \vert  Z=0], \mathbb{E}[1-D  \vert  Z=1]\}$.

First, suppose $\tilde{p}_{df}=\mathbb E[D \mid Z=0]$. Then $\mathbb{E}[D  \vert  Z=0] \leq \mathbb{E}[1-D  \vert  Z=1]$, and $\tilde{p}_a=0$. Define
\begin{eqnarray*}
\tilde{p}_a &=& 0, \ \tilde{p}_n=\mathbb E[1-D \mid Z=1]-\mathbb E[D \mid Z=0], \ \tilde{p}_c=\mathbb E[D\mid Z=1],\\
\mathbb P(\tilde{Y}_1 \leq y_1\mid \tilde{T}=c) &=& \mathbb P(Y \leq y_1 \mid D=1, Z=1),\\
\mathbb P(\tilde{Y}_1 \leq y_1\mid \tilde{T}=df) &=& \mathbb P(Y \leq y_1 \mid D=1, Z=0),\\
\mathbb P(\tilde{Y}_0 \leq y_0 \mid \tilde{T}=n) &=& \frac{\int^{y_0}_{-\infty}\min\{f_{Y,D\vert Z}(y,0\vert 1),f_{Y,D\vert Z}(y,0\vert 0)\}d\mu(y)}{\int^{\infty}_{-\infty}\min\{f_{Y,D\vert Z}(y,0\vert 1),f_{Y,D\vert Z}(y,0\vert 0)\}d\mu(y)},\\
\mathbb P(\tilde{Y}_0 \leq y_0, \vert \tilde{T}=c) &=& \frac{\mathbb P(Y \leq y_0, D=0 \vert Z=0)-\tilde{p}_n \mathbb P(\tilde{Y}_0 \leq y_0 \mid \tilde{T}=n)}{\tilde{p}_c},\\
    \mathbb P(\tilde{Y}_0 \leq y_0 \vert \tilde{T}=df) &=& \frac{\mathbb P(Y \leq y_0, D=0 \vert Z=1)-\tilde{p}_n \mathbb P(\tilde{Y}_0 \leq y_0 \mid \tilde{T}=n)}{\tilde{p}_{df}},\\
    \mathbb P(\tilde{Y}_1 \leq y_1 \mid \tilde{T}=n) &=& \frac{\int^{y_1}_{-\infty}\min\{f_{Y,D\vert Z}(y,0\vert 1),f_{Y,D\vert Z}(y,0\vert 0)\}d\mu(y)}{\int^{\infty}_{-\infty}\min\{f_{Y,D\vert Z}(y,0\vert 1),f_{Y,D\vert Z}(y,0\vert 0)\}d\mu(y)}.
\end{eqnarray*}
It is clear that all the above quantities are well-defined cdfs, except $\mathbb P(\tilde{Y}_0 \leq y_0, \vert \tilde{T}=c)$, and $\mathbb P(\tilde{Y}_0 \leq y_0, \vert \tilde{T}=df)$. As in the previous bullet point, we show that the density analogs are nonnegative and integrate to 1. We have 
\begin{eqnarray*}
     f_{\tilde{Y}_0 \vert \tilde{T}}(y_0 \vert c) &=& \frac{f_{Y, D \vert Z}(y_0,0 \vert 0)-\tilde{p}_n f_{\tilde{Y}_0 \mid \tilde{T}}(y_0 \mid n)}{\tilde{p}_c},\\
    f_{\tilde{Y}_0 \vert \tilde{T}}(y_0 \vert df) &=& \frac{f_{Y, D \vert Z}(y_0,0 \vert 1)-\tilde{p}_n f_{\tilde{Y}_0 \mid \tilde{T}}(y_0 \mid n)}{\tilde{p}_{df}}.
\end{eqnarray*}
As before, it suffices to show $\tilde{p}_n \leq \int^{\infty}_{-\infty}\min\{f_{Y,D\vert Z}(y,0\vert 1),f_{Y,D\vert Z}(y,0\vert 0)\}d\mu(y)$. We have shown above that $$\int^{\infty}_{-\infty}\min\{f_{Y,D\vert Z}(y,0\vert 1),f_{Y,D\vert Z}(y,0\vert 0)\}d\mu(y)=\mathbb{P}(Y \in \overline{C}, D=0 \vert Z=0)+\mathbb{P}(Y \in C, D=0 \vert Z=1).$$ 
Therefore, 
\begin{eqnarray*}
  &&\tilde{p}_n -  \int^{\infty}_{-\infty}\min\{f_{Y,D\vert Z}(y,0\vert 1),f_{Y,D\vert Z}(y,0\vert 0)\}d\mu(y)\\
  && \qquad \mathbb = \mathbb E[1-D \mid Z=1]-\mathbb E[D \mid Z=0] - \mathbb{P}(Y \in \overline{C}, D=0 \vert Z=0)-\mathbb{P}(Y \in C, D=0 \vert Z=1),\\
  && \qquad =(\mathbb E[1-D \mid Z=1]-\mathbb{P}(Y \in C, D=0 \vert Z=1)) + (\mathbb E[1-D \mid Z=0]-\mathbb{P}(Y \in \overline{C}, D=0 \vert Z=0))-1,\\
  && \qquad =\mathbb{P}(Y \in \overline{C}, D=0 \vert Z=1)+\mathbb{P}(Y \in C, D=0 \vert Z=0)-1.
\end{eqnarray*}
Inequality \eqref{test:imp} implies $\int_{\mathcal Y} \max_{z\in \{0,1\}}f_{Y,D\vert Z}(y,0\vert z)d\mu(y) \leq 1$. Then,
\begin{eqnarray*}
    1 \geq \int_{\mathcal Y} \max_{z\in \{0,1\}}f_{Y,D\vert Z}(y,0\vert z)d\mu(y) &=& \int_{\overline{C}} \max_{z\in \{0,1\}}f_{Y,D\vert Z}(y,0\vert z)d\mu(y)+\int_{C} \max_{z\in \{0,1\}}f_{Y,D\vert Z}(y,0\vert z)d\mu(y),\\
    &\geq& \int_{\overline{C}} f_{Y,D\vert Z}(y,0\vert 1)d\mu(y)+\int_{C} f_{Y,D\vert Z}(y,0\vert 0)d\mu(y),\\
    &=& \mathbb{P}(Y \in \overline{C}, D=0 \vert Z=1)+\mathbb{P}(Y \in C, D=0 \vert Z=0).
\end{eqnarray*}
Hence, $\tilde{p}_n -  \int^{\infty}_{-\infty}\min\{f_{Y,D\vert Z}(y,0\vert 1),f_{Y,D\vert Z}(y,0\vert 0)\}d\mu(y)\leq 0$. 

Finally, we can check that the above proposed distributions satisfy Equations \eqref{eq:proof1}-\eqref{eq:proof4}. 

Second, if $\tilde{p}_{df}=\mathbb E[1-D \mid Z=1]$, then $\mathbb{E}[D  \vert  Z=0] \geq \mathbb{E}[1-D  \vert  Z=1]$, and $\tilde{p}_n=0$. The argument is similar to the one above if we replace $D$ and $Z$ by $\tilde{D}=1-D$ and $\tilde{Z}=1-Z$, respectively.  

\item \label{subcase6} $\tilde{p}_{df} =\lambda \max_{s \in\{0,1\}}\left\{\sup _A\{\mathbb{P}(Y \in A, D=s  \vert  Z=1-s)-\mathbb{P}(Y \in A, D=s  \vert  Z=s)\}\right\} + (1-\lambda)\min \{\mathbb{E}[D  \vert  Z=0], \mathbb{E}[1-D  \vert  Z=1]\}$, where $\lambda \in (0,1)$.

As we can see, $\tilde{p}_{df}$ is defined as a mixture of the previous sub-cases \ref{subcase4} and \ref{subcase5}. Hence, it suffices to define the joint distribution based on the corresponding mixture distributions in the previous two points \ref{subcase4} and \ref{subcase5}.

\end{enumerate}

Finally, to finish the proof of Proposition \ref{prop_pdf}, we show that the following identified set 
\begin{equation} \label{eq:id_bound_df}
    \begin{aligned}
      \Theta_I(p_{df}) =  & \bigg[\max \left\{\max _{s \in\{0,1\}}\left\{\sup _A\{\mathbb{P}(Y \in A, D=s  \vert  Z=1-s)-\mathbb{P}(Y \in A, D=s  \vert  Z=s)\}\right\}, 0\right\}, \\
        & \qquad \qquad \qquad \min \{\mathbb{E}[D  \vert  Z=0], \mathbb{E}[1-D  \vert  Z=1]\}\bigg]
    \end{aligned}
\end{equation}
is empty if and only if Equation \eqref{test:imp} does not hold.

When the instrument and the treatment are binary, $Z \in \{0, 1\}$ and $D \in \{0, 1\}$, we can write Equation \eqref{test:imp} as 
\begin{equation} \label{eq:test_implication}
    \max _{d \in\{0,1\}} \int_{\mathcal{Y}} \max _{z \in \{0,1\}} f_{Y, D \mid Z}(y, d \mid z) d \mu(y) \leq 1.
\end{equation}

Let $B_1$ denote the Borel set on $\mathcal{Y}$ where $f_{Y, D \mid Z}(y, 1 \mid 0)$ is no smaller than $f_{Y, D \mid Z}(y, 1 \mid 1)$, $B_1 = \{y: f_{Y, D \mid Z}(y, 1 \mid 0) \geq f_{Y, D \mid Z}(y, 1 \mid 1)\}$, and $B_0$ denote the Borel set on $\mathcal{Y}$ where $f_{Y, D \mid Z}(y, 0 \mid 1)$ is no smaller than $f_{Y, D \mid Z}(y, 0 \mid 0)$, $B_0 = \{y: f_{Y, D \mid Z}(y, 0 \mid 1) \geq f_{Y, D \mid Z}(y, 0 \mid 0)\}$. The notation $B^c$ represents the complement set of $B$. Given $d = 1$, 
\begin{equation*}
    \begin{aligned}
        &\int_{\mathcal{Y}} \max _{z \in \{0,1\}} f_{Y, D \mid Z}(y, 1 \mid z) d \mu(y) \\
        =& \int_{B_1} f_{Y, D \mid Z}(y, 1 \mid 0) d \mu(y) + \int_{B_1^C} f_{Y, D \mid Z}(y, 1 \mid 1) d \mu(y) \\
        =& \mathbb{P}\left(Y \in B_1, D = 1 \mid Z = 0\right) - \mathbb{P}\left(Y \in B_1, D = 1 \mid Z = 1\right) + \mathbb{P}\left(D = 1 \mid Z = 1\right),
    \end{aligned}
\end{equation*}
and given $d = 0$,
\begin{equation*}
    \begin{aligned}
        &\int_{\mathcal{Y}} \max _{z \in \{0,1\}} f_{Y, D \mid Z}(y, 0 \mid z) d \mu(y) \\
        =& \int_{B_0} f_{Y, D \mid Z}(y, 0 \mid 1) d \mu(y) + \int_{B_0^C} f_{Y, D \mid Z}(y, 0 \mid 0) d \mu(y) \\
        =& \mathbb{P}\left(Y \in B_0, D = 0 \mid Z = 1\right) - \mathbb{P}\left(Y \in B_0, D = 0 \mid Z = 0\right) + \mathbb{P}\left(D = 0 \mid Z = 0\right).
    \end{aligned}
\end{equation*}

Therefore, the violation of Equation \eqref{eq:test_implication} can be represented as 
\begin{equation} \label{eq:cond1}
    \begin{aligned}
        \max& \big\{ \mathbb{P}\left(Y \in B_1, D = 1 \mid Z = 0\right) - \mathbb{P}\left(Y \in B_1, D = 1 \mid Z = 1\right) + \mathbb{P}\left(D = 1 \mid Z = 1\right) , \\
        & \mathbb{P}\left(Y \in B_0, D = 0 \mid Z = 1\right) - \mathbb{P}\left(Y \in B_0, D = 0 \mid Z = 0\right) + \mathbb{P}\left(D = 0 \mid Z = 0\right) \big\} > 1 \\
        \Longleftrightarrow & \max _{s \in \{0, 1\}} \Big\{\mathbb{P} \left(Y \in B_s, D = s \mid Z = 1-s\right) - \mathbb{P}\left(Y \in B_s, D = s \mid Z = s\right) - \mathbb{P}\left(D = 1-s \mid Z = s\right)\Big\} > 0.
    \end{aligned}
\end{equation}

In Equation \eqref{eq:id_bound_df}, since the empty set $\emptyset$ is in the Borel set generated by $Y$,
\begin{equation*}
    \begin{aligned}
        &\sup _A\{\mathbb{P}(Y \in A, D=s \mid Z=1-s)-\mathbb{P}(Y \in A, D=s \mid Z=s)\} \\
        \geq& \mathbb{P}(Y \in \emptyset, D=s \mid Z=1-s)-\mathbb{P}(Y \in \emptyset, D=s \mid Z=s) = 0,
    \end{aligned}
\end{equation*}
for any $s \in \{0,1\}$. Then, the identified set $\Theta_I(p_{df})$ in Equation \eqref{eq:id_bound_df} would be empty when 
\begin{equation*}
    \begin{aligned}
        &\min \{\mathbb{E}[D \mid Z=0], \mathbb{E}[1-D \mid Z=1]\} \\
        &< \max _{s \in\{0,1\}}\left\{\sup _A\{\mathbb{P}(Y \in A, D=s \mid Z=1-s)-\mathbb{P}(Y \in A, D=s \mid Z=s)\}\right\}.
    \end{aligned}
\end{equation*}
The above inequality is equivalent to 
\begin{equation} \label{eq:cond2}
    \begin{aligned}
        \max \Big\{&\sup_A \left\{\mathbb{P}\left(Y \in A, D = 1 \mid Z = 0\right) - \mathbb{P}\left(Y \in A, D = 1 \mid Z = 1\right) - \mathbb{P}\left(D = 0 \mid Z = 1\right)\right\}, \\
        &\sup_A \left\{\mathbb{P}\left(Y \in A, D = 1 \mid Z = 0\right) - \mathbb{P}\left(Y \in A, D = 1 \mid Z = 1\right) - \mathbb{P}\left(D = 1 \mid Z = 0\right)\right\}, \\
        &\sup_A \left\{\mathbb{P}\left(Y \in A, D = 0 \mid Z = 1\right) - \mathbb{P}\left(Y \in A, D = 0 \mid Z = 0\right) - \mathbb{P}\left(D = 1 \mid Z = 0\right)\right\}, \\
        &\sup_A \left\{\mathbb{P}\left(Y \in A, D = 0 \mid Z = 1\right) - \mathbb{P}\left(Y \in A, D = 0 \mid Z = 0\right) - \mathbb{P}\left(D = 0 \mid Z = 1\right)\right\} \Big\} >0.
    \end{aligned}
\end{equation}

If the inequality in Equation \eqref{eq:cond1} holds, it implies that either 
\begin{equation*}
    \sup_A \left\{\mathbb{P}\left(Y \in A, D = 1 \mid Z = 0\right) - \mathbb{P}\left(Y \in A, D = 1 \mid Z = 1\right) - \mathbb{P}\left(D = 0 \mid Z = 1\right)\right\} > 0
\end{equation*}
or 
\begin{equation*}
    \sup_A \left\{\mathbb{P}\left(Y \in A, D = 0 \mid Z = 1\right) - \mathbb{P}\left(Y \in A, D = 0 \mid Z = 0\right) - \mathbb{P}\left(D = 1 \mid Z = 0\right)\right\} > 0,
\end{equation*}
and the inequality in Equation \eqref{eq:cond2} holds. Therefore, the violation of Equation \eqref{eq:test_implication} implies that the identified set in Equation \eqref{eq:id_bound_df} is empty.

Suppose now that the inequality in \eqref{eq:cond2} holds so that the identified set in Equation \eqref{eq:id_bound_df} is empty. Since 
\begin{equation*}
    \begin{aligned}
        &\sup_A \left\{\mathbb{P}\left(Y \in A, D = 1 \mid Z = 0\right) - \mathbb{P}\left(Y \in A, D = 1 \mid Z = 1\right) - \mathbb{P}\left(D = 1 \mid Z = 0\right)\right\} \\
        \leq&\sup_A \left\{\mathbb{P}\left(Y \in A, D = 1 \mid Z = 0\right) - \mathbb{P}\left(D = 1 \mid Z = 0\right)\right\} = 0, \\
        &\sup_A \left\{\mathbb{P}\left(Y \in A, D = 0 \mid Z = 1\right) - \mathbb{P}\left(Y \in A, D = 0 \mid Z = 0\right) - \mathbb{P}\left(D = 0 \mid Z = 1\right)\right\} \\
        \leq&\sup_A \left\{\mathbb{P}\left(Y \in A, D = 0 \mid Z = 1\right) - \mathbb{P}\left(D = 0 \mid Z = 1\right)\right\} = 0,
    \end{aligned}
\end{equation*}
it should be that either the first or the third expression in Equation \eqref{eq:cond2} is greater than zero. Assume that the maximum in Equation \eqref{eq:cond2} is achieved by the first expression. Based on the definition of $B_1$, the supremum is achieved when $A = B_1$. Then, we have 
\begin{equation*}
    \mathbb{P} \left(Y \in B_1, D = 1 \mid Z = 0\right) - \mathbb{P}\left(Y \in B_1, D = 1 \mid Z = 1\right) - \mathbb{P}\left(D = 0 \mid Z = 1\right) > 0,
\end{equation*}
which implies that the inequality in \eqref{eq:cond1} holds. Following similar reasoning, if the third term in Equation \eqref{eq:cond2} attains its maximum, this would also imply that the inequality in \eqref{eq:cond1} holds. Therefore, an empty identified set for the defier's proportion in Equation \eqref{eq:id_bound_df} would imply the failure of the testable implication in Equation \eqref{eq:test_implication}. 

To conclude, the violation of testable implication in Equation \eqref{eq:test_implication} is equivalent to the identified set $\Theta_I(p_{df})$ in Equation \eqref{eq:id_bound_df} being empty. Consequently, we can write the identified set for defier's probability as
\begin{equation*}
    \begin{aligned}
        \Theta_{I}(p_{df}) = & \bigg[\max _{s \in\{0,1\}}\left\{\sup _A\{\mathbb{P}(Y \in A, D=s  \vert  Z=1-s)-\mathbb{P}(Y \in A, D=s  \vert  Z=s)\}\right\}, \\
        & \qquad \qquad \qquad \min \{\mathbb{E}[D  \vert  Z=0], \mathbb{E}[1-D  \vert  Z=1]\}\bigg].
    \end{aligned}
\end{equation*}
\end{proof}

\subsection{Proof of Proposition \ref{prop:distbound} and Theorem \ref{thm1}} \label{app:proof_prop2_thm1}

\begin{proof}

For each $p_{df} \in \Theta_I(p_{df})$, we derive in bounds on $\mu_{1a}(p_a)$, $\mu_{1c}(p_a)$,$\mu_{1df}(p_a)$, $\mu_{0n}(p_n)$, $\mu_{0c}(p_n)$, $\mu_{0df}(p_n)$ where $p_a=\mathbb{E}[D \mid Z=0]-p_{d f}$, and $p_n =\mathbb{E}[1-D \mid Z=1]-p_{d f}$. We build on results from \cite{Horowitz1995IdentificationData}. 


First, suppose $0 < p_{df} < \min\{\mathbb E[D\vert Z=0], \mathbb E[1-D \vert Z=1]\}$.
Equation \eqref{important3} implies 
\begin{eqnarray*}
    \mathbb P(Y_1 \in A \vert T=a) &=& \frac{\mathbb P(Y\in A, D=1 \vert Z=1)- p_c \mathbb P(Y_1 \in A \vert T=c)}{p_a}.
\end{eqnarray*}
Since $\mathbb P(Y_1 \in A \vert T=c) \in [0,1],$ and $\mathbb P(Y_1 \in A \vert T=a) \in [0,1]$, we have
\begin{eqnarray*}
   \max\left\{\frac{\mathbb P(Y\in A, D=1 \vert Z=1)- p_c}{p_a},0\right\} \leq \mathbb P(Y_1 \in A \vert T=a) \leq \min\left\{\frac{\mathbb P(Y\in A, D=1 \vert Z=1)}{p_a},1\right\},  
   \end{eqnarray*}
which we can equivalently rewrite as 
\begin{eqnarray*}
   && \frac{(\mathbb P(Y\in A, D=1 \vert Z=1)- p_c)\mathbbm{1}\{\mathbb P(Y\in A, D=1 \vert Z=1)- p_c\geq 0\}}{p_a}\\
   && \qquad \qquad \leq \mathbb P(Y_1 \in A \vert T=a) \leq \\
   && \frac{\mathbb P(Y\in A, D=1 \vert Z=1)}{p_a}\mathbbm{1}\left\{\frac{\mathbb P(Y\in A, D=1 \vert Z=1)}{p_a} \leq 1\right\}+\mathbbm{1}\left\{\frac{\mathbb P(Y\in A, D=1 \vert Z=1)}{p_a} > 1\right\}.  
\end{eqnarray*}

Similarly, Equation \eqref{important4} implies
\begin{eqnarray*}
   \max\left\{\frac{\mathbb P(Y\in A, D=1 \vert Z=0)- p_{df}}{p_a},0\right\} \leq \mathbb P(Y_1 \in A \vert T=a) \leq \min\left\{\frac{\mathbb P(Y\in A, D=1 \vert Z=0)}{p_a},1\right\},   
   \end{eqnarray*}
or equivalently
\begin{eqnarray*}
   && \frac{(\mathbb P(Y\in A, D=1 \vert Z=0)- p_{df})\mathbbm{1}\{\mathbb P(Y\in A, D=1 \vert Z=0)- p_{df}\geq 0\}}{p_a}\\
   && \qquad \qquad \leq \mathbb P(Y_1 \in A \vert T=a) \leq \\
   && \frac{\mathbb P(Y\in A, D=1 \vert Z=0)}{p_a}\mathbbm{1}\left\{\frac{\mathbb P(Y\in A, D=1 \vert Z=0)}{p_a} \leq 1\right\}+\mathbbm{1}\left\{\frac{\mathbb P(Y\in A, D=1 \vert Z=0)}{p_a} > 1\right\}.  
\end{eqnarray*}
   
   Hence, Equations \eqref{important3}-\eqref{important4} together imply 
   \begin{eqnarray*}
   && \max\left\{\frac{\mathbb P(Y\in A, D=1 \vert Z=1)- p_{c}}{p_a}, \frac{\mathbb P(Y\in A, D=1 \vert Z=0)- p_{df}}{p_a},0\right\}\\
   && \qquad \qquad \leq \mathbb P(Y_1 \in A \vert T=a) \leq \\
   && \min\left\{\frac{\mathbb P(Y\in A, D=1 \vert Z=1)}{p_a}, \frac{\mathbb P(Y\in A, D=1 \vert Z=0)}{p_a},1\right\}.   
   \end{eqnarray*}
Therefore, the identified set for the distribution $F_{1a}\equiv F_{Y_1\vert T=a}$ is
\begin{eqnarray}
    \Theta_I(F_{1a})&=&\Bigg\{
    F_{1a} \in \mathcal F: \max\bigg\{ \frac{(\mathbb P(Y\in (y,y'], D=1 \vert Z=1)- p_c)\mathbbm{1}\{\mathbb P(Y\in (y,y'], D=1 \vert Z=1)- p_c\geq 0\}}{p_a},\nonumber\\
   && \qquad \frac{(\mathbb P(Y\in (y,y'], D=1 \vert Z=0)- p_{df})\mathbbm{1}\{\mathbb P(Y\in (y,y'], D=1 \vert Z=0)- p_{df}\geq 0\}}{p_a}\bigg\}\nonumber \\
   && \qquad \leq F_{1a}(y')-F_{1a}(y) \leq \\
   && \min\bigg\{ \frac{\mathbb P(Y\in (y,y'], D=1 \vert Z=1)}{p_a}\mathbbm{1}\left\{\frac{\mathbb P(Y\in (y,y'], D=1 \vert Z=1)}{p_a} \leq 1\right\}\nonumber \\
   && \qquad \qquad +\mathbbm{1}\left\{\frac{\mathbb P(Y\in (y,y'], D=1 \vert Z=1)}{p_a} > 1\right\},\nonumber\\
   && \qquad \frac{\mathbb P(Y\in (y,y'], D=1 \vert Z=0)}{p_a}\mathbbm{1}\left\{\frac{\mathbb P(Y\in (y,y'], D=1 \vert Z=0)}{p_a} \leq 1\right\}\nonumber \\
   && \qquad \qquad +\mathbbm{1}\left\{\frac{\mathbb P(Y\in (y,y'], D=1 \vert Z=0)}{p_a} > 1\right\} \bigg\} \text{ for all } y<y'    
   \Bigg\},\nonumber
    \end{eqnarray}
   where $\mathcal F$ is the set of all distributions. 

Combining Equations \eqref{important}-\eqref{important2}, we obtain the following identified set for $F_{0n}\equiv F_{Y_0\vert T=n}$:
\begin{eqnarray}
    \Theta_I(F_{0n})&=&\Bigg\{
    F_{0n} \in \mathcal F: \max\bigg\{ \frac{(\mathbb P(Y\in (y,y'], D=0 \vert Z=0)- p_c)\mathbbm{1}\{\mathbb P(Y\in (y,y'], D=0 \vert Z=0)- p_c\geq 0\}}{p_n},\nonumber\\
   && \qquad \frac{(\mathbb P(Y\in (y,y'], D=0 \vert Z=1)- p_{df})\mathbbm{1}\{\mathbb P(Y\in (y,y'], D=0 \vert Z=1)- p_{df}\geq 0\}}{p_n}\bigg\}\nonumber \\
   && \qquad \leq F_{0n}(y')-F_{0n}(y) \leq \\
   && \min\bigg\{ \frac{\mathbb P(Y\in (y,y'], D=0 \vert Z=0)}{p_n}\mathbbm{1}\left\{\frac{\mathbb P(Y\in (y,y'], D=0 \vert Z=0)}{p_n} \leq 1\right\}\nonumber \\
   && \qquad \qquad +\mathbbm{1}\left\{\frac{\mathbb P(Y\in (y,y'], D=0 \vert Z=0)}{p_n} > 1\right\},\nonumber\\
   && \qquad \frac{\mathbb P(Y\in (y,y'], D=0 \vert Z=1)}{p_n}\mathbbm{1}\left\{\frac{\mathbb P(Y\in (y,y'], D=0 \vert Z=1)}{p_n} \leq 1\right\}\nonumber \\
   && \qquad \qquad +\mathbbm{1}\left\{\frac{\mathbb P(Y\in (y,y'], D=0 \vert Z=1)}{p_n} > 1\right\} \bigg\} \text{ for all } y<y'    
   \Bigg\}.\nonumber
    \end{eqnarray}
Note that the above bounds are \textit{functionally} sharp, as discussed in \cite{mourifie2020sharp}. 
The pointwise bounds are given below:
\begin{eqnarray}
   &&F_{1a}^{LB}(y)\equiv \max\left\{F_{1a}^{LB_1}(y),F_{1a}^{LB_0}(y)\right\} \leq F_{1a}(y) \leq \min\left\{F_{1a}^{UB_1}(y),F_{1a}^{UB_0}(y)\right\}\equiv F_{1a}^{UB}(y),
\end{eqnarray}
where
\begin{eqnarray*}
   &&F_{1a}^{LB_1}(y)\equiv \frac{(\mathbb P(Y\leq y, D=1 \vert Z=1)- p_c)\mathbbm{1}\{\mathbb P(Y\leq y, D=1 \vert Z=1)- p_c\geq 0\}}{p_a},\\
   &&F_{1a}^{LB_0}(y)\equiv \frac{(\mathbb P(Y\leq y, D=1 \vert Z=0)- p_{df})\mathbbm{1}\{\mathbb P(Y\leq y, D=1 \vert Z=0)- p_{df}\geq 0\}}{p_a}, \\
   &&F_{1a}^{UB_1}(y)\equiv \frac{\mathbb P(Y\leq y, D=1 \vert Z=1)}{p_a}\mathbbm{1}\left\{\frac{\mathbb P(Y\leq y, D=1 \vert Z=1)}{p_a} \leq 1\right\}+\mathbbm{1}\left\{\frac{\mathbb P(Y\leq y, D=1 \vert Z=1)}{p_a} > 1\right\},\\
   && F_{1a}^{UB_0}(y) \equiv \frac{\mathbb P(Y\leq y, D=1 \vert Z=0)}{p_a}\mathbbm{1}\left\{\frac{\mathbb P(Y\leq y, D=1 \vert Z=0)}{p_a} \leq 1\right\} + \mathbbm{1}\left\{\frac{\mathbb P(Y\leq y, D=1 \vert Z=0)}{p_a} > 1\right\}.
    \end{eqnarray*}
    
\begin{eqnarray}
   &&F_{0n}^{LB}(y)\equiv \max\left\{F_{0n}^{LB_1}(y),F_{0n}^{LB_0}(y)\right\} \leq F_{0n}(y) \leq \min\left\{F_{0n}^{UB_1}(y),F_{0n}^{UB_0}(y)\right\}\equiv F_{0n}^{UB}(y),
\end{eqnarray}
where
\begin{eqnarray*}
    && F_{0n}^{LB_0}(y) \equiv \frac{(\mathbb P(Y\leq y, D=0 \vert Z=0)- p_c)\mathbbm{1}\{\mathbb P(Y \leq y, D=0 \vert Z=0)- p_c\geq 0\}}{p_n},\\
   && F_{0n}^{LB_1}(y) \equiv \frac{(\mathbb P(Y\leq y, D=0 \vert Z=1)- p_{df})\mathbbm{1}\{\mathbb P(Y\leq y, D=0 \vert Z=1)- p_{df}\geq 0\}}{p_n},\\
   && F_{0n}^{UB_0}(y) \equiv \frac{\mathbb P(Y \leq y, D=0 \vert Z=0)}{p_n}\mathbbm{1}\left\{\frac{\mathbb P(Y \leq y, D=0 \vert Z=0)}{p_n} \leq 1\right\} + \mathbbm{1}\left\{\frac{\mathbb P(Y \leq y, D=0 \vert Z=0)}{p_n} > 1\right\},\\
   && F_{0n}^{UB_1}(y) \equiv \frac{\mathbb P(Y \leq y, D=0 \vert Z=1)}{p_n}\mathbbm{1}\left\{\frac{\mathbb P(Y\leq y, D=0 \vert Z=1)}{p_n} \leq 1\right\} + \mathbbm{1}\left\{\frac{\mathbb P(Y\leq y, D=0 \vert Z=1)}{p_n} > 1\right\}.
    \end{eqnarray*}
The identified sets for $F_{dt}\equiv F_{Y_d\vert T=t}$, $d\in\{0,1\}$, $t\in \{c,df\}$ are given by
  \begin{eqnarray*}
      F_{1c}(y) &=& \frac{\mathbb P(Y \leq y, D=1 \vert Z=1)-p_a F_{1a}(y)}{p_c },\\
    F_{1df}(y) &=& \frac{\mathbb P(Y \leq y, D=1 \vert Z=0)-p_a F_{1a}(y)}{p_{df} },\\
    F_{0c}(y) &=& \frac{\mathbb P(Y \leq y, D=0 \vert Z=0)-p_n F_{0n}(y)}{p_c },\\
    F_{0df}(y) &=& \frac{\mathbb P(Y \leq y, D=0 \vert Z=1)-p_n F_{0n}(y)}{p_{df}},
  \end{eqnarray*}  
where $F_{1a} \in \Theta_I(F_{1a})$, and $F_{0n} \in \Theta_I(F_{0n})$. 

Let $\mu_{F}$ denote the expected value of a given cdf $F$. Then sharp bounds for $\mu_{1a}$ and $\mu_{0n}$ are as follows: 
\begin{eqnarray*}
    \mu_{F_{1a}^{UB}} &\leq& \mu_{1a} \leq \mu_{F_{1a}^{LB}}\\
    \mu_{F_{0n}^{UB}} &\leq& \mu_{0n} \leq \mu_{F_{0n}^{LB}}.
\end{eqnarray*}

Once obtaining the sharp bounds on $\mu_{1a}, \mu_{0n}$, we can characterize the sharp bounds for $\mu_{dt}, d \in \{0, 1\}, t \in \{c, df\}$ as
\begin{eqnarray*}
\begin{aligned}
& \mu_{1 c} = \frac{\mathbb{E}[Y D  \vert  Z=1]-p_a \mu_{1 a}}{\mathbb{E}[D  \vert  Z=1]-p_a},\ \ \ \mu_{1 d f} = \frac{\mathbb{E}[Y D  \vert  Z=0]-p_a \mu_{1 a}}{\mathbb{E}[D  \vert  Z=0]-p_a}, \\
& \mu_{0 c} = \frac{\mathbb{E}[Y(1-D)  \vert  Z=0]-p_n \mu_{0 n}}{\mathbb{E}[1-D  \vert  Z=0]-p_n},\ \ \ \mu_{0 d f} = \frac{\mathbb{E}[Y(1-D)  \vert  Z=1]-p_n \mu_{0 n}}{\mathbb{E}[1-D  \vert  Z=1]-p_n}.
\end{aligned}
\end{eqnarray*}

Finally, sharp bounds for parameters $\theta_{0c}$, $\theta_{1c}$, $\theta_{0df}$, $\theta_{1df}$ are given by
\begin{eqnarray*}
    \theta_{0c}=\theta_{1c}&=& \mu_{1c}-\mu_{0c},\ \ \ \theta_{0df}=\theta_{1df}(p_{df})= \mu_{1df}-\mu_{0df}.
\end{eqnarray*}

Therefore, when $p_{df} \in \mathring{\Theta}_I(p_{df})$, we can characterize the identified set for $\Gamma$ with $\Theta_I^1(A_2)$.

These sharp bounds can be difficult to compute in practice. Valid outer sets can be computed:
\begin{eqnarray*}
    \max\left\{\mu_{F_{1a}^{UB_1}}, \mu_{F_{1a}^{UB_0}}\right\} &\leq& \mu_{1a} \leq \min\left\{\mu_{F_{1a}^{LB_1}}, \mu_{F_{1a}^{LB_0}}\right\}\\
    \max\left\{\mu_{F_{0n}^{UB_1}}, \mu_{F_{0n}^{UB_0}}\right\} &\leq& \mu_{0n} \leq \min\left\{\mu_{F_{0n}^{LB_1}}, \mu_{F_{0n}^{LB_0}}\right\}.
\end{eqnarray*}

\begin{remark}
    $\max\left\{\mu_{F_{1a}^{UB_1}}, \mu_{F_{1a}^{UB_0}}\right\}=\mu_{F_{1a}^{UB}}$ if $F_{1a}^{UB_1}$ first-order stochastically dominates $F_{1a}^{UB_0}$ or vice versa. Similarly, $\min\left\{\mu_{F_{1a}^{LB_1}}, \mu_{F_{1a}^{LB_0}}\right\}=\mu_{F_{1a}^{LB}}$ if $F_{1a}^{LB_1}$ first-order stochastically dominates $F_{1a}^{LB_0}$ or vice versa. However, in general, 
    $\max\left\{\mu_{F_{1a}^{UB_1}}, \mu_{F_{1a}^{UB_0}}\right\} \leq \mu_{F_{1a}^{UB}}$, and $\min\left\{\mu_{F_{1a}^{LB_1}}, \mu_{F_{1a}^{LB_0}}\right\}\geq \mu_{F_{1a}^{LB}}$. 
\end{remark}

\subsubsection{Proof of Lemma \ref{lem:lee2009}}

\begin{proof}
    To compute the expected value of $F_{1a}^{LB_z}$, we only need to identify the support and the density of a random variable with such cdf. Take $\mu_{F_{1a}^{LB_1}}$ as the example. Since the outcome is continuous with a strictly increasing cdf, the support is characterized by $\{y: 0<F_{1a}^{LB_1}(y) < 1\}$, which is this case equal to $\{y: \mathbb P(Y \leq y, D=1 \vert Z=1)- p_c > 0\}=\left\{y: \mathbb P(Y \leq y\vert D=1, Z=1)> \frac{p_c}{\mathbb E[D\vert Z=1]}\right\}=\left\{y: y> F^{-1}_{Y\vert D=1,Z=1}\left(\frac{p_c}{\mathbb E[D\vert Z=1]}\right)\right\}$. Then for all $y> F^{-1}_{Y\vert D=1,Z=1}\left(\frac{p_c}{\mathbb E[D\vert Z=1]}\right)$, $F_{1a}^{LB_1}(y)=\frac{\mathbb P(Y \leq y, D=1 \vert Z=1)- p_c}{p_a}$ and its density $f_{1a}^{LB_1}(y)=\partial F_{1a}^{LB_1}(y)/\partial y=\frac{f_{Y,D\vert Z}(y,1\vert 1)}{p_a}$. Therefore,
    \begin{eqnarray*}
       \mu_{F_{1a}^{LB_1}} &=& \int_{y> F^{-1}_{Y\vert D=1,Z=1}\left(\frac{p_c}{\mathbb E[D\vert Z=1]}\right)} y \frac{f_{Y,D\vert Z}(y,1\vert 1)}{p_a} dy,\\
       &=& \int_{y> F^{-1}_{Y\vert D=1,Z=1}\left(1-\frac{p_a}{\mathbb E[D\vert Z=1]}\right)} y \frac{f_{Y,D\vert Z}(y,1\vert 1)}{p_a} dy\ \ \text{ since } p_c=\mathbb E[D\vert Z=1]-p_a,\\
       &=& \int_{y> F^{-1}_{Y\vert D=1,Z=1}\left(1-\frac{p_a}{\mathbb E[D\vert Z=1]}\right)} y \frac{f_{Y\vert D, Z}(y\vert 1, 1)}{p_a/\mathbb E[D\vert Z=1]} dy\ \ \text{ from Bayes' rule}\\
       &=& \int_{y> F^{-1}_{Y\vert D=1,Z=1}\left(1-\frac{p_a}{\mathbb E[D\vert Z=1]}\right)} y \frac{f_{Y\vert D, Z}(y\vert 1, 1)}{\mathbb P\left(Y> F^{-1}_{Y\vert D=1,Z=1}\left(1-\frac{p_a}{\mathbb E[D\vert Z=1]}\right) \vert D = 1, Z = 1\right)} dy,\\ 
       && \qquad \qquad \text{ because } \mathbb P\left(Y> F^{-1}_{Y\vert D=1,Z=1}\left(1-\frac{p_a}{\mathbb E[D\vert Z=1]}\right) \vert D = 1, Z = 1\right)=p_a/\mathbb E[D\vert Z=1],\\
       &=& \mathbb E\left[Y\vert D=1, Z=1, Y > F^{-1}_{Y\vert D=1, Z=1}\left(1-\frac{p_a}{\mathbb E[D\vert Z=1]}\right)\right]. 
    \end{eqnarray*}

    A similar argument for $\mu_{F_{1a}^{LB_0}}, \mu_{F_{0n}^{LB_z}}$ yields the formula in the lemma.
    
    For $\mu_{F_{1a}^{UB_z}}$, the support is $\{y: \frac{\mathbb P(Y\leq y, D=1 \vert Z=z)}{p_a} < 1\}=\{y: \mathbb P(Y \leq y \vert D=1, Z=z) < \frac{p_a}{\mathbb E[D\vert Z=z]}\}=\left\{y: y< F^{-1}_{Y\vert D=1,Z=z}\left(\frac{p_a}{\mathbb E[D\vert Z=z]}\right)\right\}$. Also, the density of $F_{1a}^{UB_z}$ is $f_{1a}^{UB_z}(y)=\partial F_{1a}^{UB_z}(y)/\partial y=\frac{f_{Y,D\vert Z}(y,1\vert z)}{p_a}$. Therefore,
    \begin{eqnarray*}
       \mu_{F_{1a}^{UB_z}} &=& \int_{y< F^{-1}_{Y\vert D=1,Z=z}\left(\frac{p_a}{\mathbb E[D\vert Z=z]}\right)} y \frac{f_{Y,D\vert Z}(y,1\vert z)}{p_a} dy,\\
       &=& \int_{y< F^{-1}_{Y\vert D=1,Z=z}\left(\frac{p_a}{\mathbb E[D\vert Z=z]}\right)} y \frac{f_{Y,D\vert Z}(y,1\vert z)}{p_a} dy,\\
       &=& \int_{y< F^{-1}_{Y\vert D=1,Z=z}\left(\frac{p_a}{\mathbb E[D\vert Z=z]}\right)} y \frac{f_{Y\vert D, Z}(y\vert 1, z)}{p_a/\mathbb E[D\vert Z=z]} dy\ \ \text{ from Bayes' rule}\\
       &=& \int_{y < F^{-1}_{Y\vert D=1,Z=z}\left(\frac{p_a}{\mathbb E[D\vert Z=z]}\right)} y \frac{f_{Y\vert D, Z}(y\vert 1, z)}{\mathbb P\left(Y < F^{-1}_{Y\vert D=1,Z=z}\left(\frac{p_a}{\mathbb E[D\vert Z=z]}\right) \vert D = 1, Z = z\right)} dy,\\ 
       && \qquad \qquad \text{ because } \mathbb P\left(Y < F^{-1}_{Y\vert D=1,Z=z}\left(\frac{p_a}{\mathbb E[D\vert Z=z]}\right) \vert D = 1, Z = z\right)=p_a/\mathbb E[D\vert Z=z],\\
       &=& \mathbb E\left[Y\vert D=1, Z=z, Y < F^{-1}_{Y\vert D=1, Z=z}\left(\frac{p_a}{\mathbb E[D\vert Z=z]}\right)\right]. 
    \end{eqnarray*}

    A similar argument for $\mu_{F_{0n}^{UB_z}}$ yields the formula in the lemma.
    
\end{proof}

Second, when $p_{df}=0$, we have the bounds in $\Theta_I(A_1)$.

Third, suppose $p_{df}=\min\{\mathbb E[D\vert Z=0],\mathbb E[1-D\vert Z=1]\}$. We distinguish three cases.
\begin{enumerate}
    \item $p_{df}=\mathbb E[D\vert Z=0] < \mathbb E[1-D\vert Z=1]$

    In this case, $p_a=\mathbb E[D\vert Z=0]-p_{df}=0$, $p_c=\mathbb E[D\vert Z=1]-p_a=\mathbb E[D\vert Z=1]$, and $p_n=\mathbb E[1-D\vert Z=1]-\mathbb E[D\vert Z=0]$. Equations \eqref{important3}-\eqref{important4} imply
    \begin{eqnarray*}
        \mathbb P(Y_1 \leq y \vert T=c) &=& \mathbb P(Y \leq y \vert D=1, Z=1),\\
        \mathbb P(Y_1 \leq y \vert T=df) &=& \mathbb P(Y \leq y \vert D=1, Z=0),\\
    \end{eqnarray*}
Following a similar reasoning as before and combining Equations \eqref{important}-\eqref{important2}, we obtain the following identified set for $F_{0n}$:
\begin{eqnarray}
    && \Theta_I(F_{0n})=\Bigg\{
    F_{0n} \in \mathcal F: \nonumber\\
    &&\max\bigg\{ \frac{(\mathbb P(Y\in (y,y'], D=0 \vert Z=0)- \mathbb E[D\vert Z=1])}{\mathbb E[1-D\vert Z=1]-\mathbb E[D\vert Z=0]} \cdot\nonumber\\
    && \qquad \mathbbm{1}\{\mathbb P(Y\in (y,y'], D=0 \vert Z=0)- \mathbb E[D\vert Z=1]\geq 0\}, \nonumber \\
   && \qquad \frac{(\mathbb P(Y\in (y,y'], D=0 \vert Z=1)- \mathbb E[D\vert Z=0])}{\mathbb E[1-D\vert Z=1]-\mathbb E[D\vert Z=0]} \cdot \nonumber \\
   && \qquad \mathbbm{1}\{\mathbb P(Y\in (y,y'], D=0 \vert Z=1)- \mathbb E[D\vert Z=0]\geq 0\} \bigg\}\nonumber \\
   && \qquad \leq F_{0n}(y')-F_{0n}(y) \leq \\
   && \min\bigg\{ \frac{\mathbb P(Y\in (y,y'], D=0 \vert Z=0)}{\mathbb E[1-D\vert Z=1]-\mathbb E[D\vert Z=0]}\mathbbm{1}\left\{\frac{\mathbb P(Y\in (y,y'], D=0 \vert Z=0)}{\mathbb E[1-D\vert Z=1]-\mathbb E[D\vert Z=0]} \leq 1\right\}\nonumber \\
   && \qquad \qquad +\mathbbm{1}\left\{\frac{\mathbb P(Y\in (y,y'], D=0 \vert Z=0)}{\mathbb E[1-D\vert Z=1]-\mathbb E[D\vert Z=0]} > 1\right\},\nonumber\\
   && \qquad \frac{\mathbb P(Y\in (y,y'], D=0 \vert Z=1)}{\mathbb E[1-D\vert Z=1]-\mathbb E[D\vert Z=0]}\mathbbm{1}\left\{\frac{\mathbb P(Y\in (y,y'], D=0 \vert Z=1)}{\mathbb E[1-D\vert Z=1]-\mathbb E[D\vert Z=0]} \leq 1\right\}\nonumber \\
   && \qquad \qquad +\mathbbm{1}\left\{\frac{\mathbb P(Y\in (y,y'], D=0 \vert Z=1)}{\mathbb E[1-D\vert Z=1]-\mathbb E[D\vert Z=0]} > 1\right\} \bigg\} \text{ for all } y<y'    
   \Bigg\}.\nonumber
    \end{eqnarray}
\begin{eqnarray*}
   &&F_{0n}^{LB}(y)\equiv \max\left\{F_{0n}^{LB_1}(y),F_{0n}^{LB_0}(y)\right\} \leq F_{0n}(y) \leq \min\left\{F_{0n}^{UB_1}(y),F_{0n}^{UB_0}(y)\right\}\equiv F_{0n}^{UB}(y),
\end{eqnarray*}
where
\begin{eqnarray*}
    && F_{0n}^{LB_0}(y) \equiv \frac{(\mathbb P(Y\leq y, D=0 \vert Z=0)- \mathbb E[D\vert Z=1])\mathbbm{1}\{\mathbb P(Y \leq y, D=0 \vert Z=0)- \mathbb E[D\vert Z=1]\geq 0\}}{\mathbb E[1-D\vert Z=1]-\mathbb E[D\vert Z=0]},\\
   && F_{0n}^{LB_1}(y) \equiv \frac{(\mathbb P(Y\leq y, D=0 \vert Z=1)- \mathbb E[D\vert Z=0])\mathbbm{1}\{\mathbb P(Y\leq y, D=0 \vert Z=1)- \mathbb E[D\vert Z=0]\geq 0\}}{\mathbb E[1-D\vert Z=1]-\mathbb E[D\vert Z=0]},\\
   && F_{0n}^{UB_0}(y) \equiv \frac{\mathbb P(Y \leq y, D=0 \vert Z=0)}{\mathbb E[1-D\vert Z=1]-\mathbb E[D\vert Z=0]}\mathbbm{1}\left\{\frac{\mathbb P(Y \leq y, D=0 \vert Z=0)}{\mathbb E[1-D\vert Z=1]-\mathbb E[D\vert Z=0]} \leq 1\right\}\\
   && + \mathbbm{1}\left\{\frac{\mathbb P(Y \leq y, D=0 \vert Z=0)}{\mathbb E[1-D\vert Z=1]-\mathbb E[D\vert Z=0]} > 1\right\},\\
   && F_{0n}^{UB_1}(y) \equiv \frac{\mathbb P(Y \leq y, D=0 \vert Z=1)}{\mathbb E[1-D\vert Z=1]-\mathbb E[D\vert Z=0]}\mathbbm{1}\left\{\frac{\mathbb P(Y\leq y, D=0 \vert Z=1)}{\mathbb E[1-D\vert Z=1]-\mathbb E[D\vert Z=0]} \leq 1\right\}\\
   && + \mathbbm{1}\left\{\frac{\mathbb P(Y\leq y, D=0 \vert Z=1)}{\mathbb E[1-D\vert Z=1]-\mathbb E[D\vert Z=0]} > 1\right\}.
    \end{eqnarray*}
The identified sets for $F_{dt}$, $d\in\{0,1\}$, $t\in \{c,df\}$ are given by
  \begin{eqnarray*}
    F_{0c}(y) &=& \frac{\mathbb P(Y \leq y, D=0 \vert Z=0)-(\mathbb E[1-D\vert Z=1]-\mathbb E[D\vert Z=0]) F_{0n}(y)}{\mathbb E[D\vert Z=1]},\\
    F_{0df}(y) &=& \frac{\mathbb P(Y \leq y, D=0 \vert Z=1)-(\mathbb E[1-D\vert Z=1]-\mathbb E[D\vert Z=0]) F_{0n}(y)}{\mathbb E[D\vert Z=0]},
  \end{eqnarray*}  
where $F_{0n} \in \Theta_I(F_{0n})$. 

\begin{eqnarray*}
    \Theta_I^{3,1}(A_2) &=& \Bigg \{\Gamma \in \mathbb R^{20}: \theta_{0a}(p_a)=\theta_{1a}(p_a) = \mu_{1a} - \mu_{0a}, \theta_{0n}(p_n) = \theta_{1n}(p_n) = \mu_{1n} - \mu_{0n}(p_n),\\
   && \qquad \theta_{0c}(p_{c})=\theta_{1c}(p_{c})=\mathbb{E}[Y \vert  D=1, Z=1]-\frac{\mathbb{E}[Y(1-D)  \vert  Z=0]-p_n \mu_{0 n}(p_n)}{\mathbb{E}[1-D  \vert  Z=0]-p_n},\\   
   &&\qquad \theta_{0df}(p_{df})=\theta_{1df}(p_{df}) = \mathbb{E}[Y \vert  D=1, Z=0]- \frac{\mathbb{E}[Y(1-D)  \vert  Z=1]-p_n \mu_{0 n}(p_n)}{\mathbb{E}[1-D  \vert  Z=1]-p_n},\\   
   && \qquad \delta_{0a}=\delta_{1a}=\delta_{0c}=\delta_{1c}=\delta_{0df}=\delta_{1df}=\delta_{0n}=\delta_{1n}=0,\\
   && \qquad p_a=0, p_c=\mathbb{E}[D  \vert  Z=1], p_{df}=\mathbb{E}[D  \vert  Z=0], p_n=\mathbb{E}[1-D  \vert  Z=1]-\mathbb{E}[D  \vert  Z=0],\\
    && \qquad  \mu_{0n}(p_n)\in \left[\mu_{F_{0n}^{UB}}(p_{n}),\mu_{F_{0n}^{LB}}(p_{n})\right], \mu_{1n}, \mu_{0a}, \mu_{1a} \in \left[\inf \mathcal{Y}, \sup \mathcal{Y}\right]
   \Bigg \}, 
\end{eqnarray*}

    \item $p_{df}=\mathbb E[1-D\vert Z=1] < \mathbb E[D\vert Z=0]$

    This is case is symmetric to the previous. By setting $\tilde{D}=1-D$ and $\tilde{Z}=1-Z$, we can adapt the previous identification results.

    Hence, $p_{n}=\mathbb E[1-D\vert Z=1]-p_{df}=0$, $p_c=\mathbb E[1-D\vert Z=0]-p_{n}=\mathbb E[1-D\vert Z=0]$, and $p_a=\mathbb E[D\vert Z=0]-\mathbb E[1-D\vert Z=1]$. Equations \eqref{important3}-\eqref{important4} imply
    \begin{eqnarray*}
        \mathbb P(Y_0 \leq y \vert T=c) &=& \mathbb P(Y \leq y \vert D=0, Z=0),\\
        \mathbb P(Y_0 \leq y \vert T=df) &=& \mathbb P(Y \leq y \vert D=0, Z=1),\\
    \end{eqnarray*}
Following a similar reasoning as before and combining Equations \eqref{important}-\eqref{important2}, we obtain the following identified set for $F_{1a}$:
\begin{eqnarray}
    && \Theta_I(F_{1a})=\Bigg\{
    F_{1a} \in \mathcal F: \nonumber\\
    &&\max\bigg\{ \frac{(\mathbb P(Y\in (y,y'], D=1 \vert Z=1)- \mathbb E[1-D\vert Z=0])}{\mathbb E[D\vert Z=0]-\mathbb E[1-D\vert Z=1]} \cdot \nonumber\\
    && \qquad \mathbbm{1}\{\mathbb P(Y\in (y,y'], D=1 \vert Z=1)- \mathbb E[1-D\vert Z=0]\geq 0\} \nonumber\\
   && \qquad \frac{(\mathbb P(Y\in (y,y'], D=1 \vert Z=0)- \mathbb E[1-D\vert Z=1])}{\mathbb E[D\vert Z=0]-\mathbb E[1-D\vert Z=1]} \cdot \nonumber \\
   && \qquad \mathbbm{1}\{\mathbb P(Y\in (y,y'], D=1 \vert Z=0)- \mathbb E[1-D\vert Z=1]\geq 0\} \bigg\} \nonumber \\
   && \qquad \leq F_{1a}(y')-F_{1a}(y) \leq \\
   && \min\bigg\{ \frac{\mathbb P(Y\in (y,y'], D=1 \vert Z=1)}{\mathbb E[D\vert Z=0]-\mathbb E[1-D\vert Z=1]}\mathbbm{1}\left\{\frac{\mathbb P(Y\in (y,y'], D=1 \vert Z=1)}{\mathbb E[D\vert Z=0]-\mathbb E[1-D\vert Z=1]} \leq 1\right\}\nonumber \\
   && \qquad \qquad +\mathbbm{1}\left\{\frac{\mathbb P(Y\in (y,y'], D=1 \vert Z=1)}{\mathbb E[D\vert Z=0]-\mathbb E[1-D\vert Z=1]} > 1\right\},\nonumber\\
   && \qquad \frac{\mathbb P(Y\in (y,y'], D=1 \vert Z=0)}{\mathbb E[D\vert Z=0]-\mathbb E[1-D\vert Z=1]}\mathbbm{1}\left\{\frac{\mathbb P(Y\in (y,y'], D=1 \vert Z=0)}{\mathbb E[D\vert Z=0]-\mathbb E[1-D\vert Z=1]} \leq 1\right\}\nonumber \\
   && \qquad \qquad +\mathbbm{1}\left\{\frac{\mathbb P(Y\in (y,y'], D=1 \vert Z=0)}{\mathbb E[D\vert Z=0]-\mathbb E[1-D\vert Z=1]} > 1\right\} \bigg\} \text{ for all } y<y'    
   \Bigg\}.\nonumber
    \end{eqnarray}
\begin{eqnarray*}
   &&F_{1a}^{LB}(y)\equiv \max\left\{F_{1a}^{LB_1}(y),F_{1a}^{LB_0}(y)\right\} \leq F_{1a}(y) \leq \min\left\{F_{1a}^{UB_1}(y),F_{1a}^{UB_0}(y)\right\}\equiv F_{1a}^{UB}(y),
\end{eqnarray*}
where
\begin{eqnarray*}
    && F_{1a}^{LB_1}(y) \equiv \frac{(\mathbb P(Y\leq y, D=1 \vert Z=1)- \mathbb E[1-D\vert Z=0])}{\mathbb E[D\vert Z=0]-\mathbb E[1-D\vert Z=1]} \cdot\\
    && \qquad \mathbbm{1}\{\mathbb P(Y \leq y, D=1 \vert Z=1)- \mathbb E[1-D\vert Z=0]\geq 0\}, \\
   && F_{1a}^{LB_0}(y) \equiv \frac{(\mathbb P(Y\leq y, D=1 \vert Z=0)- \mathbb E[1-D\vert Z=1])}{\mathbb E[D\vert Z=0]-\mathbb E[1-D\vert Z=1]} \cdot\\
   && \qquad \mathbbm{1}\{\mathbb P(Y\leq y, D=1 \vert Z=0)- \mathbb E[1-D\vert Z=1]\geq 0\}, \\
   && F_{1a}^{UB_1}(y) \equiv \frac{\mathbb P(Y \leq y, D=1 \vert Z=1)}{\mathbb E[D\vert Z=0]-\mathbb E[1-D\vert Z=1]}\mathbbm{1}\left\{\frac{\mathbb P(Y \leq y, D=1 \vert Z=1)}{\mathbb E[D\vert Z=0]-\mathbb E[1-D\vert Z=1]} \leq 1\right\}\\
   && + \mathbbm{1}\left\{\frac{\mathbb P(Y \leq y, D=1 \vert Z=1)}{\mathbb E[D\vert Z=0]-\mathbb E[1-D\vert Z=1]} > 1\right\},\\
   && F_{1a}^{UB_0}(y) \equiv \frac{\mathbb P(Y \leq y, D=1 \vert Z=0)}{\mathbb E[D\vert Z=0]-\mathbb E[1-D\vert Z=1]}\mathbbm{1}\left\{\frac{\mathbb P(Y\leq y, D=1 \vert Z=0)}{\mathbb E[D\vert Z=0]-\mathbb E[1-D\vert Z=1]} \leq 1\right\}\\
   && + \mathbbm{1}\left\{\frac{\mathbb P(Y\leq y, D=1 \vert Z=0)}{\mathbb E[D\vert Z=0]-\mathbb E[1-D\vert Z=1]} > 1\right\}.
    \end{eqnarray*}
The identified sets for $F_{dt}$, $d\in\{0,1\}$, $t\in \{c,df\}$ are given by
  \begin{eqnarray*}
    F_{1c}(y) &=& \frac{\mathbb P(Y \leq y, D=1 \vert Z=1)-(\mathbb E[D\vert Z=0]-\mathbb E[1-D\vert Z=1]) F_{1a}(y)}{\mathbb E[1-D\vert Z=0]},\\
    F_{1df}(y) &=& \frac{\mathbb P(Y \leq y_, D=1 \vert Z=0)-(\mathbb E[D\vert Z=0]-\mathbb E[1-D\vert Z=1]) F_{1a}(y)}{\mathbb E[1-D\vert Z=1]},
  \end{eqnarray*}  
where $F_{1a} \in \Theta_I(F_{1a})$.

\begin{eqnarray*}
    \Theta_I^{3,2}(A_2) &=& \Bigg \{\Gamma \in \mathbb R^{20}: \theta_{0a}(p_a)=\theta_{1a}(p_a) = \mu_{1a}(p_a) - \mu_{0a}, \theta_{0n}(p_n) = \theta_{1n}(p_n) = \mu_{1n} - \mu_{0n},\\
   && \qquad \theta_{0c}(p_{c})=\theta_{1c}(p_{c})=\frac{\mathbb{E}[Y D  \vert  Z=1]-p_a \mu_{1 a}(p_a)}{\mathbb{E}[D  \vert  Z=1]-p_a}-\mathbb{E}[Y  \vert D=0,  Z=0],\\   
   &&\qquad \theta_{0df}(p_{df})=\theta_{1df}(p_{df}) = \frac{\mathbb{E}[Y D  \vert  Z=0]-p_a \mu_{1 a}(p_a)}{\mathbb{E}[D  \vert  Z=0]-p_a}- \mathbb{E}[Y  \vert D=0, Z=1],\\   
   && \qquad \delta_{0a}=\delta_{1a}=\delta_{0c}=\delta_{1c}=\delta_{0df}=\delta_{1df}=\delta_{0n}=\delta_{1n}=0,\\
   && \qquad p_a=\mathbb{E}[D  \vert  Z=0]-\mathbb{E}[1-D  \vert  Z=1], p_c=\mathbb{E}[1-D  \vert  Z=0],\\
   && \qquad p_{df} = \mathbb{E}[1-D  \vert  Z=1], p_n=0,\\
    && \qquad \mu_{1a}(p_a) \in \left[\mu_{F_{1a}^{UB}}(p_a),\mu_{F_{1a}^{LB}}(p_a)\right], \mu_{1n}, \mu_{0n}, \mu_{0a} \in \left[\inf \mathcal{Y}, \sup \mathcal{Y}\right]
   \Bigg \}, 
\end{eqnarray*}

    \item $p_{df}=\mathbb E[D\vert Z=0] = \mathbb E[1-D\vert Z=1]$

    Then $p_a=p_{n}=0$, $p_c=\mathbb E[D\vert Z=1]=\mathbb E[1-D\vert Z=0]$.
    \begin{eqnarray*}
    \mathbb P(Y_1 \leq y \vert T=c) &=& \mathbb P(Y \leq y \vert D=1, Z=1),\\
    \mathbb P(Y_1 \leq y \vert T=df) &=& \mathbb P(Y \leq y \vert D=1, Z=0),\\
    \mathbb P(Y_0 \leq y \vert T=c) &=& \mathbb P(Y \leq y \vert D=0, Z=0),\\
    \mathbb P(Y_0 \leq y \vert T=df) &=& \mathbb P(Y \leq y \vert D=0, Z=1),\\
    \end{eqnarray*}
    
\end{enumerate}

\begin{eqnarray*}
    \Theta_I^{3,3}(A_2) &=& \Bigg \{\Gamma \in \mathbb R^{20}: \theta_{0a}(p_a)=\theta_{1a}(p_a) = \mu_{1a} - \mu_{0a}, \theta_{0n}(p_n) = \theta_{1n}(p_n) = \mu_{1n} - \mu_{0n},\\
   && \qquad \theta_{0c}(p_{c})=\theta_{1c}(p_{c})=\mathbb{E}[Y  \vert  D=1, Z=1]-\mathbb{E}[Y  \vert D=0, Z=0],\\   
   &&\qquad \theta_{0df}(p_{df})=\theta_{1df}(p_{df}) = \mathbb{E}[Y  \vert  D=1, Z=0]- \mathbb{E}[Y  \vert D=0, Z=1],\\   
   && \qquad \delta_{0a}=\delta_{1a}=\delta_{0c}=\delta_{1c}=\delta_{0df}=\delta_{1df}=\delta_{0n}=\delta_{1n}=0,\\
   && \qquad p_a=0, p_c=\mathbb{E}[D  \vert  Z=1]=\mathbb{E}[1-D  \vert  Z=0], \\
   && \qquad p_{df}=\mathbb{E}[D  \vert  Z=0]=\mathbb{E}[1-D  \vert  Z=1], p_n=0,\\
    && \qquad \mu_{1a}, \mu_{0n}, \mu_{1n}, \mu_{0a} \in \left[\inf \mathcal{Y}, \sup \mathcal{Y}\right]
   \Bigg \}. 
\end{eqnarray*}

Finally,
\begin{eqnarray*} 
\Theta_I^3(A_2)=\left\{\begin{array}{l}
        \Theta_I^{3,1}(A_2),\ \text{ if }\ p_{df}=\mathbb E[D\vert Z=0] < \mathbb E[1-D\vert Z=1],\\
        \Theta_I^{3,2}(A_2),\ \text{ if }\ p_{df}=\mathbb E[1-D\vert Z=1] < \mathbb E[D\vert Z=0],\\
       \Theta_I^{3,3}(A_2),\ \text{ if }\ p_{df}=\mathbb E[D\vert Z=0] = \mathbb E[1-D\vert Z=1].
        \end{array}\right.
\end{eqnarray*}
\end{proof}

\section{Proof of Proposition \ref{prop:noer}}\label{propC}

\begin{proof}
When $\mathbb{E}[D  \vert  Z=1] - \mathbb{E}[D  \vert  Z=0]> 0$, under Assumption \ref{ass:mon}, there are no defiers, i.e., $p_{df}=0$. This implies $p_c=\mathbb{E}[D  \vert  Z=1] - \mathbb{E}[D  \vert  Z=0]$, $p_a=\mathbb E[D\vert Z=0]$, and $p_n=\mathbb E[1-D\vert Z=1]$.
From the main text, we have shown the following: for all Borel set $A$, 
\begin{eqnarray*}
    \mathbb P(Y_{11} \in A \vert T=c) &=& \frac{\mathbb P(Y\in A, D=1 \vert Z=1)- p_a \mathbb P(Y_{11} \in A \vert T=a)}{p_c},
\end{eqnarray*}
\begin{eqnarray*}
   \max\left\{\frac{\mathbb P(Y\in A, D=1 \vert Z=1)- p_c}{p_a},0\right\} \leq \mathbb P(Y_{11} \in A \vert T=a) \leq \min\left\{\frac{\mathbb P(Y\in A, D=1 \vert Z=1)}{p_a},1\right\},  
\end{eqnarray*}
$\mathbb P(Y_{10}\in A \vert T=a)=\mathbb P(Y \in A \vert D=1, Z=0)$, $\mathbb P(Y_{01}\in A \vert T=n)=\mathbb P(Y \in A \vert D=0, Z=1)$,
\begin{eqnarray*}
   \max\left\{\frac{\mathbb P(Y\in A, D=0 \vert Z=0)- p_c}{p_n},0\right\} \leq \mathbb P(Y_{00} \in A \vert T=n) \leq \min\left\{\frac{\mathbb P(Y\in A, D=0 \vert Z=0)}{p_n},1\right\},  
\end{eqnarray*}
\begin{eqnarray*}
    \mathbb P(Y_{00} \in A \vert T=c) &=& \frac{\mathbb P(Y\in A, D=0 \vert Z=0)- p_n \mathbb P(Y_{00} \in A \vert T=n)}{p_c}.
\end{eqnarray*}
The results of this proposition are straightforward if $A$ is set to be $(-\infty, y]$, as one can easily check that the bounds are well-defined cdfs.  
\end{proof}

\begin{corollary}\label{cor:noer1}
 Suppose Assumptions \ref{ass:RA} and \ref{ass:mon} hold, and $\mathbb{E}[D  \vert  Z=1] - \mathbb{E}[D  \vert  Z=0]> 0$. Then,
\begin{eqnarray*}
    && p_c=\mathbb{E}[D  \vert  Z=1] - \mathbb{E}[D  \vert  Z=0],\ p_a=\mathbb E[D\vert Z=0],\ p_n=\mathbb E[1-D\vert Z=1], p_{df}=0,\\
    && \mu_{F_{11a}^{UB}} \leq \mu_{11a} \leq \mu_{F_{11a}^{LB}},\ \mu_{F_{00n}^{UB}} \leq \mu_{00n} \leq \mu_{F_{00n}^{LB}},\\
    && \mu_{11c} = \frac{\mathbb E[Y D \vert Z=1]- p_a \mu_{11a}}{p_c},\  \mu_{00c} = \frac{\mathbb E[Y (1-D) \vert Z=0]- p_n \mu_{00n}}{p_c},\\
    && \mu_{10a}=\mathbb E[Y \vert D=1, Z=0],\ \mu_{01n}=\mathbb E[Y \vert D=0, Z=1]. 
\end{eqnarray*}   
These bounds are sharp. 
For continuous outcomes with strictly increasing cdfs, we have: 
\begin{eqnarray*}
\mu_{F_{11a}^{LB}} &=& \mathbb E\left[Y\vert D=1, Z=1, Y> F^{-1}_{Y\vert D=1, Z=1}\left(1-\frac{p_a}{\mathbb E[D\vert Z=1]}\right)\right],\\
\mu_{F_{11a}^{UB}} &=& \mathbb E\left[Y\vert D=1, Z=1, Y< F^{-1}_{Y\vert D=1, Z=1}\left(\frac{p_a}{\mathbb E[D\vert Z=1]}\right)\right],\\
\mu_{F_{00n}^{LB}} &=& \mathbb E\left[Y\vert D=0, Z=0, Y> F^{-1}_{Y\vert D=0, Z=0}\left(1-\frac{p_{n}}{\mathbb E[1-D\vert Z=0]}\right)\right],\\
\mu_{F_{00n}^{UB}} &=& \mathbb E\left[Y\vert D=0, Z=0, Y< F^{-1}_{Y\vert D=0, Z=0}\left(\frac{p_{n}}{\mathbb E[1-D\vert Z=0]}\right)\right].
\end{eqnarray*}
\end{corollary}

\begin{proof}
    The results of this corollary hold from Proposition \ref{prop:noer} and Lemma \ref{lem:lee2009}. 
\end{proof}

\subsection{Proof of Corollary \ref{cor:noer2}}

\begin{proof}
   By definition, $\delta_{1a}=\mu_{11a}-\mu_{10a}$, $\delta_{0n}=\mu_{01n}-\mu_{00n}$, $\mu_{11c}-\mu_{00c}=\theta_{0c}+\delta_{1c}=\theta_{1c}+\delta_{0c}$. Therefore, from Corollary \ref{cor:noer1}, the results hold. 
\end{proof}

\section{Binary Outcomes}\label{binaryoutcome}
In this subsection, we provide bounds on the same causal parameters when the outcome of interest is binary. 
Equation \eqref{test:imp} becomes
\begin{eqnarray}\label{test:impbinary}
    \max_{d\in \{0,1\}}\left\{\max_{z\in\{0,1\}} \mathbb P(Y=0,D=d \vert Z=z) +\max_{z\in\{0,1\}} \mathbb P(Y=1,D=d \vert Z=z)\right\} \leq 1.
\end{eqnarray}
The bounds in Proposition \ref{prop_pdf} become
\small \begin{eqnarray}\label{prop21binary}
\Theta_{I}(p_{df}) = \left\{\begin{array}{lll}
& \bigg[\max \Big\{\max _{s \in\{0,1\}}\{\mathbb{P}(Y=0, D=s  \vert  Z=1-s)-\mathbb{P}(Y=0, D=s  \vert  Z=s),\\ &\mathbb{P}(Y=1, D=s  \vert  Z=1-s)-\mathbb{P}(Y=1, D=s  \vert  Z=s),\\ 
&\mathbb{P}(D=s  \vert  Z=1-s)-\mathbb{P}(D=s  \vert  Z=s)\}, 0 \Big\}, \\
& \qquad \qquad \qquad \min \{\mathbb{E}[D  \vert  Z=0], \mathbb{E}[1-D  \vert  Z=1]\}\bigg], \text{ if inequality } \eqref{test:impbinary} \text{ holds } \\
& \emptyset, \text{ otherwise}.
\end{array}\right.
\end{eqnarray}




In the main text, we have shown that for any Borel set A, 
\begin{eqnarray*}
   && \max\left\{\frac{\mathbb P(Y\in A, D=1 \vert Z=1)- p_{c}}{p_a}, \frac{\mathbb P(Y\in A, D=1 \vert Z=0)- p_{df}}{p_a},0\right\}\\
   && \qquad \qquad \leq \mathbb P(Y_1 \in A \vert T=a) \leq \\
   && \min\left\{\frac{\mathbb P(Y\in A, D=1 \vert Z=1)}{p_a}, \frac{\mathbb P(Y\in A, D=1 \vert Z=0)}{p_a},1\right\}.   
   \end{eqnarray*}
   Therefore, for $A=\{1\}$, we have
   \begin{eqnarray*}
   && \max\left\{\frac{\mathbb P(Y=1, D=1 \vert Z=1)- p_{c}}{p_a}, \frac{\mathbb P(Y=1, D=1 \vert Z=0)- p_{df}}{p_a},0\right\}\\
   && \qquad \qquad \leq \mu_{1a}(p_a) \leq \\
   && \min\left\{\frac{\mathbb P(Y=1, D=1 \vert Z=1)}{p_a}, \frac{\mathbb P(Y=1, D=1 \vert Z=0)}{p_a},1\right\}.   
   \end{eqnarray*}
   By a similar argument, we have
\begin{eqnarray*}
   && \max\left\{\frac{\mathbb P(Y=1, D=0 \vert Z=0)- p_{c}}{p_n}, \frac{\mathbb P(Y=1, D=0 \vert Z=1)- p_{df}}{p_n},0\right\}\\
   && \qquad \qquad \leq \mu_{0n}(p_n) \leq \\
   && \min\left\{\frac{\mathbb P(Y=1, D=0 \vert Z=0)}{p_n}, \frac{\mathbb P(Y=1, D=0 \vert Z=1)}{p_n},1\right\}.   
   \end{eqnarray*}

\section{Estimation and inference} \label{esti_inf}

\subsection{Estimation and inference under random assignment and exclusion restriction}

Our identification results rely on the observable data $(Y, D, Z)$ equipped with the probability measure $\mathbb{P}$. In this section, we provide a brief introduction of the estimation and inference methods based on a sample of size $n$, consisting of independently and identically distributed observations $\{(Y_i, D_i, Z_i)\}_{i=1}^n$ drawn from $\mathbb{P}$.

First, we aim to estimate the identified set for the proportion of defiers, as characterized in Proposition \ref{prop_pdf}. We can estimate the upper bound $UB(p_{df}) \equiv \min \{\mathbb{E}[D \mid Z=0], \mathbb{E}[1-D \mid Z=1]\}$ by 
\begin{equation*}
    \begin{aligned}
        & \widehat{UB}(p_{df}) \equiv \min \{\widehat{\mathbb{E}}[D \mid Z=0], \widehat{\mathbb{E}}[1-D \mid Z=1]\}, \\
        &\widehat{\mathbb{E}}[D \mid Z=0] = \frac{\sum_{i=1}^{n}D_i (1-Z_i)}{\sum_{i=1}^{n}(1-Z_i)}, \widehat{\mathbb{E}}[1-D \mid Z=1] = \frac{\sum_{i=1}^{n}(1-D_i) Z_i}{\sum_{i=1}^{n} Z_i}.
    \end{aligned}
\end{equation*}
It is straightforward to verify that $\widehat{\mathbb{E}}[D \mid Z=0]$ and $\widehat{\mathbb{E}}[1 - D \mid Z=1]$ are consistent estimators of $\mathbb{E}[D \mid Z=0]$ and $\mathbb{E}[1 - D \mid Z=1]$ by the law of large numbers. Consistency of $\widehat{UB}(p_{df})$ for $UB(p_{df})$ follows from the continuous mapping theorem, as the minimum function $\min\{\cdot, \cdot\}$ is continuous everywhere on $\mathbb{R}^2$.

If the outcome $Y$ has finite support, $Y \in \{y_1, y_2, \cdots, y_t\}$, $t < \infty$, we can estimate the lower bound of defier's proportion as 
\begin{equation*}
    \begin{aligned}
        \widehat{LB}^{disc}(p_{df}) \equiv \max_{s \in \{0,1\}}\bigg\{ \max_{y \in \{y_1, y_2, \cdots, y_t\}} \bigg\{ & \frac{\sum_{i=1}^{n} \mathbbm{1} \{Y_i = y\} \mathbbm{1}\{D_i = s\} \mathbbm{1}\{Z_i = 1-s\}}{\mathbbm{1}\{Z_i = 1-s\}} \\
        &- \frac{\sum_{i=1}^{n} \mathbbm{1} \{Y_i = y\} \mathbbm{1}\{D_i = s\} \mathbbm{1}\{Z_i = s\}}{\mathbbm{1}\{Z_i = s\}} \bigg\} \bigg\}.
    \end{aligned}
\end{equation*}
According to the law of large number and the continuous mapping theorem, $\widehat{LB}^{disc}(p_{df})$ is a consistent estimator of $LB(p_{df}) \equiv \max _{s \in\{0,1\}}\big\{\sup _A\{\mathbb{P}(Y \in A, D=s \mid Z=1-s)-\mathbb{P}(Y \in A, D=s \mid Z=s)\}\big\}$.

If the outcome $Y$ is continuous, taking the supremum over all Borel sets $A$ in the support of $Y$ becomes intractable. To address this, we partition the support of $Y$ into $d$ disjoint intervals, denoted as $I_1, \dots, I_d$, and define the lower bound for the proportion of defiers based on this partition as
\begin{equation*}
    LB^{part}(p_{df}) \equiv \max_{A \in \mathcal{A}^{part}}\{\mathbb{P}(Y \in A, D=s \mid Z=1-s)-\mathbb{P}(Y \in A, D=s \mid Z=s)\},
\end{equation*}
where $\mathcal{A}^{part}$ is the finite $\sigma$-algebra generated by partitioning intervals $I_1, \cdots, I_d$, and $d$ is a fixed number specified by the researcher. Since $\mathcal{A}^{part}$ is a subset of the Borel $\sigma$-algebra generated by $Y$, it follows that $LB^{part}(p_{df}) \leq LB(p_{df})$. As a result, when employing $LB^{part}(p_{df})$, the estimated lower bound on the proportion of defiers remains valid, though the sharpness of the bound may be compromised. Increasing the number of partition intervals can improve the approximation, making $LB^{part}(p_{df})$ closer to the sharp lower bound $LB(p_{df})$. However, partitioning the support into a large number of intervals, especially when the number of intervals is large relative to the sample size or grows with the sample size, can introduce inference problems associated with many moment inequalities. For example, the confidence intervals might be too conservative when correcting for the large number of moment inequalities using critical values based on max-statistics. Therefore, we recommend choosing the number of intervals $d$ sufficiently small relative to the sample size $n$. An estimator of $LB^{part}(p_{df})$ is 
\begin{equation*}
    \begin{aligned}
        \widehat{LB}^{part}(p_{df}) \equiv \max_{s \in \{0,1\}}\bigg\{ \max_{A \in \mathcal{A}^{part}} \bigg\{ & \frac{\sum_{i=1}^{n} \mathbbm{1} \{Y_i \in A\} \mathbbm{1}\{D_i = s\} \mathbbm{1}\{Z_i = 1-s\}}{\mathbbm{1}\{Z_i = 1-s\}} \\
        &- \frac{\sum_{i=1}^{n} \mathbbm{1} \{Y_i \in A\} \mathbbm{1}\{D_i = s\} \mathbbm{1}\{Z_i = s\}}{\mathbbm{1}\{Z_i = s\}} \bigg\} \bigg\},
    \end{aligned}
\end{equation*}
which is consistent based on the law of large number and the continuous mapping theorem. To conclude, we estimate the identified bound of defier's proportion as 
\begin{equation*}
    \widehat{\Theta}_I(p_{d f})=\left\{\begin{array}{l}
        \big[\widehat{LB}^{disc}(p_{df}), \widehat{UB}(p_{df}) \big] \quad \text {if $Y$ is discrete,} \\
        \big[\widehat{LB}^{part}(p_{df}), \widehat{UB}(p_{df})\big] \quad \text {if $Y$ is continuous.}
        \end{array}\right.
\end{equation*}

Next, we conduct inference on the identified bound for the proportion of defiers, which corresponds to an intersection of finitely many intervals when the number of partitions $d$ is sufficiently small. We employ the inference procedure in \cite{Nevo2012IdentificationInstruments}, which is specifically designed for identified bounds in the form of intersections of finite intervals. The inference steps include selecting the set of moment inequalities close to binding at the estimated boundary points of the identified set. The corresponding estimates are then adjusted using their standard errors, along with critical values based on the distribution of the maximum or minimum over the moments. For a detailed description of the inference procedure and the underlying asymptotic arguments, we refer to Section IV of \cite{Nevo2012IdentificationInstruments}. Implementing their method yields a confidence band, denoted as $CI(p_{df})$, uniformly valid over points in the identified bound of defier's proportion. 

The proportions of the remaining types can be identified as $p_a = \mathbb{P}(D = 1 \mid Z = 0) - p_{df}$, $p_n = \mathbb{P}(D=0 \mid Z=1)$, and $p_c = \mathbb{P}(D = 1 \mid Z = 1) - \mathbb{P}(D = 1 \mid Z = 0) + p_{df}$. Given these expressions, we estimate their respective identified bounds as
\begin{equation*}
    \begin{aligned}
        & \widehat{\Theta}_I(p_a) = \bigg[\frac{\sum_{i=1}^{n}D_i (1-Z_i)}{\sum_{i=1}^{n}(1-Z_i)} - \widehat{UB}(p_{df}), \frac{\sum_{i=1}^{n}D_i (1-Z_i)}{\sum_{i=1}^{n}(1-Z_i)} - \widehat{LB}(p_{df}) \bigg], \\
        & \widehat{\Theta}_I(p_n) = \bigg[\frac{\sum_{i=1}^{n}(1-D_i) Z_i}{\sum_{i=1}^{n}Z_i} - \widehat{UB}(p_{df}), \frac{\sum_{i=1}^{n}(1-D_i) Z_i}{\sum_{i=1}^{n}Z_i} - \widehat{LB}(p_{df}) \bigg], \\
        & \widehat{\Theta}_I(p_c) = \bigg[\frac{\sum_{i=1}^{n}D_i Z_i}{\sum_{i=1}^{n}Z_i} - \frac{\sum_{i=1}^{n}D_i (1-Z_i)}{\sum_{i=1}^{n}(1-Z_i)} + \widehat{LB}(p_{df}), \frac{\sum_{i=1}^{n}D_i Z_i}{\sum_{i=1}^{n}Z_i} - \frac{\sum_{i=1}^{n}D_i (1-Z_i)}{\sum_{i=1}^{n}(1-Z_i)} + \widehat{UB}(p_{df}) \bigg],
    \end{aligned}
\end{equation*}
where $\widehat{LB}(p_{df}) = \widehat{LB}^{disc}(p_{df})$ if $Y$ is discrete and $\widehat{LB}(p_{df}) = \widehat{LB}^{part}(p_{df})$ if $Y$ is continuous. It is straightforward that the proposed estimators of the bounds are consistent. We also apply the inference procedure in \cite{Nevo2012IdentificationInstruments} to conduct uniform inference on the proportion of each type.

In the application of \cite{bursztyn2020misperceived}, which investigates the effects of social norm perceptions on women's labor force participation in Saudi Arabia. The sample consists of 381 observations. The outcome is a continuous index measuring the long-term women labor force performance, so we equally partition the support of $Y$ into 20 intervals. Table \ref{tab:inf_type_prop_no_mon} presents the confidence bands for the proportions of all types.

\begin{table}
    \centering
    \caption{Estimation and Inference of proportion of each type under $A_2$}
      \begin{tabular}{crr}
      \toprule
      \multicolumn{1}{c}{Parameter} & \multicolumn{1}{c}{Estimated identified bounds} & \multicolumn{1}{c}{Confidence band} \\
      \midrule
      $p_a$ & [0, 0.1436] &  [0,  0.5369] \\
      $p_c$ & [0.1967, 0.3403] & [ 0,  0.3800] \\
      $p_n$ & [0.4123, 0.5559] & [ 0.3375,  1] \\ 
       $p_{df}$    & [0.1038, 0.2474] & [ 0,  0.2813]\\
      Observations & & 381 \\
      \bottomrule
      \end{tabular}%
    \label{tab:inf_type_prop_no_mon}%
\end{table}%

After estimating the type proportions, we proceed to estimate the parameters $\theta_{0 c}=\theta_{1 c}$ and $\theta_{0 d f}=\theta_{1 d f}$, which measure the local average direct effects of the treatment for compliers and defiers. Based on the identification results in Section \ref{sec:ra_er:ass}, we estimate the tractable valid outer set of $\hat{\theta}_{0 c}=\hat{\theta}_{1 c}$ as
\begin{equation*}
    \begin{aligned}
        \Bigg[ \max \Bigg\{& \frac{\sum_{i=1}^{n}Y_i D_i Z_i / \sum_{i=1}^{n} Z_i - \hat{p}_a \hat{\mu}_{F_{1a}^{LB_1}}(\hat{p}_a)}{\sum_{i=1}^{n} D_i Z_i / \sum_{i=1}^{n} Z_i-\hat{p}_a}-\frac{\sum_{i=1}^{n} Y_i (1-D_i) (1-Z_i) / \sum_{i=1}^{n}(1-Z_i)-\hat{p}_n \hat{\mu}_{F_{0n}^{UB_1}}\left(\hat{p}_n\right)}{\sum_{i=1}^{n} (1-D_i) (1-Z_i) / \sum_{i=1}^{n} (1-Z_i) - \hat{p}_n}, \\
         & \frac{\sum_{i=1}^{n}Y_i D_i Z_i / \sum_{i=1}^{n} Z_i - \hat{p}_a \hat{\mu}_{F_{1a}^{LB_1}}(\hat{p}_a)}{\sum_{i=1}^{n} D_i Z_i / \sum_{i=1}^{n} Z_i-\hat{p}_a}-\frac{\sum_{i=1}^{n} Y_i (1-D_i) (1-Z_i) / \sum_{i=1}^{n}(1-Z_i)-\hat{p}_n \hat{\mu}_{F_{0n}^{UB_0}}\left(\hat{p}_n\right)}{\sum_{i=1}^{n} (1-D_i) (1-Z_i) / \sum_{i=1}^{n} (1-Z_i) - \hat{p}_n}, \\ 
         & \frac{\sum_{i=1}^{n}Y_i D_i Z_i / \sum_{i=1}^{n} Z_i - \hat{p}_a \hat{\mu}_{F_{1a}^{LB_0}}(\hat{p}_a)}{\sum_{i=1}^{n} D_i Z_i / \sum_{i=1}^{n} Z_i-\hat{p}_a}-\frac{\sum_{i=1}^{n} Y_i (1-D_i) (1-Z_i) / \sum_{i=1}^{n}(1-Z_i)-\hat{p}_n \hat{\mu}_{F_{0n}^{UB_1}}\left(\hat{p}_n\right)}{\sum_{i=1}^{n} (1-D_i) (1-Z_i) / \sum_{i=1}^{n} (1-Z_i) - \hat{p}_n}, \\ 
         & \frac{\sum_{i=1}^{n}Y_i D_i Z_i / \sum_{i=1}^{n} Z_i - \hat{p}_a \hat{\mu}_{F_{1a}^{LB_0}}(\hat{p}_a)}{\sum_{i=1}^{n} D_i Z_i / \sum_{i=1}^{n} Z_i-\hat{p}_a}-\frac{\sum_{i=1}^{n} Y_i (1-D_i) (1-Z_i) / \sum_{i=1}^{n}(1-Z_i)-\hat{p}_n \hat{\mu}_{F_{0n}^{UB_0}}\left(\hat{p}_n\right)}{\sum_{i=1}^{n} (1-D_i) (1-Z_i) / \sum_{i=1}^{n} (1-Z_i) - \hat{p}_n} \Bigg\}, \\
         \min \Bigg\{& \frac{\sum_{i=1}^{n}Y_i D_i Z_i / \sum_{i=1}^{n} Z_i - \hat{p}_a \hat{\mu}_{F_{1a}^{UB_1}}(\hat{p}_a)}{\sum_{i=1}^{n} D_i Z_i / \sum_{i=1}^{n} Z_i-\hat{p}_a}-\frac{\sum_{i=1}^{n} Y_i (1-D_i) (1-Z_i) / \sum_{i=1}^{n}(1-Z_i)-\hat{p}_n \hat{\mu}_{F_{0n}^{LB_1}}\left(\hat{p}_n\right)}{\sum_{i=1}^{n} (1-D_i) (1-Z_i) / \sum_{i=1}^{n} (1-Z_i) - \hat{p}_n}, \\
         & \frac{\sum_{i=1}^{n}Y_i D_i Z_i / \sum_{i=1}^{n} Z_i - \hat{p}_a \hat{\mu}_{F_{1a}^{UB_1}}(\hat{p}_a)}{\sum_{i=1}^{n} D_i Z_i / \sum_{i=1}^{n} Z_i-\hat{p}_a}-\frac{\sum_{i=1}^{n} Y_i (1-D_i) (1-Z_i) / \sum_{i=1}^{n}(1-Z_i)-\hat{p}_n \hat{\mu}_{F_{0n}^{LB_0}}\left(\hat{p}_n\right)}{\sum_{i=1}^{n} (1-D_i) (1-Z_i) / \sum_{i=1}^{n} (1-Z_i) - \hat{p}_n}, \\ 
         & \frac{\sum_{i=1}^{n}Y_i D_i Z_i / \sum_{i=1}^{n} Z_i - \hat{p}_a \hat{\mu}_{F_{1a}^{UB_0}}(\hat{p}_a)}{\sum_{i=1}^{n} D_i Z_i / \sum_{i=1}^{n} Z_i-\hat{p}_a}-\frac{\sum_{i=1}^{n} Y_i (1-D_i) (1-Z_i) / \sum_{i=1}^{n}(1-Z_i)-\hat{p}_n \hat{\mu}_{F_{0n}^{LB_1}}\left(\hat{p}_n\right)}{\sum_{i=1}^{n} (1-D_i) (1-Z_i) / \sum_{i=1}^{n} (1-Z_i) - \hat{p}_n}, \\ 
         & \frac{\sum_{i=1}^{n}Y_i D_i Z_i / \sum_{i=1}^{n} Z_i - \hat{p}_a \hat{\mu}_{F_{1a}^{UB_0}}(\hat{p}_a)}{\sum_{i=1}^{n} D_i Z_i / \sum_{i=1}^{n} Z_i-\hat{p}_a}-\frac{\sum_{i=1}^{n} Y_i (1-D_i) (1-Z_i) / \sum_{i=1}^{n}(1-Z_i)-\hat{p}_n \hat{\mu}_{F_{0n}^{LB_0}}\left(\hat{p}_n\right)}{\sum_{i=1}^{n} (1-D_i) (1-Z_i) / \sum_{i=1}^{n} (1-Z_i) - \hat{p}_n} \Bigg\} \Bigg], \\
    \end{aligned}
\end{equation*}
and the tractable valid outer set of $\hat{\theta}_{0 df}=\hat{\theta}_{1 df}$ as 
\begin{equation*}
    \begin{aligned}
        \Bigg[ \max \Bigg\{& \frac{\sum_{i=1}^{n}Y_i D_i (1-Z_i) / \sum_{i=1}^{n} (1-Z_i) - \hat{p}_a \hat{\mu}_{F_{1a}^{LB_1}}(\hat{p}_a)}{\sum_{i=1}^{n} D_i (1-Z_i) / \sum_{i=1}^{n} (1-Z_i)-\hat{p}_a}-\frac{\sum_{i=1}^{n} Y_i (1-D_i) Z_i / \sum_{i=1}^{n}Z_i-\hat{p}_n \hat{\mu}_{F_{0n}^{UB_1}}\left(\hat{p}_n\right)}{\sum_{i=1}^{n} (1-D_i) Z_i/ \sum_{i=1}^{n} Z_i - \hat{p}_n}, \\
         & \frac{\sum_{i=1}^{n}Y_i D_i (1-Z_i) / \sum_{i=1}^{n} (1-Z_i) - \hat{p}_a \hat{\mu}_{F_{1a}^{LB_1}}(\hat{p}_a)}{\sum_{i=1}^{n} D_i (1-Z_i) / \sum_{i=1}^{n} (1-Z_i)-\hat{p}_a}-\frac{\sum_{i=1}^{n} Y_i (1-D_i) Z_i / \sum_{i=1}^{n}Z_i-\hat{p}_n \hat{\mu}_{F_{0n}^{UB_0}}\left(\hat{p}_n\right)}{\sum_{i=1}^{n} (1-D_i) Z_i/ \sum_{i=1}^{n} Z_i - \hat{p}_n}, \\ 
         & \frac{\sum_{i=1}^{n}Y_i D_i (1-Z_i) / \sum_{i=1}^{n} (1-Z_i) - \hat{p}_a \hat{\mu}_{F_{1a}^{LB_0}}(\hat{p}_a)}{\sum_{i=1}^{n} D_i (1-Z_i) / \sum_{i=1}^{n} (1-Z_i)-\hat{p}_a}-\frac{\sum_{i=1}^{n} Y_i (1-D_i) Z_i / \sum_{i=1}^{n}Z_i-\hat{p}_n \hat{\mu}_{F_{0n}^{UB_1}}\left(\hat{p}_n\right)}{\sum_{i=1}^{n} (1-D_i) Z_i/ \sum_{i=1}^{n} Z_i - \hat{p}_n}, \\ 
         & \frac{\sum_{i=1}^{n}Y_i D_i (1-Z_i) / \sum_{i=1}^{n} (1-Z_i) - \hat{p}_a \hat{\mu}_{F_{1a}^{LB_0}}(\hat{p}_a)}{\sum_{i=1}^{n} D_i (1-Z_i) / \sum_{i=1}^{n} (1-Z_i)-\hat{p}_a}-\frac{\sum_{i=1}^{n} Y_i (1-D_i) Z_i / \sum_{i=1}^{n}Z_i-\hat{p}_n \hat{\mu}_{F_{0n}^{UB_0}}\left(\hat{p}_n\right)}{\sum_{i=1}^{n} (1-D_i) Z_i/ \sum_{i=1}^{n} Z_i - \hat{p}_n} \Bigg\}, \\
         \min \Bigg\{& \frac{\sum_{i=1}^{n}Y_i D_i (1-Z_i) / \sum_{i=1}^{n} (1-Z_i) - \hat{p}_a \hat{\mu}_{F_{1a}^{UB_1}}(\hat{p}_a)}{\sum_{i=1}^{n} D_i (1-Z_i) / \sum_{i=1}^{n} (1-Z_i)-\hat{p}_a}-\frac{\sum_{i=1}^{n} Y_i (1-D_i) Z_i / \sum_{i=1}^{n}Z_i-\hat{p}_n \hat{\mu}_{F_{0n}^{LB_1}}\left(\hat{p}_n\right)}{\sum_{i=1}^{n} (1-D_i) Z_i/ \sum_{i=1}^{n} Z_i - \hat{p}_n}, \\
         & \frac{\sum_{i=1}^{n}Y_i D_i (1-Z_i) / \sum_{i=1}^{n} (1-Z_i) - \hat{p}_a \hat{\mu}_{F_{1a}^{UB_1}}(\hat{p}_a)}{\sum_{i=1}^{n} D_i (1-Z_i) / \sum_{i=1}^{n} (1-Z_i)-\hat{p}_a}-\frac{\sum_{i=1}^{n} Y_i (1-D_i) Z_i / \sum_{i=1}^{n}Z_i-\hat{p}_n \hat{\mu}_{F_{0n}^{LB_0}}\left(\hat{p}_n\right)}{\sum_{i=1}^{n} (1-D_i) Z_i/ \sum_{i=1}^{n} Z_i - \hat{p}_n}, \\ 
         & \frac{\sum_{i=1}^{n}Y_i D_i (1-Z_i) / \sum_{i=1}^{n} (1-Z_i) - \hat{p}_a \hat{\mu}_{F_{1a}^{UB_0}}(\hat{p}_a)}{\sum_{i=1}^{n} D_i (1-Z_i) / \sum_{i=1}^{n} (1-Z_i)-\hat{p}_a}-\frac{\sum_{i=1}^{n} Y_i (1-D_i) Z_i / \sum_{i=1}^{n}Z_i-\hat{p}_n \hat{\mu}_{F_{0n}^{LB_1}}\left(\hat{p}_n\right)}{\sum_{i=1}^{n} (1-D_i) Z_i/ \sum_{i=1}^{n} Z_i - \hat{p}_n}, \\ 
         & \frac{\sum_{i=1}^{n}Y_i D_i (1-Z_i) / \sum_{i=1}^{n} (1-Z_i) - \hat{p}_a \hat{\mu}_{F_{1a}^{UB_0}}(\hat{p}_a)}{\sum_{i=1}^{n} D_i (1-Z_i) / \sum_{i=1}^{n} (1-Z_i)-\hat{p}_a}-\frac{\sum_{i=1}^{n} Y_i (1-D_i) Z_i / \sum_{i=1}^{n}Z_i-\hat{p}_n \hat{\mu}_{F_{0n}^{LB_0}}\left(\hat{p}_n\right)}{\sum_{i=1}^{n} (1-D_i) Z_i/ \sum_{i=1}^{n} Z_i - \hat{p}_n} \Bigg\} \Bigg],
    \end{aligned}
\end{equation*}
given $(\hat{p}_a, \hat{p}_n)$ in the confidence band of identified set of types' proportions. If $ \hat{\mu}_{F_{1a}^{LB_z}}$, $\hat{\mu}_{F_{1a}^{UB_z}}$, $ \hat{\mu}_{F_{0n}^{LB_z}}$, $\hat{\mu}_{F_{0n}^{UB_z}}$ are consistent estimators of $\mu_{F_{1a}^{LB_z}}$, $\mu_{F_{1a}^{UB_z}}$, $\mu_{F_{0n}^{LB_z}}$, $\mu_{F_{0n}^{UB_z}}$, $z \in \{0,1\}$, in Lemma \ref{lem:lee2009}, then the proposed estimators consistently estimate the lower and upper bounds of the tractable valid outer sets by the law of large numbers and the continuous mapping theorem. We impose the estimators as
\begin{equation*}
    \begin{aligned}
        & \hat{\mu}_{F_{1a}^{LB_z}} = \frac{\sum_{i = 1}^{n} Y_i D_i \mathbbm{1} \{Z_i = z\} \mathbbm{1} \{Y_i > \hat{F}^{-1}_{Y \mid D = 1, Z = z}(1 - \hat{\alpha})\}}{\sum_{i = 1}^{n} D_i \mathbbm{1} \{Z_i = z\} \mathbbm{1} \{Y_i > \hat{F}^{-1}_{Y \mid D = 1, Z = z}(1 - \hat{\alpha})\}}, \\
        & \hat{\mu}_{F_{1a}^{UB_z}} = \frac{\sum_{i = 1}^{n} Y_i D_i \mathbbm{1} \{Z_i = z\} \mathbbm{1} \{Y_i < \hat{F}^{-1}_{Y \mid D = 1, Z = z}(\hat{\alpha})\}}{\sum_{i = 1}^{n} D_i \mathbbm{1} \{Z_i = z\} \mathbbm{1} \{Y_i < \hat{F}^{-1}_{Y \mid D = 1, Z = z}(\hat{\alpha})\}}, \\
        &\hat{\mu}_{F_{0n}^{LB_z}} = \frac{\sum_{i = 1}^{n} Y_i (1-D_i) \mathbbm{1} \{Z_i = z\} \mathbbm{1} \{Y_i > \hat{F}^{-1}_{Y \mid D = 0, Z = z}(1 - \hat{\gamma})\}}{\sum_{i = 1}^{n} (1-D_i) \mathbbm{1} \{Z_i = z\} \mathbbm{1} \{Y_i > \hat{F}^{-1}_{Y \mid D = 0, Z = z}(1 - \hat{\gamma})\}}, \\
        & \hat{\mu}_{F_{0n}^{UB_z}} = \frac{\sum_{i = 1}^{n} Y_i (1-D_i) \mathbbm{1} \{Z_i = z\} \mathbbm{1} \{Y_i < \hat{F}^{-1}_{Y \mid D = 0, Z = z}(\hat{\gamma})\}}{\sum_{i = 1}^{n} (1-D_i )\mathbbm{1} \{Z_i = z\} \mathbbm{1} \{Y_i < \hat{F}^{-1}_{Y \mid D = 0, Z = z}(\hat{\gamma})\}}, \\
        \text{where }& \hat{F}^{-1}_{Y \mid D = 1, Z = z}(q) = \min \left\{y: \frac{\sum_{i = 1}^{n}D_i \mathbbm{1} \{Z_i = z\} \mathbbm{1}\{Y_i \leq y\}}{\sum_{i = 1}^{n} D_i \mathbbm{1} \{Z_i = z\}} \geq q\right\}, \\
        & \hat{F}^{-1}_{Y \mid D = 0, Z = z}(q) = \min \left\{y: \frac{\sum_{i = 1}^{n}(1-D_i) \mathbbm{1} \{Z_i = z\} \mathbbm{1}\{Y_i \leq y\}}{\sum_{i = 1}^{n} (1-D_i) \mathbbm{1} \{Z_i = z\}} \geq q\right\}, \\
        & \hat{\alpha} = \frac{\hat{p_a}}{\sum_{i=1}^{n} D_i \mathbbm{1} \{Z_i = z\} / \sum_{i=1}^{n} \mathbbm{1} \{Z_i = z\}}, \\
        & \hat{\gamma} = \frac{\hat{p_n}}{\sum_{i=1}^{n} (1-D_i) \mathbbm{1} \{Z_i = z\} / \sum_{i=1}^{n} \mathbbm{1} \{Z_i = z\}},
    \end{aligned}
\end{equation*}
and Proposition 2 of \cite{Lee2009TrainingEffects} proves that the above estimators are consistent under the condition that the outcome variable $Y$ has bounded support.

The estimated outer sets of $\theta_{1c} = \theta_{0c}$ take the form of finite intersections of intervals. Therefore, we apply the inference procedure in \cite{Nevo2012IdentificationInstruments} to conduct uniform inference on these parameters. We apply our estimation and inference methods to the study in \cite{bursztyn2020misperceived}. Figure \ref{fig:inf_id_sets} presents the estimated identified bounds and the corresponding confidence bands for $\theta_{1c} = \theta_{0c}$ and $\theta_{1 df} = \theta_{0 df}$ given that $p_c$ and $p_{df}$ lie in their respective estimated identified sets.

\begin{figure}
	\centering
	\includegraphics[scale=0.7]{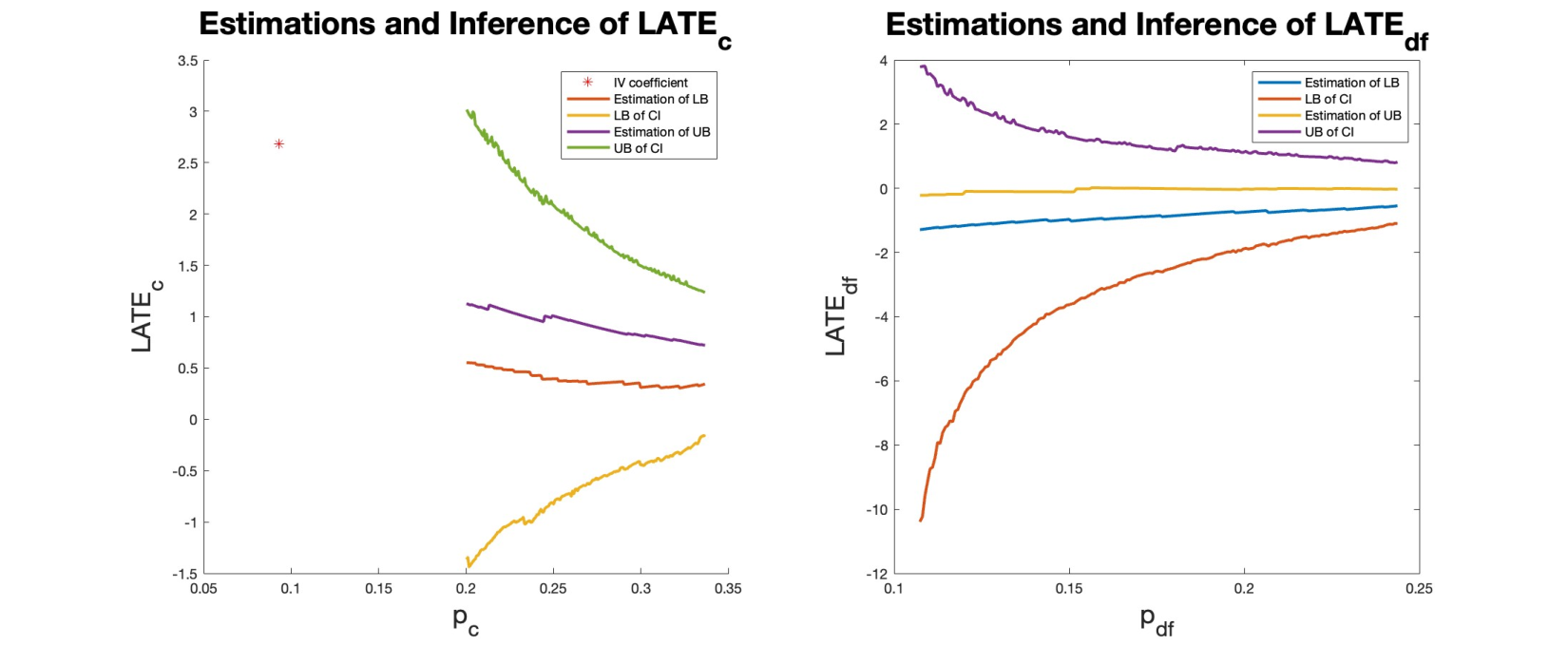}
	\caption{Estimated bounds and inference bands for LAICDEs under $A_2$ and IV estimand}
	\label{fig:inf_id_sets}
\end{figure}

\subsection{Estimation and inference under random assignment and monotonicity}
We can use the sample analog to estimate the bounds derived in Corollary \ref{cor:noer2}. First, we need to estimate the trimming probabilities $\alpha\equiv \frac{p_a}{\mathbb E[D\vert Z=1]}$ and $\gamma \equiv \frac{p_n}{\mathbb E[1-D\vert Z=0]}$ by 
\begin{equation*}
    \begin{aligned}
        & \hat{\alpha} = \frac{\hat{p}_a}{(\sum_{i = 1}^{n} D_i Z_i) / (\sum_{i = 1}^{n} Z_i)}, \\
        & \hat{\gamma} = \frac{\hat{p}_n}{(\sum_{i = 1}^{n} (1 - D_i) (1 - Z_i)) / (\sum_{i = 1}^{n} (1 - Z_i))}.
    \end{aligned}
\end{equation*}
Then, we can estimate the $\hat{\alpha}$th and $(1 - \hat{\alpha})$th quantile of the empirical distribution of $Y \mid D = 1, Z = 1$, and calculate the estimators of bounds on $\delta_{1a}$ as
\begin{equation*}
    \begin{aligned}
        & \widehat{\operatorname{LB}_1} = \frac{\sum_{i = 1}^{n} Y_i D_i Z_i \mathbbm{1} \{Y_i \leq \hat{F}^{-1}_{Y \mid D = 1, Z = 1}(\hat{\alpha})\}}{\sum_{i = 1}^{n} D_i Z_i \mathbbm{1} \{Y_i \leq \hat{F}^{-1}_{Y \mid D = 1, Z = 1}(\hat{\alpha})\}} - \frac{\sum_{i = 1}^{n} Y_i D_i (1 - Z_i)}{\sum_{i = 1}^{n} D_i (1 - Z_i)}, \\
        & \widehat{\operatorname{UB}_1} = \frac{\sum_{i = 1}^{n} Y_i D_i Z_i \mathbbm{1} \{Y_i > \hat{F}^{-1}_{Y \mid D = 1, Z = 1}(1 - \hat{\alpha})\}}{\sum_{i = 1}^{n} D_i Z_i \mathbbm{1} \{Y_i > \hat{F}^{-1}_{Y \mid D = 1, Z = 1}(1 - \hat{\alpha})\}} - \frac{\sum_{i = 1}^{n} Y_i D_i (1 - Z_i)}{\sum_{i = 1}^{n} D_i (1 - Z_i)}, \\
        & \hat{F}^{-1}_{Y \mid D = 1, Z = 1}(q) = \min \left\{y: \frac{\sum_{i = 1}^{n}D_i Z_i \mathbbm{1}\{Y_i \leq y\}}{\sum_{i = 1}^{n} D_i Z_i} \geq q\right\}.
    \end{aligned}
\end{equation*}
Repeating the same process, we can estimate the bounds on $\delta_{0n}$ as
\begin{equation*}
    \begin{aligned}
        & \widehat{\operatorname{LB}_2} = \frac{\sum_{i = 1}^{n} Y_i (1 - D_i) Z_i}{\sum_{i = 1}^{n} (1 - D_i) Z_i} - \frac{\sum_{i = 1}^{n} Y_i (1 - D_i) (1 - Z_i) \mathbbm{1} \{Y_i > \hat{F}^{-1}_{Y \mid D = 0, Z = 0}(1 - \hat{\gamma})\}}{\sum_{i = 1}^{n} (1 - D_i) (1 - Z_i) \mathbbm{1} \{Y_i > \hat{F}^{-1}_{Y \mid D = 0, Z = 0}(1 - \hat{\gamma})\}}, \\
        & \widehat{\operatorname{UB}_2} = \frac{\sum_{i = 1}^{n} Y_i (1 - D_i) Z_i}{\sum_{i = 1}^{n} (1 - D_i) Z_i} - \frac{\sum_{i = 1}^{n} Y_i (1 - D_i) (1 - Z_i) \mathbbm{1} \{Y_i \leq \hat{F}^{-1}_{Y \mid D = 0, Z = 0}(\hat{\gamma})\}}{\sum_{i = 1}^{n} (1 - D_i) (1 - Z_i) \mathbbm{1} \{Y_i \leq \hat{F}^{-1}_{Y \mid D = 0, Z = 0}(\hat{\gamma})\}}, \\
        & \hat{F}^{-1}_{Y \mid D = 0, Z = 0}(q) = \min \left\{y: \frac{\sum_{i = 1}^{n} (1 - D_i) (1 - Z_i) \mathbbm{1}\{Y_i \leq y\}}{\sum_{i = 1}^{n} (1 - D_i) (1 - Z_i)} \geq q\right\},
    \end{aligned}
\end{equation*}
and estimate the bounds on $\theta_{1c}+\delta_{0c}=\theta_{0c}+\delta_{1c}$ as
\begin{equation*}
    \begin{aligned}
        & \widehat{\operatorname{LB}_3} = \frac{\sum_{i = 1}^{n} Y_i D_i Z_i \mathbbm{1} \{Y_i \leq \hat{F}^{-1}_{Y \mid D = 1, Z = 1}(1 - \hat{\alpha})\}}{\sum_{i = 1}^{n} D_i Z_i \mathbbm{1} \{Y_i \leq \hat{F}^{-1}_{Y \mid D = 1, Z = 1}(1 - \hat{\alpha})\}} - \frac{\sum_{i = 1}^{n} Y_i (1 - D_i) (1 - Z_i) \mathbbm{1} \{Y_i > \hat{F}^{-1}_{Y \mid D = 0, Z = 0}(\hat{\gamma})\}}{\sum_{i = 1}^{n} (1 - D_i) (1 - Z_i) \mathbbm{1} \{Y_i > \hat{F}^{-1}_{Y \mid D = 0, Z = 0}(\hat{\gamma})\}}, \\
        & \widehat{\operatorname{UB}_3} = \frac{\sum_{i = 1}^{n} Y_i D_i Z_i \mathbbm{1} \{Y_i > \hat{F}^{-1}_{Y \mid D = 1, Z = 1}(\hat{\alpha})\}}{\sum_{i = 1}^{n} D_i Z_i \mathbbm{1} \{Y_i > \hat{F}^{-1}_{Y \mid D = 1, Z = 1}(\hat{\alpha})\}} - \frac{\sum_{i = 1}^{n} Y_i (1 - D_i) (1 - Z_i) \mathbbm{1} \{Y_i \leq \hat{F}^{-1}_{Y \mid D = 0, Z = 0}(1 - \hat{\gamma})\}}{\sum_{i = 1}^{n} (1 - D_i) (1 - Z_i) \mathbbm{1} \{Y_i \leq \hat{F}^{-1}_{Y \mid D = 0, Z = 0}(1 - \hat{\gamma})\}}.
    \end{aligned}
\end{equation*}
Proposition 2 in \cite{Lee2009TrainingEffects} proves that these estimate bounds are consistent.

In order to obtain the confidence intervals for bounds described in Corollary \ref{cor:noer2}, we have to know the asymptotic distributions of the estimated bound values. We apply Proposition 3 in \cite{Lee2009TrainingEffects} to derive the asymptotic behaviors of those estimated bound values. For the estimated bounds on $\delta_{1a}$, $\sqrt{N} (\widehat{\operatorname{LB}}_1 - \operatorname{LB}_1) \stackrel{d}{\rightarrow} N\left(0, V_{\operatorname{LB}_1} + V_{C_1}\right)$, $\sqrt{N} (\widehat{\operatorname{UB}}_1 - \operatorname{UB}_1) \stackrel{d}{\rightarrow} N\left(0, V_{\operatorname{UB}_1} + V_{C_1}\right)$, where
\begin{equation*}
    \begin{aligned}
        V_{\operatorname{LB}_1} =& \frac{V(Y \mid D = 1, Z = 1, Y \leq F^{-1}_{Y \mid D = 1, Z = 1} (\alpha))}{\mathbb{E}[D Z] \alpha} \\
        &+ \frac{(F^{-1}_{Y \mid D = 1, Z = 1}(\alpha) - \mathbb{E}[Y \mid D = 1, Z = 1, Y \leq F^{-1}_{Y \mid D = 1, Z = 1} (\alpha)])^2 (1 - \alpha)}{\mathbb{E}[DZ] \alpha} \\
        &+ \left(\frac{F^{-1}_{Y \mid D = 1, Z = 1} (\alpha) - \mathbb{E}[Y \mid D = 1, Z = 1, Y \leq F^{-1}_{Y \mid D = 1, Z = 1} (\alpha)]}{\alpha}\right)^2 V_{\alpha_1}, \\
        V_{\operatorname{UB}_1} =& \frac{V(Y \mid D = 1, Z = 1, Y > F^{-1}_{Y \mid D = 1, Z = 1} (1 - \alpha))}{\mathbb{E}[D Z] \alpha} \\
        &+ \frac{(F^{-1}_{Y \mid D = 1, Z = 1}(1 - \alpha) - \mathbb{E}[Y \mid D = 1, Z = 1, Y > F^{-1}_{Y \mid D = 1, Z = 1} (1 - \alpha)])^2 (1 - \alpha)}{\mathbb{E}[DZ] \alpha} \\
        &+ \left(\frac{F^{-1}_{Y \mid D = 1, Z = 1} (1 - \alpha) - \mathbb{E}[Y \mid D = 1, Z = 1, Y > F^{-1}_{Y \mid D = 1, Z = 1} (1 - \alpha)]}{\alpha}\right)^2 V_{\alpha_1}, \\
        V_{\alpha_1} =& \alpha^2 \left[\frac{1 - p_a / \alpha}{\mathbb{E}[Z] p_a / \alpha} + \frac{1 - p_a}{(1 - \mathbb{E}[Z]) p_a}\right], \\
        V_{C_1} =& \frac{V(Y \mid D = 1, Z = 0)}{\mathbb{E}[D (1 - Z)]}.
    \end{aligned}
\end{equation*}
The estimators of bounds on $\delta_{0n}$ satisfy $\sqrt{N} (\widehat{\operatorname{LB}}_2 - \operatorname{LB}_2) \stackrel{d}{\rightarrow} N\left(0, V_{\operatorname{LB}_2} + V_{C_2}\right)$, $\sqrt{N} (\widehat{\operatorname{UB}}_2 - \operatorname{UB}_2) \stackrel{d}{\rightarrow} N\left(0, V_{\operatorname{UB}_2} + V_{C_2}\right)$, where 
\begin{equation*}
    \begin{aligned}
        V_{\operatorname{LB}_2} =& \frac{V(Y \mid D = 0, Z = 0, Y > F^{-1}_{Y \mid D = 0, Z = 0} (1 - \gamma))}{\mathbb{E}[(1 - D)(1 - Z)] \gamma} \\
        &+ \frac{(F^{-1}_{Y \mid D = 0, Z = 0} (1 - \gamma) - \mathbb{E}[Y \mid D = 0, Z = 0, Y > F^{-1}_{Y \mid D = 0, Z = 0} (1 - \gamma)])^2 (1 - \gamma)}{\mathbb{E}[(1 - D)(1 - Z)] \gamma} \\
        &+ \left(\frac{F^{-1}_{Y \mid D = 0, Z = 0} (1 - \gamma) - \mathbb{E}[Y \mid D = 0, Z = 0, Y > F^{-1}_{Y \mid D = 0, Z = 0} (1 - \gamma)]}{\gamma}\right)^2 V_{\gamma_1}, \\
        V_{\operatorname{UB}_2} =& \frac{V(Y \mid D = 0, Z = 0, Y \leq F^{-1}_{Y \mid D = 0, Z = 0} (\gamma))}{\mathbb{E}[(1 - D)(1 - Z)] \gamma} \\
        &+ \frac{(F^{-1}_{Y \mid D = 0, Z = 0}(\gamma) - \mathbb{E}[Y \mid D = 0, Z = 0, Y \leq F^{-1}_{Y \mid D = 0, Z = 0} (\gamma)])^2 (1 - \gamma)}{\mathbb{E}[(1 - D)(1 - Z)] \gamma} \\
        &+ \left(\frac{F^{-1}_{Y \mid D = 0, Z = 0}(\gamma) - \mathbb{E}[Y \mid D = 0, Z = 0, Y \leq F^{-1}_{Y \mid D = 0, Z = 0} (\gamma)]}{\gamma}\right)^2 V_{\gamma_1}, \\
    \end{aligned}
\end{equation*}
\begin{equation*}
    \begin{aligned}
        V_{\gamma_1} =& \gamma^2 \left[\frac{1 - p_n / \gamma}{\mathbb{E}[(1 - Z)] p_n / \gamma} + \frac{1 - p_n}{\mathbb{E}[Z] p_n}\right], \\
        V_{C_2} =& \frac{V(Y \mid D = 0, Z = 1)}{\mathbb{E}[(1 - D) Z]}.
    \end{aligned}
\end{equation*}

Finally, the estimators of bounds on $\delta_{1c}+\theta_{0c}=\delta_{0c}+\theta_{1c}$ have the asymptotic distributions as $\sqrt{N} (\widehat{\operatorname{LB}}_3 - \operatorname{LB}_3) \stackrel{d}{\rightarrow} N\left(0, V_{\operatorname{LB}_3} + V_{\operatorname{LB}_4}\right)$, $\sqrt{N} (\widehat{\operatorname{UB}}_3 - \operatorname{UB}_3) \stackrel{d}{\rightarrow} N\left(0, V_{\operatorname{UB}_3} + V_{\operatorname{UB}_4}\right)$, where 
\begin{equation*}
    \begin{aligned}
        V_{\operatorname{LB}_3} =& \frac{V(Y \mid D = 1, Z = 1, Y \leq F^{-1}_{Y \mid D = 1, Z = 1} (1 - \alpha))}{\mathbb{E}[D Z] (1 - \alpha)} \\
        &+ \frac{(F^{-1}_{Y \mid D = 1, Z = 1}(1 - \alpha) - \mathbb{E}[Y \mid D = 1, Z = 1, Y \leq F^{-1}_{Y \mid D = 1, Z = 1} (1 - \alpha)])^2 \alpha}{\mathbb{E}[DZ] (1 - \alpha)} \\
        &+ \left[\frac{F^{-1}_{Y \mid D = 1, Z = 1}(1 - \alpha) - \mathbb{E}[Y \mid D = 1, Z = 1, Y \leq F^{-1}_{Y \mid D = 1, Z = 1} (1 - \alpha)]}{(1 - \alpha)}\right]^2 V_{\alpha_2}, \\
        V_{\operatorname{LB}_4} =& \frac{V(Y \mid D = 0, Z = 0, Y > F^{-1}_{Y \mid D = 0, Z = 0} (\gamma))}{\mathbb{E}[(1 - D)(1 - Z)] (1 - \gamma)} \\
        &+ \frac{(F^{-1}_{Y \mid D = 0, Z = 0} (\gamma) - \mathbb{E}[Y \mid D = 0, Z = 0, Y > F^{-1}_{Y \mid D = 0, Z = 0} (\gamma)])^2 \gamma}{\mathbb{E}[(1 - D)(1 - Z)] (1 - \gamma)} \\
        &+ \left[\frac{F^{-1}_{Y \mid D = 0, Z = 0} (\gamma) - \mathbb{E}[Y \mid D = 0, Z = 0, Y > F^{-1}_{Y \mid D = 0, Z = 0} (\gamma)]}{1 - \gamma}\right]^2 V_{\gamma_2}, \\
        V_{\operatorname{UB}_3} =& \frac{V(Y \mid D = 1, Z = 1, Y > F^{-1}_{Y \mid D = 1, Z = 1} (\alpha))}{\mathbb{E}[D Z] (1 - \alpha)} \\
        &+ \frac{(F^{-1}_{Y \mid D = 1, Z = 1}(\alpha) - \mathbb{E}[Y \mid D = 1, Z = 1, Y > F^{-1}_{Y \mid D = 1, Z = 1} (\alpha)])^2 \alpha}{\mathbb{E}[DZ] (1 - \alpha)} \\
        &+ \left(\frac{F^{-1}_{Y \mid D = 1, Z = 1} (\alpha) - \mathbb{E}[Y \mid D = 1, Z = 1, Y > F^{-1}_{Y \mid D = 1, Z = 1} (\alpha)]}{1 - \alpha}\right)^2 V_{\alpha_2}, \\
        V_{\operatorname{UB}_4} =& \frac{V(Y \mid D = 0, Z = 0, Y \leq F^{-1}_{Y \mid D = 0, Z = 0} (1 - \gamma))}{\mathbb{E}[(1 - D)(1 - Z)] (1 - \gamma)} \\
        &+ \frac{(F^{-1}_{Y \mid D = 0, Z = 0}(1 - \gamma) - \mathbb{E}[Y \mid D = 0, Z = 0, Y \leq F^{-1}_{Y \mid D = 0, Z = 0} (1 - \gamma)])^2 \gamma}{\mathbb{E}[(1 - D)(1 - Z)] (1 - \gamma)} \\
        &+ \left[\frac{F^{-1}_{Y \mid D = 0, Z = 0}(1 - \gamma) - \mathbb{E}[Y \mid D = 0, Z = 0, Y \leq F^{-1}_{Y \mid D = 0, Z = 0} (1 - \gamma)]}{1 - \gamma}\right]^2 V_{\gamma_2}, \\
        V_{\alpha_2} =& (1 - \alpha)^2 \left[\frac{1 - p_a / (1 - \alpha)}{\mathbb{E}[Z] p_a / (1 - \alpha)} + \frac{1 - p_a}{(1 - \mathbb{E}[Z]) p_a}\right], \\
        V_{\gamma_2} =& (1 - \gamma)^2 \left[\frac{1 - p_n / (1 - \gamma)}{\mathbb{E}[(1 - Z)] p_n / (1 - \gamma)} + \frac{1 - p_n}{\mathbb{E}[Z] p_n}\right]. 
    \end{aligned}
\end{equation*}
For all variances derived in the asymptotic distributions, we can estimate them as their sample analogs.

Since we care about the true value of the average controlled effects instead of their identified sets, we apply the method proposed by \cite{imbens2004confidence} to construct the confidence intervals for identified bounds in Corollary \ref{cor:noer2}. The 95\% confidence intervals are constructed as $[\widehat{LB}_i - \bar{C}_i \hat{\sigma}_{LB_i} / \sqrt{N}, \widehat{UB}_i + \bar{C}_i \hat{\sigma}_{UB_i} / \sqrt{N}]$, $i = 1, 2, 3$, where $\hat{\sigma}_{LB_i}$ and $\hat{\sigma}_{UB_i}$ are estimated standard deviations of the corresponding lower bounds and upper bounds, and $\bar{C}_i$ satisfies 
\begin{equation*}
    \Phi\left(\bar{C}_i +\sqrt{N} \frac{\widehat{UB}_i - \widehat{LB}_i}{\max \left(\hat{\sigma}_{LB_i}, \hat{\sigma}_{UB_i}\right)}\right) - \Phi\left(-\bar{C}_i\right)=0.95.
\end{equation*}

We revisit the empirical study in \cite{bursztyn2020misperceived} with imposing the assumptions of random assignment and monotonicity. Using the estimation and inference methods in this section, we estimate the identified sets and construct uniform confidence intervals for the parameters $\delta_{1a}, \delta_{0n}$, and $\delta_{1c} + \theta_{0c}$. The results are presented in Table \ref{tab:inf_bounds_no_er}.
\begin{table}
    \centering
    \caption{Estimated identified sets and confidence sets under $A_3$}
      \begin{tabular}{lrrr}
      \toprule
        Parameter    & \multicolumn{1}{c}{Estimated Bounds}  & \multicolumn{1}{c}{Confidence Sets} \\
      \midrule
      $\delta_{1a}$ \quad \quad \quad \quad & [0.0557, 0.7155] \quad \quad \quad \quad &[-0.0834, 1.0184] \\
      $\delta_{0n}$ \quad \quad \quad \quad & [0.0829, 0.2505] \quad \quad \quad \quad & [-0.0135, 0.3445] \\
      $\delta_{1c}+\theta_{0c}$ \quad \quad \quad \quad& [-0.7140, 1.9242] \quad \quad \quad \quad & [-0.9048, 2.4797] \\
      \bottomrule
      \end{tabular}%
    \label{tab:inf_bounds_no_er}%
\end{table}%

\end{document}